\documentclass[preprint]{imsart}

\RequirePackage[OT1]{fontenc}
\RequirePackage{amsthm,amsmath,amssymb}
\RequirePackage{natbib}
\RequirePackage[colorlinks,citecolor=blue,urlcolor=blue]{hyperref}
\usepackage{graphicx}
\usepackage{amsthm}
\usepackage{pgf,tikz}
\usepackage{url}
\usepackage{thmtools}
\usetikzlibrary{arrows}
\usepackage{pdfcomment}
\usepackage[final]{pdfpages}

\usepackage{diagbox}
\usepackage{multirow}
\usepackage{booktabs}
\usepackage{adjustbox}
\usepackage{hhline}

\numberwithin{equation}{section}
\theoremstyle{plain}
\declaretheorem[name=Theorem,numberwithin=section]{Th}
\declaretheorem[name=Lemma,numberwithin=section]{Lem}

\theoremstyle{definition}
\newtheorem{Def}{Definition}[section]
\newtheorem{Not}{Notation}[section]

\newtheorem{Ass}{Assumption}[section]

\theoremstyle{remark}

 \newtheorem{Rem}{Remark}[section]

\DeclareMathOperator{\cov}{Cov}                   
\DeclareMathOperator{\var}{Var}                   
\DeclareMathOperator{\E}{E}                       

\renewcommand{\P}{\mathit{P}}
\renewcommand{\i}{{\,\text{\it \@i}\,}}

\newcommand{\rX}{{\color{red}{\X}}}

\newcommand{\rY}{{\color{red}{\Y}}}

\newcommand{\bX}{{\color{blue}{\X}}}
\newcommand{\bY}{{\color{blue}{\Y}}}

\newcommand{\Tr}{{{\rm Trace}}}
\newcommand{\Normal}{\mathbf{N}}

\newcommand{\I}{{\rm I}}

\newcommand{\X}{{\mathbf{X}}}
\newcommand{\Y}{{\mathbf{Y}}}

\newcommand{\1}{{\mathbf{1}}}

\renewcommand{\c}{{\textbf{\textit{c}}}}

\newcommand{\x}{{\textbf{\textit{x}}}}
\newcommand{\y}{{\textbf{\textit{y}}}}

\makeatother

\begin{document}

\begin{frontmatter}
\title{Statistical applications of random matrix theory:\\  comparison of two populations III\thanksref{T1}}
\runtitle{Comparison of two populations}
\thankstext{T1}{This paper is based on the PhD thesis of R\'emy Mari\'etan and is the final part III.}

\begin{aug}
\author{\fnms{R\'emy} \snm{Mari\'etan}\thanksref{t1}\ead[label=e1]{remy.marietan@epfl.alumni.ch}}
\and
\author{\fnms{Stephan} \snm{Morgenthaler}\thanksref{t2}\ead[label=e2]{stephan.morgenthaler@epfl.ch}}

\thankstext{t1}{PhD Student at EPFL, Member of Staff at the Swiss Federal Statistics Office}
\thankstext{t2}{Professor at EPFL in the Mathematics department}
\runauthor{R\'emy Mari\'etan and Stephan Morgenthaler}

\affiliation{\'Ecole Polytechnique F\'ed\'eral de Lausanne, EPFL}

\address{Department of Mathematics\\
\'Ecole Polytechnique F\'ed\'eral de Lausanne\\1015 Lausanne\\
\phantom{E-mail:\ }\printead{e1},\printead*{e2}}
\end{aug}

\begin{abstract}
This paper investigates a statistical procedure for testing the equality of two independently estimated covariance matrices when the number of potentially dependent data vectors is large and proportional to the size of the vectors, that is, the number of variables.
Inspired by the spike models used in random matrix theory, we concentrate on the largest eigenvalues of the matrices in order to determine significant differences. To avoid false rejections we must guard against residual spikes and need a sufficiently precise description of the properties of the largest eigenvalues under the null hypothesis. 

In this paper, we extend \cite{mainarticle} for perturbation of order $1$ and \cite{mainarticle2} studying simpler statistic. The residual spike introduce in the first paper is investigated and leads to a statistic that results in a good test of equality of two populations. 

Simulations show that this new test does not rely on some hypotheses that were necessary for the proofs and in the second paper.
\end{abstract}



\end{frontmatter}

\section{Introduction} 

In the last two decades, random matrix theory (RMT) has produced numerous results that offer a better understanding of large random matrices. These advances have enabled interesting applications in communication theory and even though it can potentially contribute to many other data-rich domains such as brain imaging or genetic research, it has rarely been applied.
The main barrier to the adoption of RMT may be the lack of concrete statistical results from the probability side. The straightforward adaptation of classical multivariate theory to high dimensions can sometimes be achieved, but such  procedures are only valid under strict assumptions about the data such as normality or independence. Even minor differences between the model assumptions and the actual data lead to catastrophic results and such procedures also often do not have enough power.

This paper proposes a statistical procedure for testing the equality of two covariance matrices when the number of potentially dependent data vectors $n$ and the number of variables $m$ are large. RMT denotes the investigation of estimates of covariance matrices $\hat{\Sigma}$ or more precisely their eigenvalues and eigenvectors when both $n$ and $m$ tend to infinity with $\lim \frac{m}{n}=c>0$.
When $m$ is finite and $n$ tends to infinity the behaviour of the random matrix is well known and presented in the books of \cite{multi3}, \cite{multi} and \cite{multi2} (or its original version \cite{multi22}).
In the RMT case, the behaviour is more complex, but many results of interest are known. \cite{Alice}, \cite{Tao} and more recently \cite{bookrecent} contain comprehensive introductions to RMT and \cite{Appliedbook} covers the case of empirical (estimated) covariance matrices.

Although the existing theory builds a good intuition of the behaviour of these matrices, it does not provide enough of a basis to construct a test with good power, which is robust with respect to the assumptions.
Inspired by the spike models, we extend the residual spikes introduced in \cite{mainarticle} and provide a description of the behaviour of this statistic under a null hypothesis when the perturbation is of order $k$. These results enable the user to test the equality of two populations as well as other null hypotheses such as the independence of two sets of variables.
This paper can be seen as a complex particular case of \cite{mainarticle2}. However simulations show that equality between eigenvalues of the perturbation are not necessary for this complex statistic and moreover they show good robustness against perturbations of distributions.\\

The remainder of the paper is organized as follows. In the next section, we develop the test statistic and discuss the problems associated with high dimensions. Then we present the main theorem \ref{TH=Main}. The proof itself is technical and presented in Appendix \ref{appendixproof}. The last section contains an example of an application. 

\section{Test statistic}\label{section:teststat}

We compare the spectral properties of two covariance estimators $\hat{\Sigma}_X$ and $\hat{\Sigma}_Y$ of dimension $m\times m$ which can be represented as
\begin{eqnarray*}
\hat{\Sigma}_X=P_X^{1/2} W_X P_X^{1/2} \text{ and } \hat{\Sigma}_Y=P_Y^{1/2} W_Y P_Y^{1/2}.
\end{eqnarray*}
In this equation, $W_X$ and $W_Y$ are of the form
\begin{eqnarray*}
W_X=O_X \Lambda_X O_X \text{ and } W_Y=O_Y \Lambda_Y O_Y,
\end{eqnarray*}
with $O_X$ and $O_Y$ being independent unit orthonormal random matrices whose distributions are invariant under rotations, while  $\Lambda_X$ and $\Lambda_Y$ are independent positive random diagonal matrices, independent of  $O_X, O_Y$ with trace equal to m and a bound on the diagonal elements. Note that the usual RMT assumption, $\frac{m}{n}=c$ is replaced by this bound! The (multiplicative) spike model of order $k$ determines the form of 
$P_X= \I_m + \sum_{s=1}^k (\theta_{X,s}-1) u_{X,s} u_{X,s}^t$ and $P_Y= \I_m+ \sum_{s=1}^k(\theta_{Y,s}-1) u_{Y,s} u_{Y,s}^t$ where $\left\langle u_{X,s} ,u_{X,r}  \right\rangle=\left\langle u_{Y,s} ,u_{Y,r}  \right\rangle=\delta_{s,r}$. 

Our results will apply to any two centered data matrices $\X \in \mathbb{R}^{m\times n_X}$ and $\Y \in \mathbb{R}^{m\times n_Y}$ which are such that 
\begin{eqnarray*}
\hat{\Sigma}_X= \frac{1}{n_X}\X \X^t \text{ and } \hat{\Sigma}_Y= \frac{1}{n_Y}\Y \Y^t
\end{eqnarray*}
and can be decomposed in the manner indicated. This is the basic assumption concerning the covariance matrices. We will assume throughout that $n_X\geq n_Y$. Because $O_X$ and $O_Y$ are independent and invariant by rotation we can assume without loss of generality that for $s=1,2,...,k$, $u_{X,s}=e_s$ as in \cite{deformedRMT}. Under the null hypothesis we have $P_X=P_Y$  and we use the simplified notation $P_k$ for both matrices where for $s=1,2,...,k$, $\theta_{X,s}=\theta_{Y,s}=\theta_s$ and $u_{X,s}=u_{Y,s}(=e_s)$. 

To test $H_0:P_k=P_X=P_Y$ against $H_1:P_X\neq P_Y$ it is natural to consider the extreme eigenvalues of 
\begin{eqnarray}
\hat{\Sigma}_X^{-1/2} \hat{\Sigma}_Y \hat{\Sigma}_X^{-1/2}\,.\label{eq:base1}
\end{eqnarray}
We could also swap the subscripts, but it turns our to be preferable to use the inversion on the matrix with larger sample size. 

The distributional approximations we will refer to are based on RMT, that is, they are derived by embedding a given data problem into a sequence of random matrices for which both $n$ and $m$ tend to infinity such that $m/n$ tends to a positive constant $c$. The most celebrated results of RMT describe the almost sure weak convergence of the empirical distribution of the eigenvalues (spectral distribution) to a non-random compactly supported limit law. An extension of this theory to the "Spike Model" suggests that we should modify $\hat{\Sigma}$ because estimates of isolated eigenvalues derived from the usual estimates are asymptotically biased. The following corrections will be used. 
\begin{Def} \label{Def=unbiased} 
Suppose $\hat{\Sigma}$ is of the form described at the start of the section. 
The \textbf{unbiased estimator of }$\theta_s$ for $s=1,...,k$ is defined as 
\begin{equation} \hat{\hat{\theta}}_s=1+\frac{1}{\frac{1}{m-k} \sum_{i=k+1}^{m} \frac{\hat{\lambda}_{\hat{\Sigma},i}}{\hat{\theta}_s-\hat{\lambda}_{\hat{\Sigma},i}}}\,,\label{eq:corrlambda}
\end{equation}
where $\hat{\lambda}_{\hat{\Sigma},i}$ is the $i^{\text{th}}$ largest eigenvalue of $\hat{\Sigma}$. When $\hat{\Sigma}=P_k^{1/2}WP_k^{1/2}$ as above, it is asymptotically equivalent to replace $\frac{1}{m-k}\sum_{i=k+1}^{m} \frac{\hat{\lambda}_{\hat{\Sigma},i}}{\hat{\theta}-\hat{\lambda}_{\hat{\Sigma},i}}$ in the denominator by $\frac{1}{m}\sum_{i=1}^{m} \frac{\hat{\lambda}_{W,i}}{\hat{\theta}-\hat{\lambda}_W,i}$.\\
Suppose that $\hat{u}_s$ is the eigenvector corresponding to $\hat{\theta}_s$, then the \textbf{filtered estimated covariance matrix } is defined as 
\begin{equation}
\hat{\hat{\Sigma}}= \I_m+\sum_{s=1}^k (\hat{\hat{\theta}}_s-1) \hat{u}_s \hat{u}_s^t\,.\label{eq:corrSigma}
\end{equation}
\end{Def}

The matrix (\ref{eq:base1}) which serves as the basis for the test then becomes either
\begin{equation}
\hat{\hat{\Sigma}}_X^{-1/2} \hat{\hat{\Sigma}}_Y \hat{\hat{\Sigma}}_X^{-1/2} 
\text{ or }\hat{\hat{\Sigma}}_X^{-1/2} \hat{\Sigma}_Y \hat{\hat{\Sigma}}_X^{-1/2} \,.\label{eq:base2}
\end{equation}

In the particular case where $\X$ and $\Y$ have independent jointly normal columns vector with constant variance $P_k=P_X=P_Y$, the distribution of the spectrum of the second of the above matrices is approximately Marcenko-Pastur distributed (see \cite{deformed}). This follows because $\hat{\hat{\Sigma}}_X$ is a finite perturbation. However, because of the non-consistency of the eigenvectors presented in \cite{deformedRMT}, we may observe \textbf{residual spikes} in the spectra, as shown in Figure \ref{fig=spike}. Thus, even if the two random matrices are based on the same perturbation, we see some spikes outside the bulk. This observation is worse in the last plot because four spikes fall outside the bulk even if there is actually no difference! 
This poses a fundamental problem for our test, because we must be able to distinguish the spikes indicative of a true difference from the residual spikes.  These remarks lead to the following definition.
\begin{Def}
The \textbf{residual spikes} are the isolated eigenvalues of 
$$\hat{\hat{\Sigma}}_X^{-1/2} \hat{\hat{\Sigma}}_Y \hat{\hat{\Sigma}}_X^{-1/2}\text{ or of }
\hat{\hat{\Sigma}}_X^{-1/2} \hat{\Sigma}_Y \hat{\hat{\Sigma}}_X^{-1/2}$$
when $P_X=P_Y$ (under the null hypothesis).
The \textbf{residual zone}  is the interval where a residual spike can fall asymptotically. 
\end{Def}

\begin{figure}[htbp] 
\centering
\begin{tabular}{cc}
 \includegraphics[width=0.26\paperwidth]{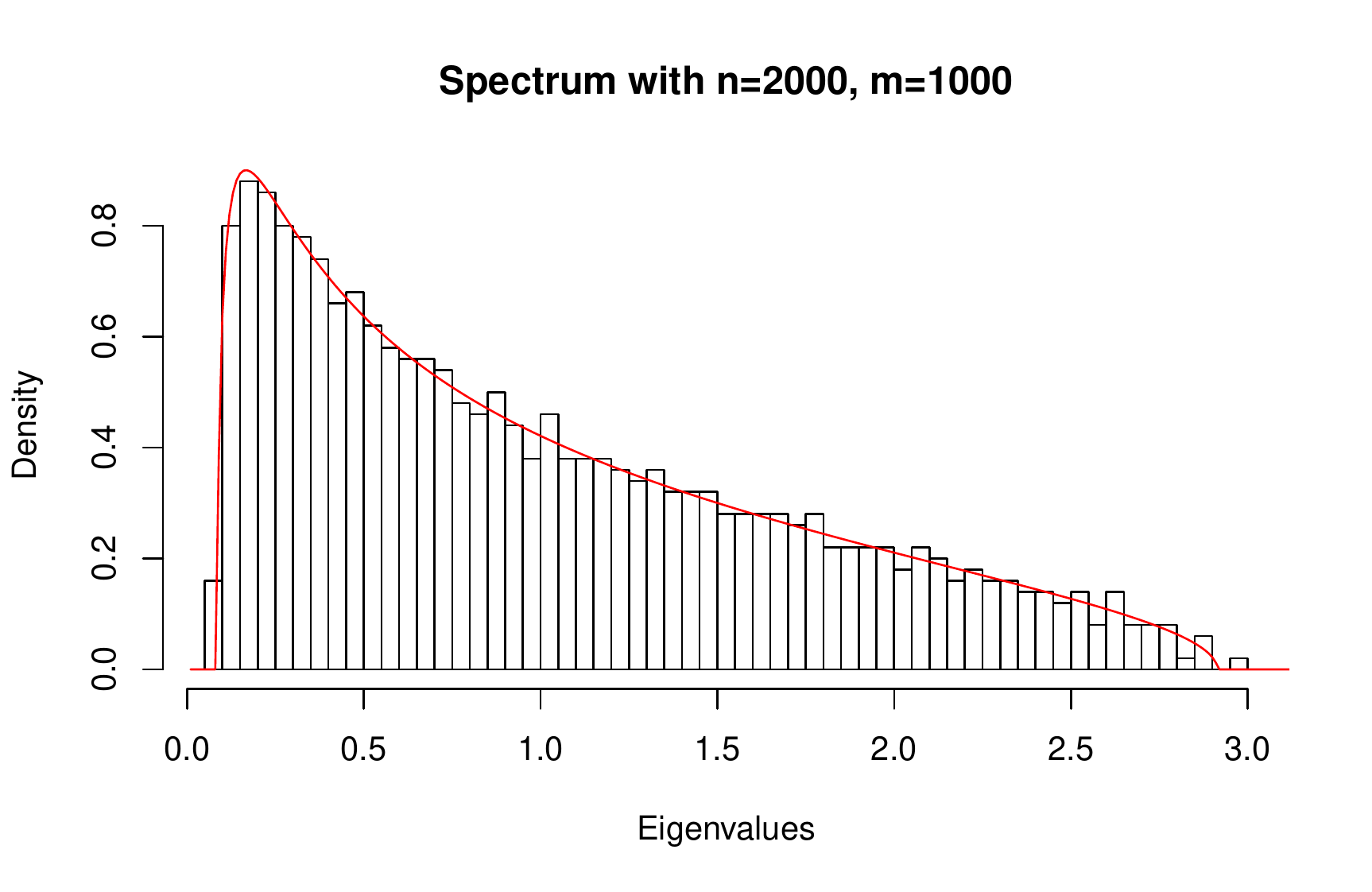} & 
 \includegraphics[width=0.26\paperwidth]{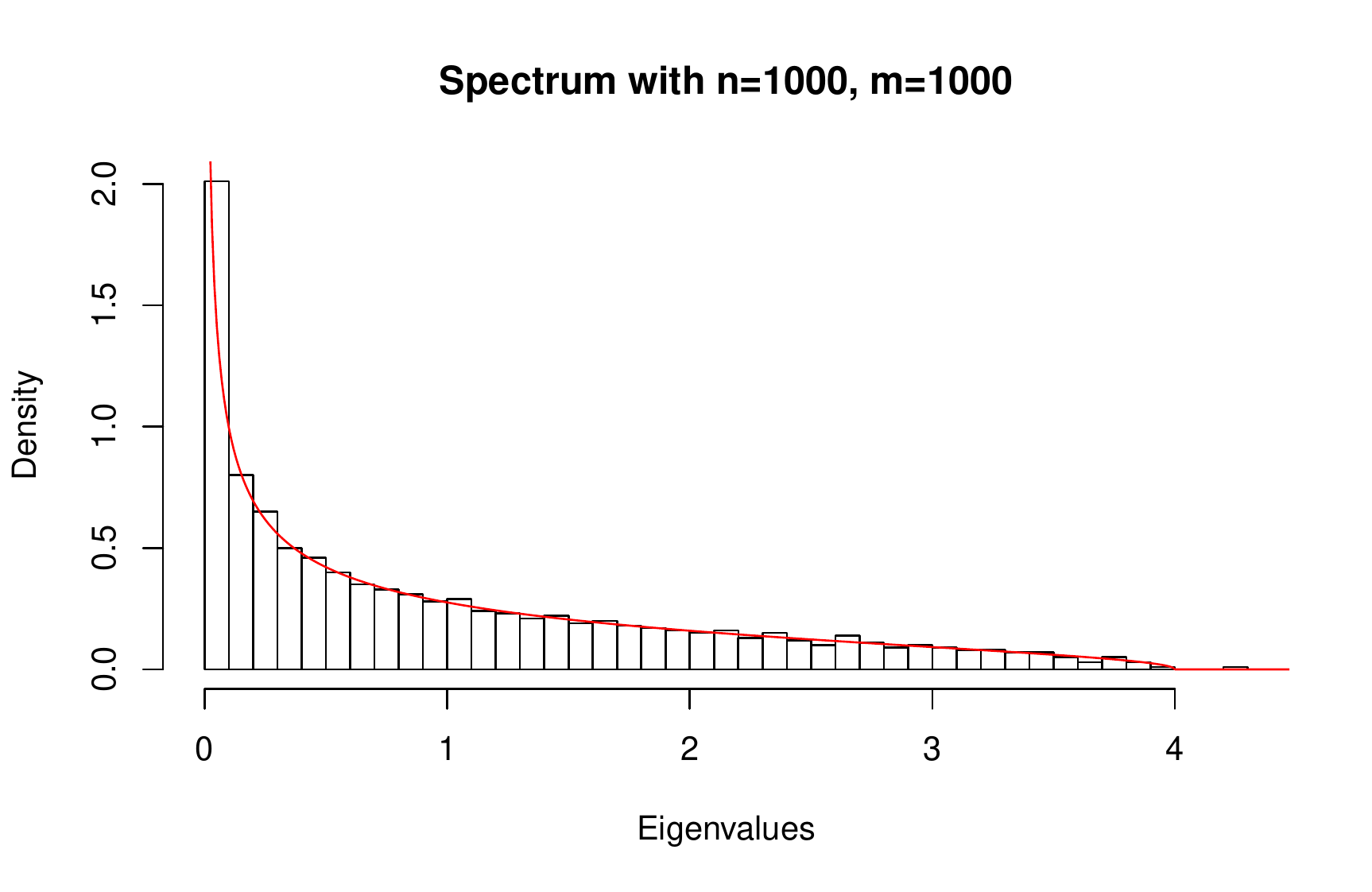} \\
 \includegraphics[width=0.26\paperwidth]{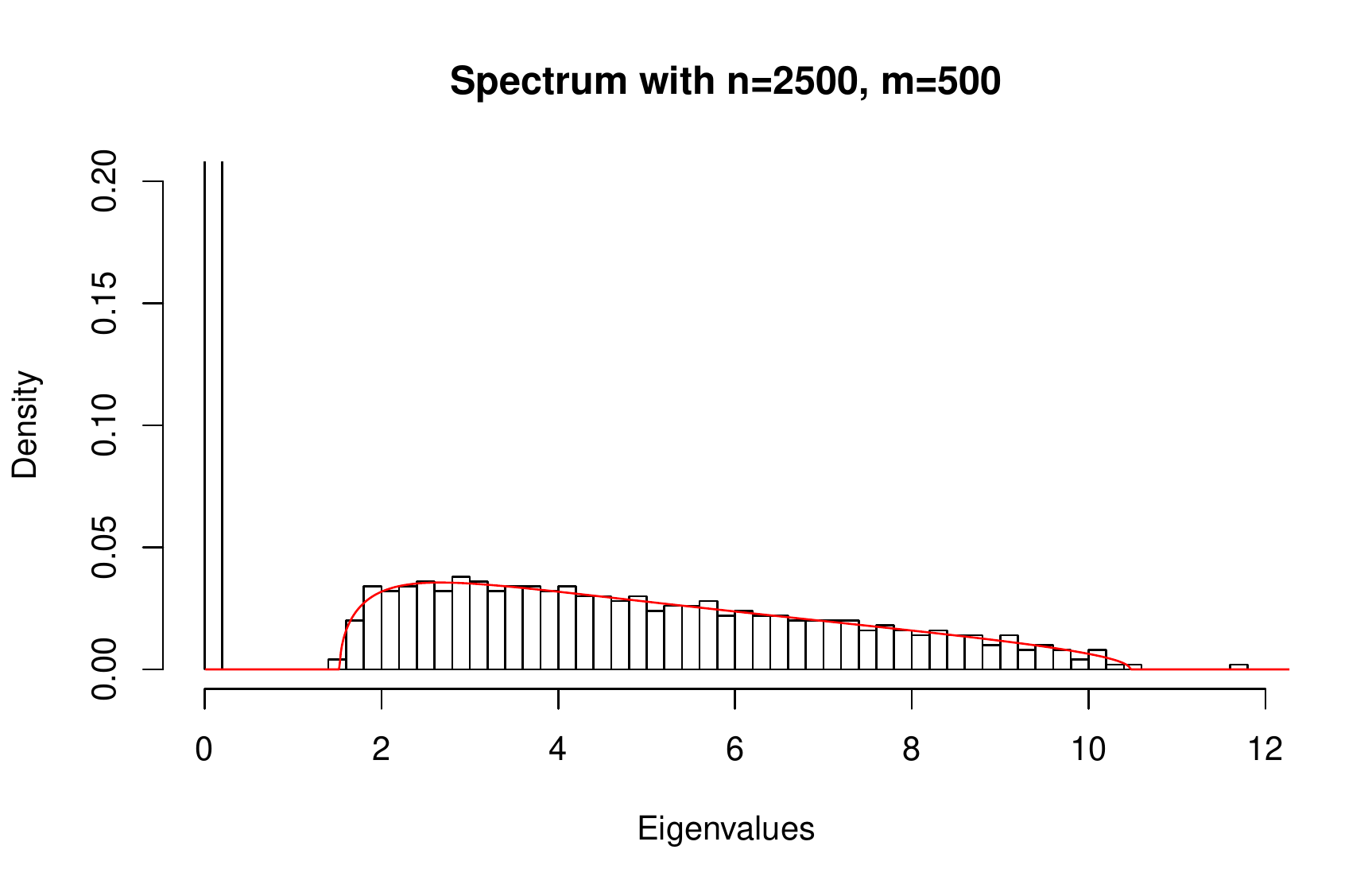}&
  \includegraphics[width=0.26\paperwidth]{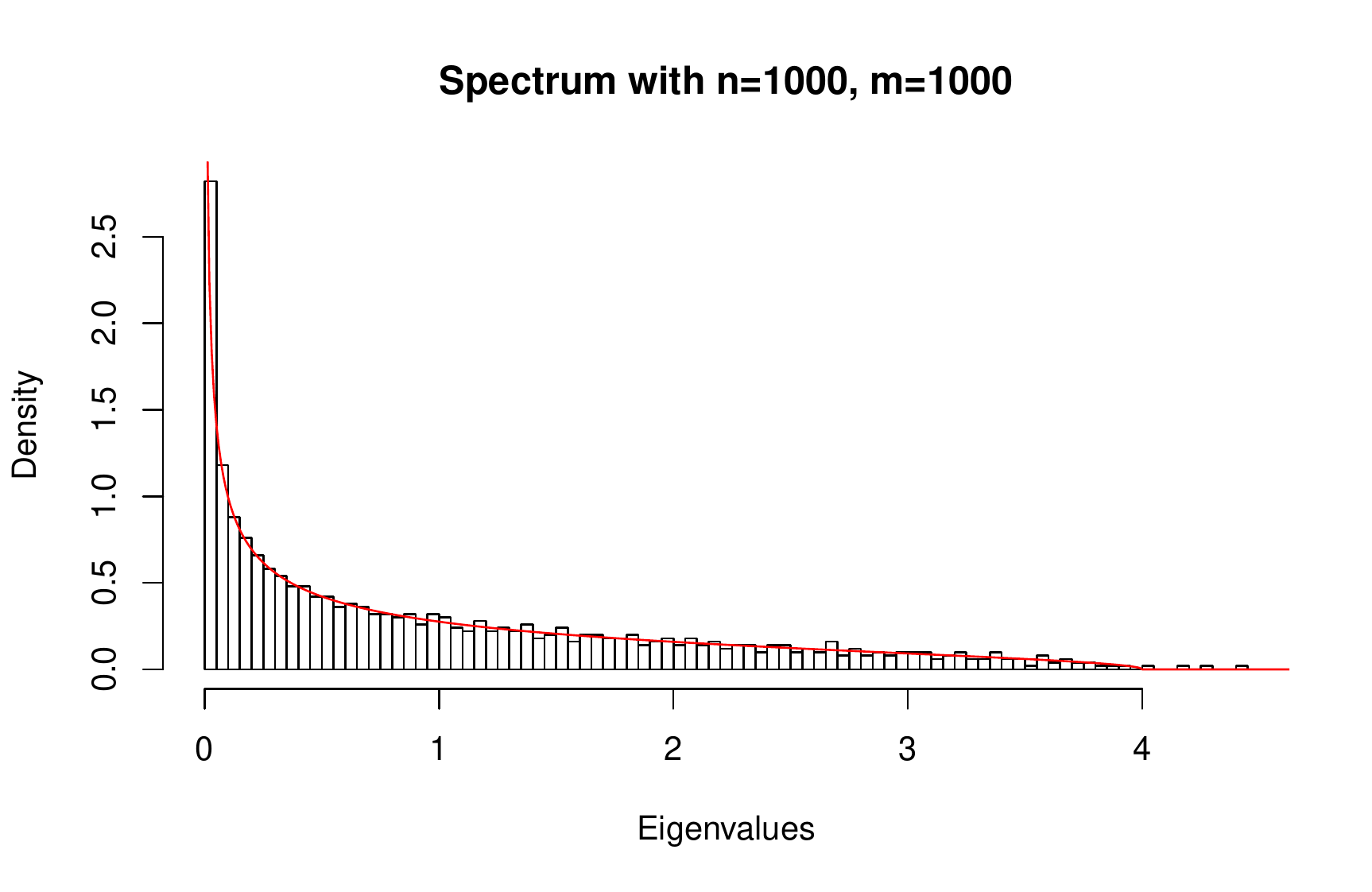}   
 \end{tabular}
\caption{Example of residual spikes of $\hat{\hat{\Sigma}}_X^{-1/2} \hat{\Sigma}_Y \hat{\hat{\Sigma}}_X^{-1/2}$ when $\theta=10$ for the first three figures and $\theta_{1,2,3,4}=10,15,20,25$ for the last figure.}\label{fig=spike}
\end{figure}

This paper studies these residual spikes by deriving the distribution of the extreme residual spikes under the null hypothesis. The philosophy is explained in Figure \ref{fig=residualzone0} with illustrations inspired by the i.i.d. normal case. All the eigenvalues inside the residual zone are potentially not indicative of real differences. However, when an eigenvalue is larger, we declare that this spike expresses a true difference. 

Most of our plots feature the seemingly more natural matrix
\begin{eqnarray*}
\hat{\hat{\Sigma}}_X^{-1/2} \hat{\Sigma}_Y \hat{\hat{\Sigma}}_X^{-1/2}\,.
\end{eqnarray*}
But, although this choice simplifies the study in terms of convergence in probability when the perturbation is of order $1$, this is no longer the case in more complex situations. In addition, the eigenvectors associated with the residual spikes are more accessible for the matrix in which all estimates are filtered.

\begin{figure}[htbp] 
\centering
\scalebox{0.7}{
\begin{tabular}{c}
\definecolor{wwqqzz}{rgb}{0.4,0.,0.6}
\definecolor{uuuuuu}{rgb}{0.26666666666666666,0.26666666666666666,0.26666666666666666}
\definecolor{ttqqqq}{rgb}{0.2,0.,0.}
\definecolor{ffqqqq}{rgb}{1.,0.,0.}
\begin{tikzpicture}[line cap=round,line join=round,>=triangle 45,x=3.0cm,y=4.0cm]
\draw[->,color=black] (-0.1260631217517215,0.) -- (5.068082604362673,0.);
\foreach \x in {,0.2,0.4,0.6,0.8,1.,1.2,1.4,1.6,1.8,2.,2.2,2.4,2.6,2.8,3.,3.2,3.4,3.6,3.8,4.,4.2,4.4,4.6,4.8,5.}
\draw[shift={(\x,0)},color=black] (0pt,-2pt);
\draw[->,color=black] (0.,-0.3711350758839508) -- (0.,0.9716909818012385);
\foreach \y in {-0.3,-0.2,-0.1,0.2,0.3,0.4,0.5,0.6,0.7,0.8,0.9}
\draw[shift={(0,\y)},color=black] (0pt,-2pt) -- (0pt,-2pt);
\clip(-0.1260631217517215,-0.3711350758839508) rectangle (5.068082604362673,0.9716909818012385);
\draw[line width=1.2pt,color=ffqqqq,smooth,samples=100,domain=0.0857904999793061:2.914212816177216] plot(\x,{sqrt(((1.0+sqrt(0.5))^(2.0)-(\x))*((\x)-(1.0-sqrt(0.5))^(2.0)))/2.0/3.1415926535/0.5/(\x)});
\draw [color=ffqqqq](0.42792951474866847,0.866057465940425) node[anchor=north west] {Marcenko-Pastur};
\draw (1.484249715766261,-0.04943300485329163) node[anchor=north west] {$1+\sqrt{c}$};
\draw (3,0.5) node[anchor=north west] {$\hat{\hat{\Sigma}}_X^{-1/2} \hat{\Sigma}_Y \hat{\hat{\Sigma}}_X^{-1/2} $};
\draw (2.644287988169762,-0.015822340715760077) node[anchor=north west] {$\left(1+\sqrt{c}\right)^2$};
\draw (0.8159438668999814,-0.06063655956580215) node[anchor=north west] {$1$};
\draw (3.5176468364387627,-0.017422848531833006) node[anchor=north west] {$T^{-1}\left(\frac{1}{\lambda-1}\right)$};
\draw (3.52508297921058,-0.16306905979446976) node[anchor=north west] {$\lambda=\frac{1}{2}\left( 2+c+\sqrt{c^2+4c} \right)$};
\draw [line width=5.2pt,color=wwqqzz] (2.914213562373095,0.)-- (3.722220771172429,0.0026511965675838013);
\draw [color=wwqqzz](3.02686141349882,0.19224367537372097) node[anchor=north west] {Residual zone};
\begin{scriptsize}
\draw [color=ttqqqq] (2.914213562373095,0.)-- ++(-4.0pt,0 pt) -- ++(8.0pt,0 pt) ++(-4.0pt,-4.0pt) -- ++(0 pt,8.0pt);
\draw [color=black] (1.69813,0.)-- ++(-4.0pt,0 pt) -- ++(8.0pt,0 pt) ++(-4.0pt,-4.0pt) -- ++(0 pt,8.0pt);
\draw [color=uuuuuu] (0.89908,0.)-- ++(-4.0pt,0 pt) -- ++(8.0pt,0 pt) ++(-4.0pt,-4.0pt) -- ++(0 pt,8.0pt);
\draw [color=black] (3.722220771172429,0.0026511965675838013)-- ++(-4.0pt,0 pt) -- ++(8.0pt,0 pt) ++(-4.0pt,-4.0pt) -- ++(0 pt,8.0pt);
\end{scriptsize}
\end{tikzpicture} \\ \  \definecolor{wwqqzz}{rgb}{0.4,0.,0.6}
\definecolor{uuuuuu}{rgb}{0.26666666666666666,0.26666666666666666,0.26666666666666666}
\begin{tikzpicture}[line cap=round,line join=round,>=triangle 45,x=2.98cm,y=8.0cm]
\draw[->,color=black] (-0.1260631217517215,0.) -- (5.068082604362673,0.);
\foreach \x in {,0.2,0.4,0.6,0.8,1.,1.2,1.4,1.6,1.8,2.,2.2,2.4,2.6,2.8,3.,3.2,3.4,3.6,3.8,4.,4.2,4.4,4.6,4.8,5.}
\draw[shift={(\x,0)},color=black] (0pt,-2pt);
\draw[->,color=black] (0.,-0.1710715988748344) -- (0.,0.16690661773675);
\foreach \y in {}
\draw[shift={(0,\y)},color=black] (0pt,-2pt) -- (0pt,-2pt);
\clip(-0.1260631217517215,-0.1710715988748344) rectangle (5.068082604362673,0.16690661773675);
\draw (1.584249715766261,-0.04943300485329163) node[anchor=north west] {$1+\sqrt{c}$};
\draw (0.8629438668999814,-0.06063655956580215) node[anchor=north west] {$1$};
\draw [color=wwqqzz](2.052726710390752,0.13302488617902253) node[anchor=north west] {Residual zone};
\draw (2.5881289899615982,-0.04623198922114577) node[anchor=north west] {$\lambda=1+c+\sqrt{c^2+2c} $};
\draw (3,0.15) node[anchor=north west] {$\hat{\hat{\Sigma}}_X^{-1/2} \hat{\hat{\Sigma}}_Y \hat{\hat{\Sigma}}_X^{-1/2} $};
\draw [line width=5.2pt,color=wwqqzz] (1.69813,0.)-- (2.9859626282538247,1.827374449692361E-4);
\begin{scriptsize}
\draw [color=black] (1.69813,0.)-- ++(-4.0pt,0 pt) -- ++(8.0pt,0 pt) ++(-4.0pt,-4.0pt) -- ++(0 pt,8.0pt);
\draw [color=uuuuuu] (0.89908,0.)-- ++(-4.0pt,0 pt) -- ++(8.0pt,0 pt) ++(-4.0pt,-4.0pt) -- ++(0 pt,8.0pt);
\draw [color=black] (2.9859626282538247,1.827374449692361E-4)-- ++(-4.0pt,0 pt) -- ++(8.0pt,0 pt) ++(-4.0pt,-4.0pt) -- ++(0 pt,8.0pt);
\end{scriptsize}
\end{tikzpicture}  
\end{tabular} 
}
\caption{Residual zone of $\hat{\hat{\Sigma}}_X^{-1/2} \hat{\Sigma}_Y \hat{\hat{\Sigma}}_X^{-1/2}$ and $\hat{\hat{\Sigma}}_X^{-1/2} \hat{\hat{\Sigma}}_Y \hat{\hat{\Sigma}}_X^{-1/2}$.}\label{fig=residualzone0}
\end{figure}
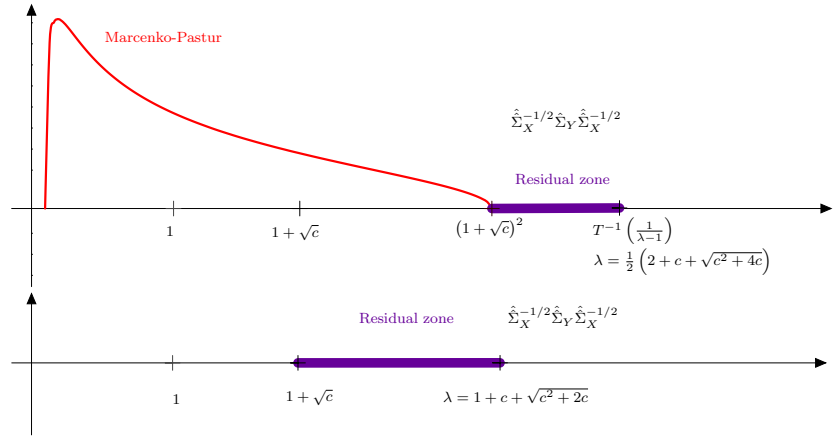


Let $\hat\theta_{X,s}$ and $\hat\theta_{Y,s}$ be isolated eigenvalues and construct the asymptotic unbiased estimators as in Equation (\ref{eq:corrlambda})

\begin{eqnarray*}
\hat{\hat{\theta}}_{X,s}=1+\frac{1}{\frac{1}{m-k} \sum_{i=k+1}^{m} \frac{\hat{\lambda}_{\hat{\Sigma}_{X},i}}{\hat{\theta}_{X,s}-\hat{\lambda}_{\hat{\Sigma}_{X},i}}} \text{ and } \hat{\hat{\theta}}_{Y,s}=1+\frac{1}{\frac{1}{m-k} \sum_{i=k+}^m \frac{\hat{\lambda}_{\hat{\Sigma}_{Y},i}}{\hat{\theta}_{Y,s}-\hat{\lambda}_{\hat{\Sigma}_{Y},i}}},
\end{eqnarray*}
where $\hat{\lambda}_{\hat{\Sigma}_{X},i}$ and $\hat{\lambda}_{\hat{\Sigma}_{Y},i}$ are the i$^{th}$ ordered eigenvalue of $\hat{\Sigma}_{X}$ and $\hat{\Sigma}_{Y}$, respectively. The test statistic is then 
\begin{eqnarray*}
 \lambda_{\min}\left(\hat{\hat{\Sigma}}_X^{-1/2} \hat{\hat{\Sigma}}_Y \hat{\hat{\Sigma}}_X^{-1/2} \right) \text{ and } \lambda_{\max}\left(\hat{\hat{\Sigma}}_X^{-1/2} \hat{\hat{\Sigma}}_Y \hat{\hat{\Sigma}}_X^{-1/2} \right)\,,
\end{eqnarray*}
where the filtered matrices are constructed as in (\ref{eq:corrSigma}).
These two statistics provide a basis for a powerful and robust test for the equality of (detectable) perturbations $P_X$ and $P_Y$.

\subsection{Null distribution} \label{sec=test}
Under $H_0$, $\lambda_{\max}\left(\hat{\hat{\Sigma}}_X^{-1/2} \hat{\hat{\Sigma}}_Y \hat{\hat{\Sigma}}_X^{-1/2} \right) $ is obviously a function of $\theta_s=\theta_{X,s}=\theta_{Y,s}$ for $s=1,2,...,k$. The suspected \textbf{worst case} occurs in the limit as $\theta_s\to\infty$ for all $s$ and it is this limit which will determine the critical values of the test. A criterion proposed in \cite{mainarticle} allows to check if this scenario is really the worst case. Let 
\begin{eqnarray*}
&&\lambda_{\max}\left(\hat{\hat{\Sigma}}_X^{-1/2} \hat{\hat{\Sigma}}_Y \hat{\hat{\Sigma}}_X^{-1/2} \right) \leqslant \lim_{\theta \rightarrow \infty} \lambda_{\max}\left(\hat{\hat{\Sigma}}_X^{-1/2} \hat{\hat{\Sigma}}_Y \hat{\hat{\Sigma}}_X^{-1/2} \right) =V_{\max} ,\\
&&\lambda_{\min}\left(\hat{\hat{\Sigma}}_X^{-1/2} \hat{\hat{\Sigma}}_Y \hat{\hat{\Sigma}}_X^{-1/2} \right) \geqslant \lim_{\theta \rightarrow \infty} \lambda_{\min}\left(\hat{\hat{\Sigma}}_X^{-1/2} \hat{\hat{\Sigma}}_Y \hat{\hat{\Sigma}}_X^{-1/2} \right) =V_{\min}.
\end{eqnarray*}
Because of our focus on the worst case scenario under $H_0$, we will investigate the asymptotic as $\theta_i=\theta p_i$ for fixed $p_i>0$ and $\frac{\theta}{\sqrt{m}} \rightarrow \infty$. We recall that \cite{mainarticle2} also allows some finite $\theta_s$ but it seems intuitive that this scenario will not create a worst case in most situations. This intuition is highlighted by \cite{mainarticle} showing by simulation that the residual spike increase as a function of $\theta$ assuming $W_X=\frac{1}{n_X}\X\X^t$.
Our test rejects the null hypothesis of equal populations if 
either $\P\left( V_{\max} > \hat{\lambda}_{\max}   \right)$ or $\P\left( V_{\min} < \hat{\lambda}_{\min}   \right)$ is small, where $\hat{\lambda}_{\max}$ and $\hat{\lambda}_{\min}$ are the observed extreme residual spikes.

The following result describes the asymptotic behavior of the extreme eigenvalues and thus of $V_{\max}$ and $V_{\min}$. 
\begin{Th}
 \label{TH=Main} 
Suppose $W_X,W_Y \in \mathbb{R}^{m\times m}$ are as described at the start of Section 2 and 
\begin{enumerate}
\item \cite{mainarticle} have already investigated  the case 
$P_1= \I_m+(\theta-1)e_1 e_1^t \in  \mathbb{R}^{m\times m}$ with $\frac{\sqrt{m}}{\theta}=o(1)$ with regard to large $m$. Let  
\begin{eqnarray*}
\hat{\Sigma}_{X,P_1}=P_1^{1/2} W_X P_1^{1/2} \text{ and } \hat{\Sigma}_{Y,P_1}=P_1^{1/2} W_Y P_1^{1/2}.
\end{eqnarray*}
and $\hat{\hat{\Sigma}}_{X,P_1}$, $\hat{\hat{\Sigma}}_{Y,P_1}$ as described above (see, \ref{Def=unbiased}). 

Then, conditional on the spectra $S_{W_X}=\left\lbrace \hat{\lambda}_{W_X,1},\hat{\lambda}_{W_X,2},...,\hat{\lambda}_{W_X,m} \right\rbrace$ and $S_{W_Y}=\left\lbrace \hat{\lambda}_{W_Y,1},\hat{\lambda}_{W_Y,2},...,\hat{\lambda}_{W_Y,m} \right\rbrace$ of $W_X$ and $W_Y$,
\begin{eqnarray*}
\left. \sqrt{m} \frac{\left(\lambda_{\max}\left(\hat{\hat{\Sigma}}_{X,P_1}^{-1/2} \hat{\hat{\Sigma}}_{Y,P_1} \hat{\hat{\Sigma}}_{X,P_1}^{-1/2} \right) -\lambda^+ \right)}{\sigma^+} \right| S_{W_X}, S_{W_Y} \sim \Normal(0,1)+o_{p}(1), 
\end{eqnarray*}
where
\begin{eqnarray*}
&&\lambda^+= \sqrt{M_2^2-1}+M_2, \hspace{30cm}
\end{eqnarray*}
\scalebox{0.77}{
\begin{minipage}{1\textwidth}
\begin{eqnarray*}
&&{\sigma^+}^2=\frac{1}{\left(M_{2,X}+M_{2,Y}-2\right) \left(M_{2,X}+M_{2,Y}+2\right)} \hspace{30cm} \\
&&\hspace{2cm} \Bigg( 9 M_{2,X}^4 M_{2,Y}+4 M_{2,X}^3 M_{2,Y}^2+4 M_{2,X}^3 M_{2,Y}+2 M_{2,X}^3 M_{3,Y}-2 M_{2,X}^2 M_{2,Y}^3\\
&& \hspace{2cm}+4 M_{2,X}^2 M_{2,Y}^2-11 M_{2,X}^2 M_{2,Y}-8 M_{3,X} M_{2,X}^2 M_{2,Y}+2 M_{2,X}^2 M_{2,Y} M_{3,Y}\\
&& \hspace{2cm}-2 M_{2,X}^2 M_{3,Y}+M_{2,X}^2 M_{4,Y}+4 M_{2,X} M_{2,Y}^3+M_{2,X} M_{2,Y}^2+4 M_{2,X} M_{2,Y}\\
&& \hspace{2cm}-4 M_{3,X} M_{2,X} M_{2,Y}^2-4 M_{3,X} M_{2,X} M_{2,Y}-2 M_{2,X} M_{2,Y}^2 M_{3,Y}-4 M_{2,X} M_{2,Y} M_{3,Y}\\
&& \hspace{2cm}-6 M_{2,X} M_{3,Y}+2 M_{4,X} M_{2,X} M_{2,Y}+2 M_{2,X} M_{2,Y} M_{4,Y}-2 M_{3,X} M_{2,Y}^2\\
&& \hspace{2cm}+2 M_{3,X} M_{2,Y}+M_{4,X} M_{2,Y}^2+4 M_{2,X}^5+2 M_{2,X}^4-4 M_{3,X} M_{2,X}^3-13 M_{2,X}^3\\
&& \hspace{2cm}-2 M_{3,X} M_{2,X}^2+M_{4,X} M_{2,X}^2-2 M_{2,X}^2+10 M_{3,X} M_{2,X}+4 M_{2,X}+4 M_{3,X}\\
&& \hspace{2cm}-2 M_{4,X}+M_{2,Y}^5+2 M_{2,Y}^4-M_{2,Y}^3-2 M_{2,Y}^2+4 M_{2,Y}-2 M_{2,Y}^3 M_{3,Y}\\
&& \hspace{2cm}-2 M_{2,Y}^2 M_{3,Y}+2 M_{2,Y} M_{3,Y}+4 M_{3,Y}+M_{2,Y}^2 M_{4,Y}-2 M_{4,Y}-4 \Bigg)\\
&&\hspace{1.25cm} + \frac{1}{\sqrt{\left(M_{2,X}+M_{2,Y}-2\right) \left(M_{2,X}+M_{2,Y}+2\right)}}\\
&& \hspace{2cm} \Bigg( 5 M_{2,X}^3 M_{2,Y}-M_{2,X}^2 M_{2,Y}^2+2 M_{2,X}^2 M_{2,Y}+2 M_{2,X}^2 M_{3,Y}-M_{2,X} M_{2,Y}^3\\
&& \hspace{2cm}+2 M_{2,X} M_{2,Y}^2-4 M_{2,X} M_{2,Y}-4 M_{3,X} M_{2,X} M_{2,Y}-2 M_{2,X} M_{3,Y}+M_{2,X} M_{4,Y}\\
&& \hspace{2cm}-2 M_{3,X} M_{2,Y}+M_{4,X} M_{2,Y}+4 M_{2,X}^4+2 M_{2,X}^3-4 M_{3,X} M_{2,X}^2-5 M_{2,X}^2\\
&& \hspace{2cm}-2 M_{3,X} M_{2,X}+M_{4,X} M_{2,X}+2 M_{2,X}+2 M_{3,X}+M_{2,Y}^4+2 M_{2,Y}^3+M_{2,Y}^2\\
&& \hspace{2cm}+2 M_{2,Y}-2 M_{2,Y}^2 M_{3,Y}-2 M_{2,Y} M_{3,Y}-2 M_{3,Y}+M_{2,Y} M_{4,Y} \Bigg),
\end{eqnarray*}
\end{minipage}}
\begin{eqnarray*}
&&M_{s,X}= \frac{1}{m} \sum_{i=1}^m \hat{\lambda}_{W_X,i}^s,\hspace{30cm}\\
&&M_{s,Y}= \frac{1}{m} \sum_{i=1}^m \hat{\lambda}_{W_Y,i}^s,\\
&&M_s=\frac{M_{s,X}+M_{s,Y}}{2}.
\end{eqnarray*}
Moreover, 
\begin{eqnarray*}
\left.\sqrt{m} \frac{\left(\lambda_{\min}\left(\hat{\hat{\Sigma}}_{X,P_1}^{-1/2} \hat{\hat{\Sigma}}_{Y,P_1} \hat{\hat{\Sigma}}_{X,P_1}^{-1/2} \right) -\lambda^- \right)}{\sigma^-}\right| S_{W_X}, S_{W_Y} \sim \Normal(0,1)+o_{m}(1), 
\end{eqnarray*}
where 
\begin{eqnarray*}
&&\lambda^-= -\sqrt{M_2^2-1}+M_2, \hspace{30cm}\\
&&{\sigma^-}^2=\left(\lambda^-\right)^4 {\sigma^+}^2.
\end{eqnarray*}
The error $o_p(1)$ in the approximation is with regard to large values of $m$.

\item Suppose that $P_k= \I_m+\sum_{s=1}^k(\theta_s-1)e_s e_s^t \in  \mathbb{R}^{m\times m}$ with $\theta_s=p_s \theta$, $p_s>0$ and $\frac{\sqrt{m}}{\theta}=o(1)$ with regard to large $m$.
\begin{eqnarray*}
\hat{\Sigma}_{X,P_k}=P_k^{1/2} W_X P_k^{1/2} \text{ and } \hat{\Sigma}_{Y,P_k}=P_k^{1/2} W_Y P_k^{1/2},
\end{eqnarray*}
and $\hat{\hat{\Sigma}}_{X,P_k}$, $\hat{\hat{\Sigma}}_{Y,P_k}$ as described above (see, \ref{Def=unbiased}). 
\noindent Then, conditioning on the spectra $S_{W_X}$ and $S_{W_Y}$,\\
\scalebox{0.82}{
\begin{minipage}{1\textwidth}
\begin{eqnarray*}
\left. \lambda_{\max}\left(\hat{\hat{\Sigma}}_{X,P_k}^{-1/2} \hat{\hat{\Sigma}}_{Y,P_k} \hat{\hat{\Sigma}}_{X,P_k}^{-1/2}\right) \right| S_{W_X},S_{W_Y}=\lambda_{\max}\left( H^+ \right)+1+O_p\left(\frac{1}{m}\right),\\
\left.\lambda_{\min}\left(\hat{\hat{\Sigma}}_{X,P_k}^{-1/2} \hat{\hat{\Sigma}}_{Y,P_k} \hat{\hat{\Sigma}}_{X,P_k}^{-1/2}\right)\right| S_{W_X},S_{W_Y}=\lambda_{\max}\left( H^- \right)+1+O_p\left(\frac{1}{m}\right),
\end{eqnarray*}
\end{minipage}}

where
\begin{eqnarray*}
H^\pm= \zeta_{\infty}^\pm \begin{pmatrix}
\hat{\zeta}_1^\pm / \zeta_{\infty}^\pm &  w_{1,2}^\pm & w_{1,3}^\pm & \cdots &  w_{1,k}^\pm \\ 
 w_{2,1}^\pm  & \hat{\zeta}_2^\pm / \zeta_{\infty}^\pm &  w_{2,3}^\pm & \cdots  & w_{2,k}^\pm \\ 
 w_{3,1}^\pm & w_{3,2}^\pm & \hat{\zeta}_3^\pm /\zeta_{\infty}^\pm & \cdots  &  w_{3,k}^\pm \\ 
\vdots & \vdots & \ddots & \ddots &  \vdots \\ 
  w_{k,1}^\pm &  w_{k,2}^\pm & w_{k,3}^\pm & \cdots  & \hat{\zeta}_k^\pm /\zeta_{\infty}^\pm  \\ 
\end{pmatrix} ,
\end{eqnarray*}
and \\
\scalebox{0.73}{
\begin{minipage}{1\textwidth}
\begin{eqnarray*}
&&\hat{\zeta}_i^+= \left.\lambda_{\max}\left(\hat{\hat{\Sigma}}_{X,\tilde{P}_i}^{1/2} \hat{\hat{\Sigma}}_{Y,\tilde{P}_i} \hat{\hat{\Sigma}}_{X,\tilde{P}_i}^{1/2} \right)-1 \right| S_{W_X},S_{W_Y},\hspace{30cm}\\
&&\hat{\zeta}_i^-= \left.\lambda_{\min}\left(\hat{\hat{\Sigma}}_{X,\tilde{P}_i}^{1/2} \hat{\hat{\Sigma}}_{Y,\tilde{P}_i} \hat{\hat{\Sigma}}_{X,\tilde{P}_i}^{1/2} \right)-1 \right| S_{W_X},S_{W_Y},\\
&&\zeta_{\infty}^\pm =\underset{m \rightarrow \infty}{\lim} \hat{\zeta}_i^\pm = \lambda^\pm -1,\\
&&w_{i,j}^\pm \sim { \color{red} \Normal}\left(0,
\frac{1}{m}\frac{2 (M_{2,X}-1) (M_{2,Y}-1)+B_{X}^\pm +B_{Y}^\pm}{\left((\zeta_{\infty}^\pm-2 M_2+1)^2+2 (M_2-1)\right)^2}
  \right)+o_p\left( \frac{1}{\sqrt{m}} \right),
\end{eqnarray*}
\end{minipage}}\\
\scalebox{0.73}{
\begin{minipage}{1\textwidth}
\begin{eqnarray*}
&&B_{X}^+=\left(1-M_2+2 M_{2,X}+\sqrt{M_2^2-1}\right)^2 (M_{2,X}-1)\hspace{20cm}\\
&&\hspace{2cm} +2\left(-1+M_2-2M_{2,x}-\sqrt{M_2^2-1}\right)(M_{3,X}-M_{2,X})+(M_{4,X}-M_{2,X}^2),\\
&&B_{Y}^+=\left(1+M_2+M_{2,Y}-M_{2,X}-\sqrt{M_2^2-1}\right)^2 (M_{2,Y}-1)\\
&&\hspace{2cm} +2\left(-1-M_2-M_{2,Y}-M_{2,X}-\sqrt{M_2^2-1}\right)(M_{3,Y}-M_{2,Y})+(M_{4,Y}-M_{2,Y}^2),\\
&&B_{X}^-=\left(1-M_2+2 M_{2,X}-\sqrt{M_2^2-1}\right)^2 (M_{2,X}-1)\hspace{20cm}\\
&&\hspace{2cm} +2\left(-1+M_2-2M_{2,x}+\sqrt{M_2^2-1}\right)(M_{3,X}-M_{2,X})+(M_{4,X}-M_{2,X}^2),\\
&&B_{Y}^-=\left(1+M_2+M_{2,Y}-M_{2,X}+\sqrt{M_2^2-1}\right)^2 (M_{2,Y}-1)\\
&&\hspace{2cm} +2\left(-1-M_2-M_{2,Y}+M_{2,X}-\sqrt{M_2^2-1}\right)(M_{3,Y}-M_{2,Y})+(M_{4,Y}-M_{2,Y}^2).
\end{eqnarray*}
\end{minipage}}\\
\noindent The matrices $H^+$ and $H^-$ are strongly correlated. However, within a matrix, all the entries are uncorrelated.
\end{enumerate}
\end{Th}

\begin{Rem} \label{Rem=thMain} \
The entries of the matrices $H^+$ and $H^-$ are asymptotically uncorrelated Normal or a sum of two Normals. 
\end{Rem}

\paragraph*{Special case}
If the spectra are Marcenko-Pastur distributed, we define $c_X=m/n_X$ and $c_Y=m/n_Y$. Then,

\scalebox{0.79}{
\begin{minipage}{1\textwidth}
\begin{eqnarray*}
&&c=\frac{c_X+c_Y}{2},\hspace{30cm}\\
&&\lambda^+=  c+\sqrt{c (c+2)}+1,\\
&&{\sigma^+}^2= c_X^3+c_X^2 c_Y+3 c_X^2+4 c_X c_Y-c_X+c_Y^2+c_Y\\
&& \hspace{2cm}+\frac{ (8 c_X+2 c_X^2+\left(c_X^3+5 c_X^2+c_X^2 c_Y+4 c_X c_Y+5 c_X +3 c_Y+c_Y^2\right)\sqrt{c(c+2)}}{c+2},\\
&&w_{i,j}^+ \sim {\color{red} \Normal}\left(0, \frac{\sigma_w^2}{m} \right),\\
&& \sigma_w^2=\frac{2c_X \left(\sqrt{c(c+2)}+2\right)+2 \c_Y \left(-\sqrt{c(c+2)}+2\right)+c_X^2+c_Y^2}
{4c \left(-\sqrt{c \left(c+2\right)}+c+2\right)^2}.
\end{eqnarray*}
\end{minipage}}\\
(Proof Appendix \ref{appendixproof}. The {\color{red} red} character is not proven and is a sum of two asymptotic uncorrelated marginally normal random variables that are certainly independent.)

\subsection{Discussion and simulation}\label{sec:smallsimulationmainth}

The above theorem gives the limiting distribution of $V_{\max}$ and $V_{\min}$. In this subsection, we first check the quality of the approximations in Theorem \ref{TH=Main}. Then we investigate the worst case with regard to $\theta$. Finally we relax some assumption on $\theta_s$ and on the distribution. 

\subsubsection{Some simulations}

Assume $\mathbf{X} \in \mathbb{R}^{m\times n_X}$ and $\mathbf{Y} \in \mathbb{R}^{m\times n_Y}$ with $\mathbf{X}=\left(X_1,X_2,...,X_{n_X}\right)$ and $\mathbf{Y}=\left(Y_1,Y_2,...,Y_{n_Y}\right)$.

\scalebox{0.65}{
\begin{minipage}{1\textwidth}
\begin{eqnarray*}
&&X_i \sim \Normal_m\left(\vec{0},\sigma^2 \I_m \right) \text{ with } X_1 = \epsilon_{X,1}\text{ and } 
X_{i+1}=\rho X_i+\sqrt{1-\rho^2} \ \epsilon_{X,i+1}, \text{ where }  \epsilon_{X,i} \overset{i.i.d}{\sim} \Normal_m\left(\vec{0},\sigma^2 \I_m\right),\\
&&Y_i \sim \Normal_m\left(\vec{0},\sigma^2 \I_m \right) \text{ with } Y_1 = \epsilon_{Y,1}\text{ and } Y_{i+1}=\rho Y_i+\sqrt{1-\rho^2} \ \epsilon_{Y,i+1}, \text{ where }  \epsilon_{Y,i} \overset{i.i.d}{\sim} \Normal_m\left(\vec{0},\sigma^2 \I_m\right)
\end{eqnarray*}
\end{minipage}}

Let $P_X= \I_m + (\theta_{X}-1) u_{X} u_{X}^t$ and $P_Y= \I_m+  (\theta_{Y}-1) u_{Y} u_{Y}^t$ be two perturbations in  $\mathbb{R}^{m\times m}$. Then,
\begin{eqnarray*}
\mathbf{X}_P=P_X^{1/2} \mathbf{X} \text{ and  } \mathbf{Y}_P=P_Y^{1/2} \mathbf{Y},\\
\hat{\Sigma}_X=\frac{\mathbf{X}_P^t \mathbf{X}_P}{n_X} \text{ and } \hat{\Sigma}_Y=\frac{\mathbf{Y}_P^t \mathbf{Y}_P}{n_Y}.
\end{eqnarray*}
We assume a common and large value for $\theta$ and $P_X=P_Y$.

\begin{table}
 \begin{tabular}{ c c }
 \begin{minipage}{.45\textwidth} \centering
      \includegraphics[width=1\textwidth]{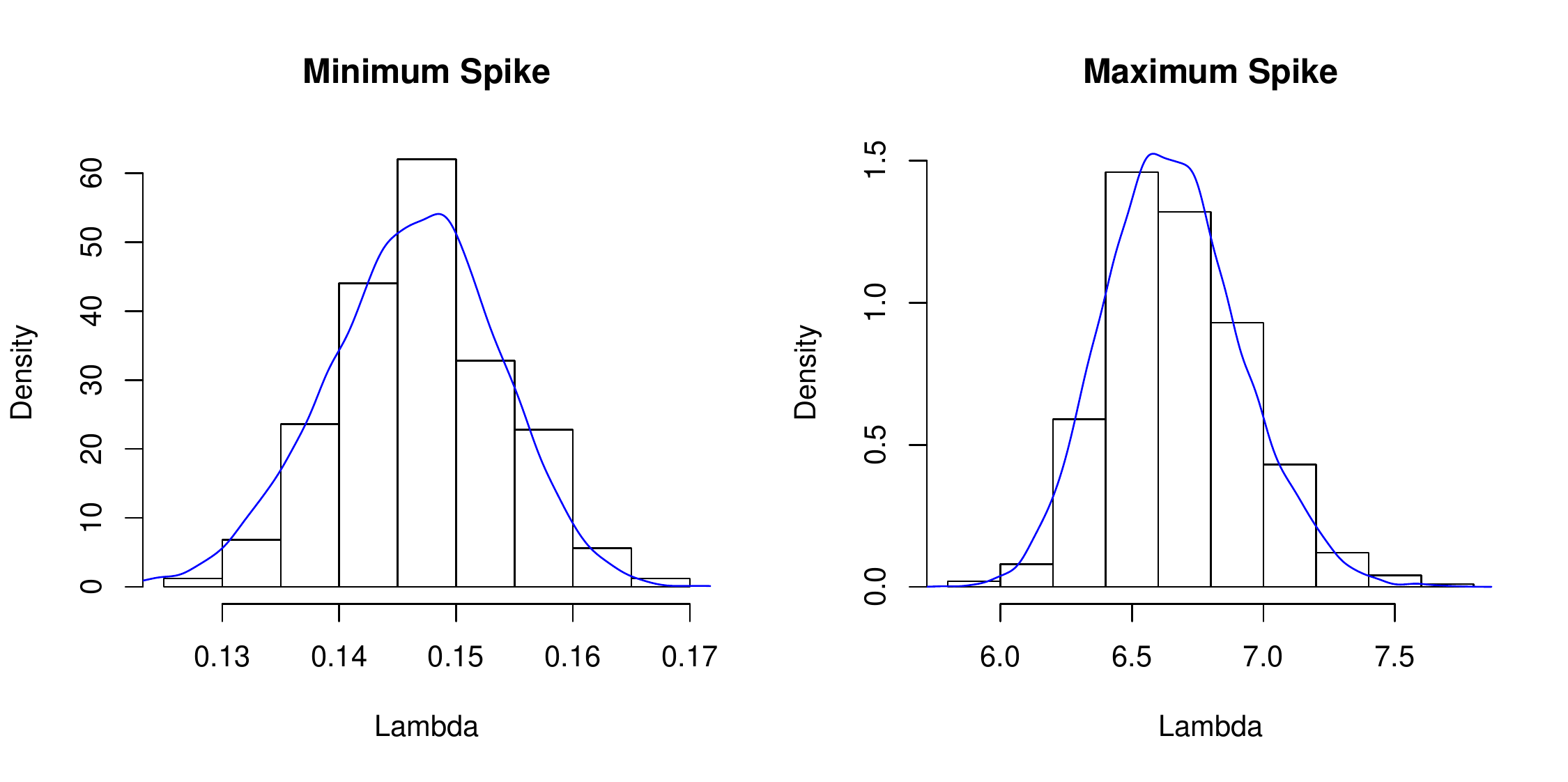}   
    \end{minipage} 
    &
    \begin{minipage}{.45\textwidth} \centering
      \includegraphics[width=1\textwidth]{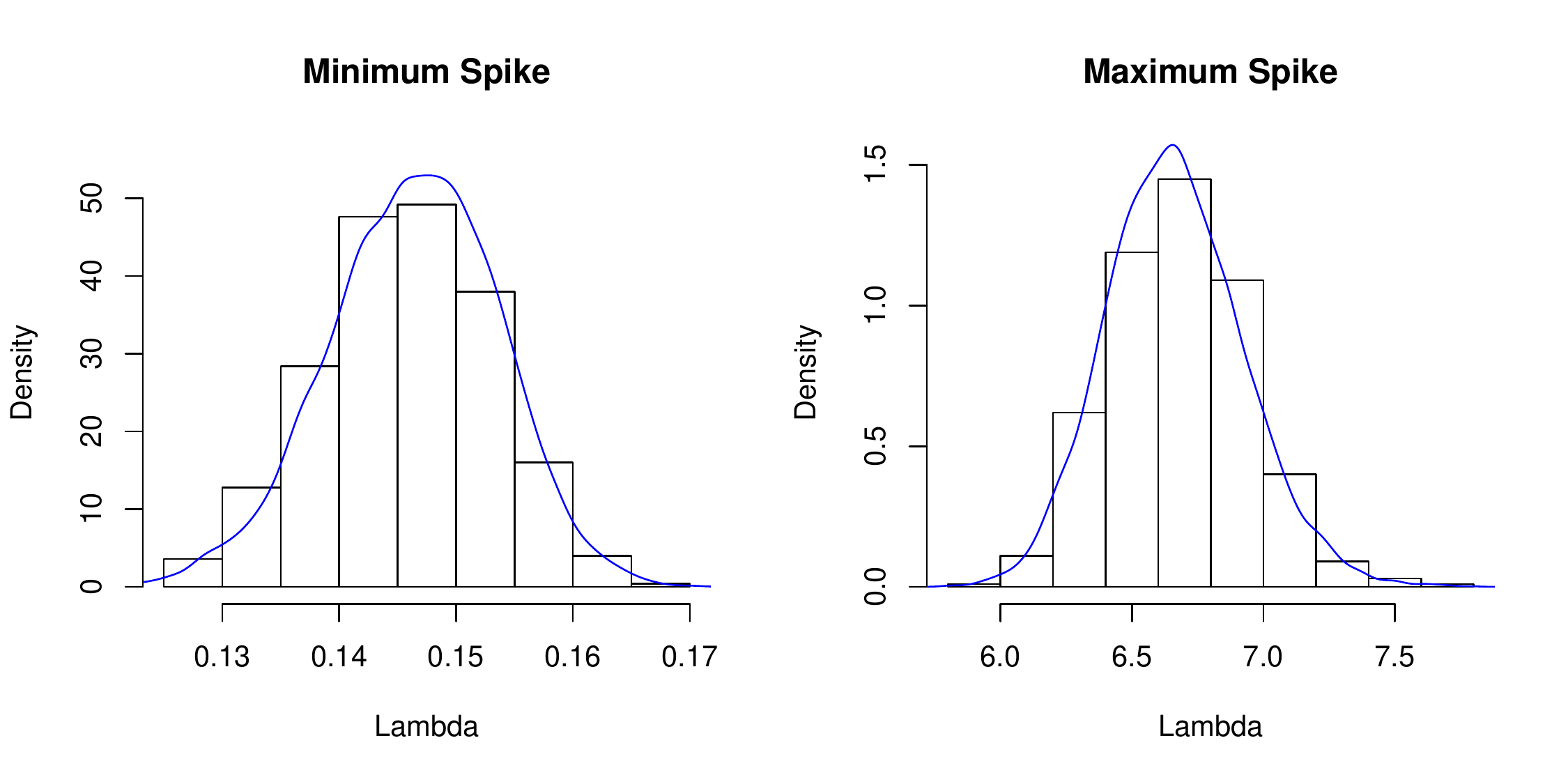}   
    \end{minipage}\\
\begin{minipage}{.45\textwidth} \centering
\scalebox{0.8}{
\begin{tabular}{c}
Scenario 1 \\
\begin{tabular}{c|c|c}
$\rho=0.5$ & $c_X=0.5$ & $c_Y=2$ \\ 
\hline 
$m=1000$ & $n_X=2000$ & $n_Y=500$ \\ 
\end{tabular} \\
 k=4, $\vec{\theta}=\left(15'000,5000,2000,500\right)$.
\end{tabular}}
\end{minipage}
 & 
\begin{minipage}{.45\textwidth} \centering
\scalebox{0.8}{
\begin{tabular}{c}
Scenario 2 \\
\begin{tabular}{c|c|c}
$\rho=0.5$ & $c_X=0.5$ & $c_Y=2$  \\ 
\hline 
$m=1000$ & $n_X=2000$ & $n_Y=500$ \\ 
\end{tabular} \\
 k=4, $\vec{\theta}=\left(5'000,5000,5000,5000\right)$.
\end{tabular}}
\end{minipage} \\
\begin{minipage}{.45\textwidth} \centering
      \includegraphics[width=1\textwidth]{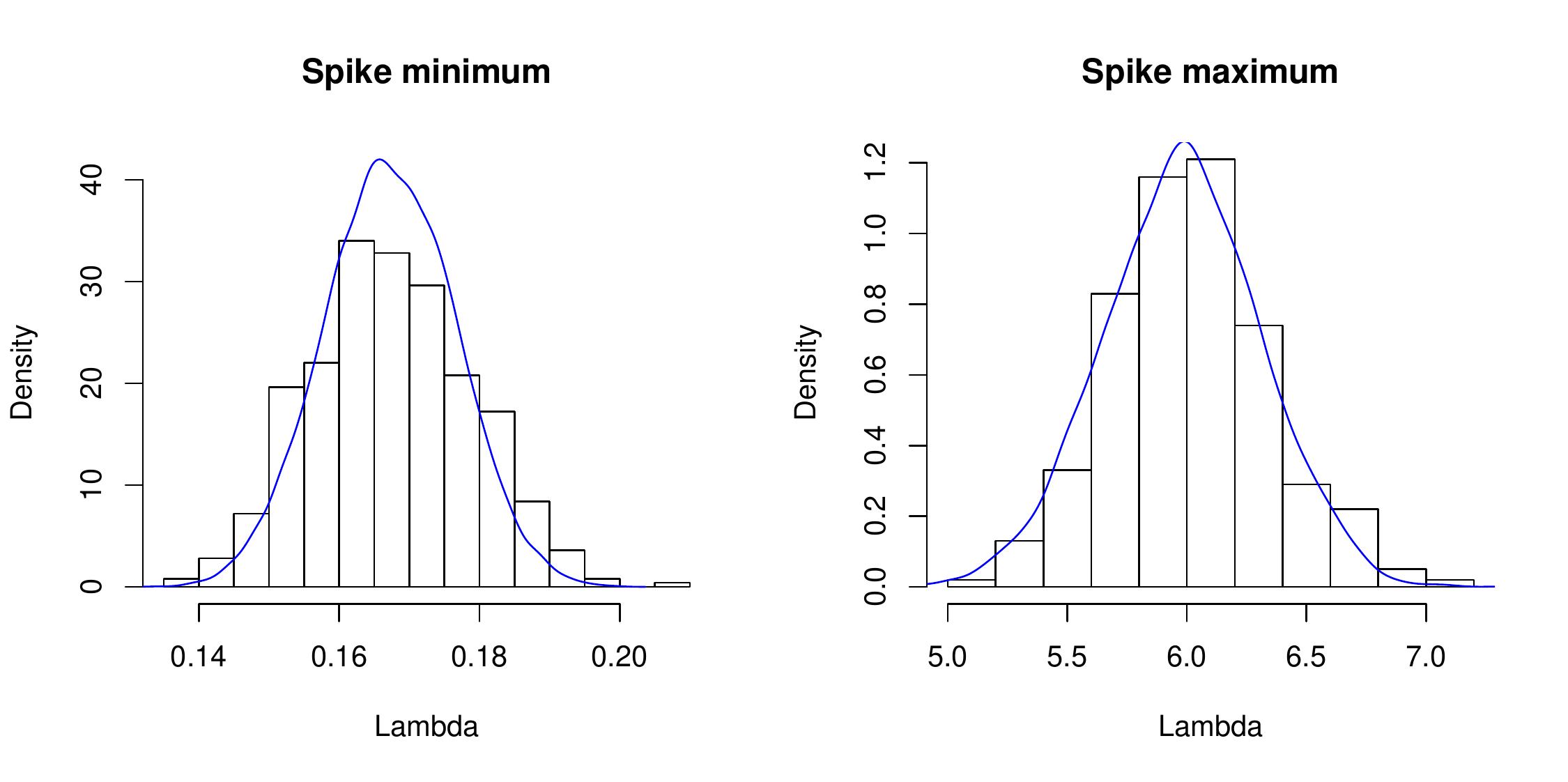}   
    \end{minipage} 
    &
    \begin{minipage}{.45\textwidth} \centering
      \includegraphics[width=1\textwidth]{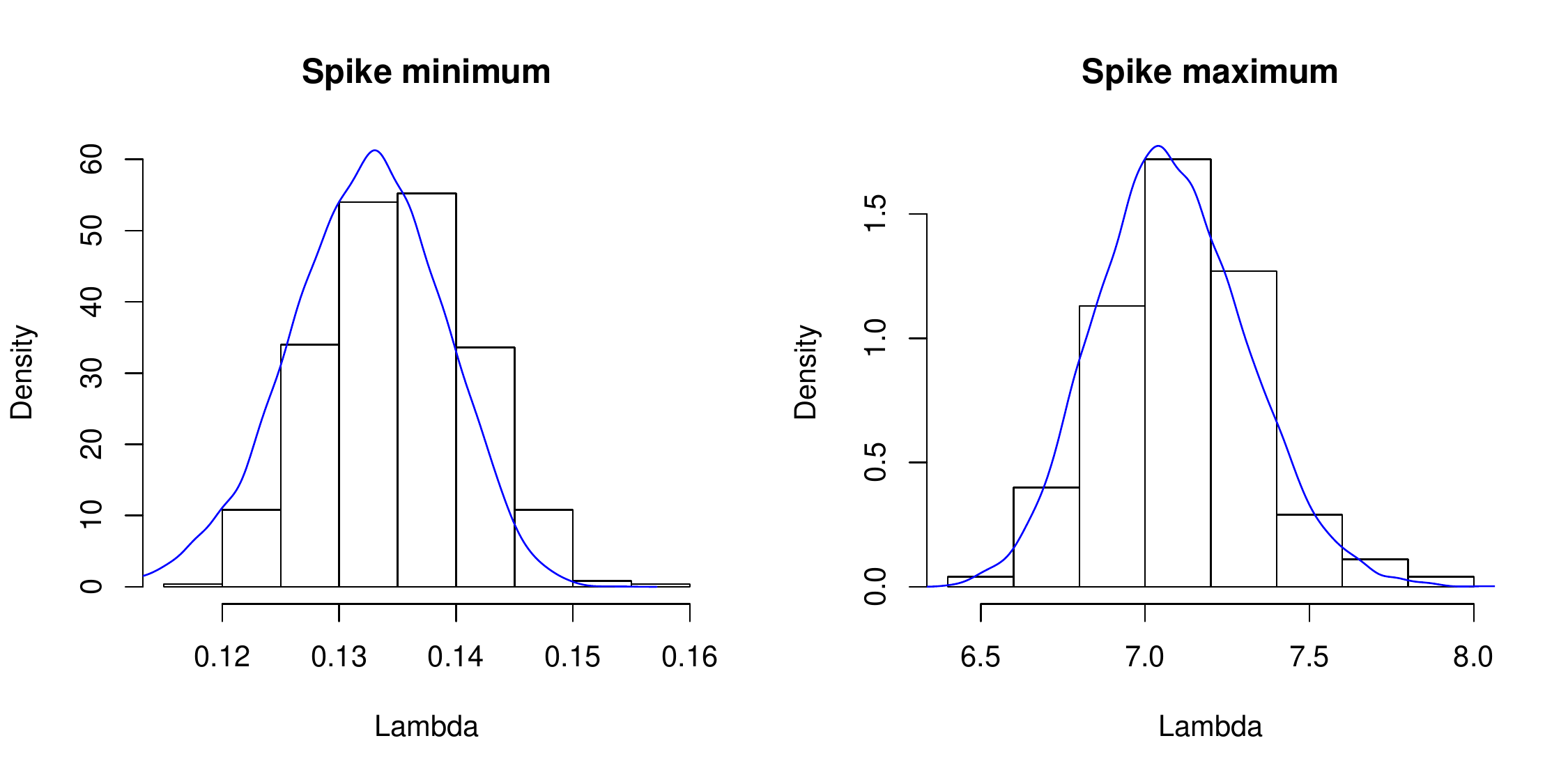}   
    \end{minipage}\\
\begin{minipage}{.45\textwidth} \centering
\scalebox{0.8}{
\begin{tabular}{c}
Scenario 3 \\
\begin{tabular}{c|c|c}
$\rho=0.5$ & $c_X=0.5$ & $c_Y=2$  \\ 
\hline 
$m=1000$ & $n_X=2000$ & $n_Y=500$ \\ 
\end{tabular} \\
 k=1, $\vec{\theta}=5'000$.
\end{tabular}}
\end{minipage}
 & 
\begin{minipage}{.45\textwidth} \centering
\scalebox{0.8}{
\begin{tabular}{c}
Scenario 4 \\
\begin{tabular}{c|c|c}
$\rho=0.5$ & $c_X=0.5$ & $c_Y=2$  \\ 
\hline 
$m=1000$ & $n_X=2000$ & $n_Y=500$ \\ 
\end{tabular} \\
 k=8, $\vec{\theta}=\left(5'000,5'000,...,5'000\right)$.
\end{tabular}}
\end{minipage}\\
\begin{minipage}{.45\textwidth} \centering
      \includegraphics[width=1\textwidth]{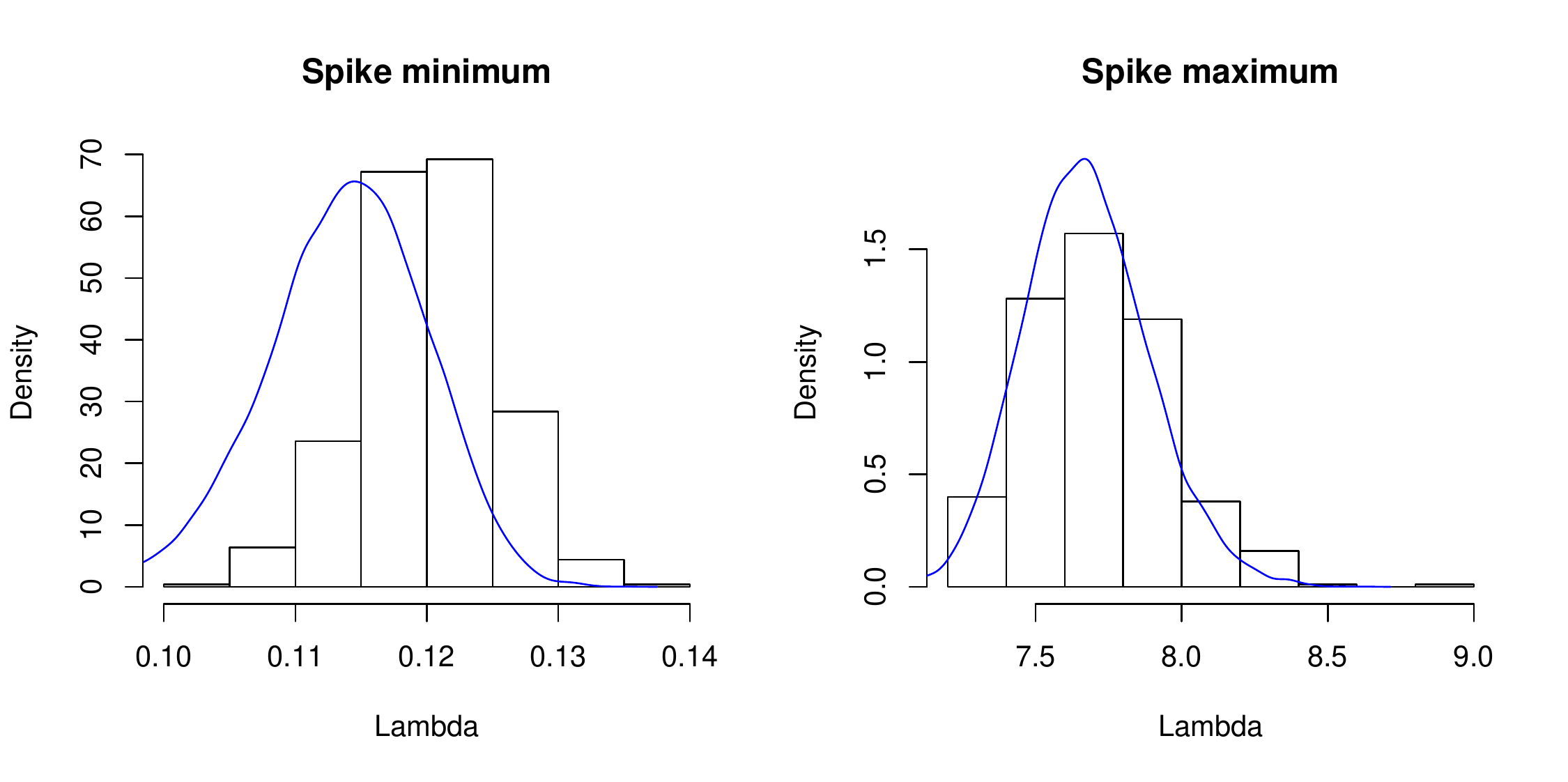}   
    \end{minipage} 
    &
    \begin{minipage}{.45\textwidth} \centering
      \includegraphics[width=1\textwidth]{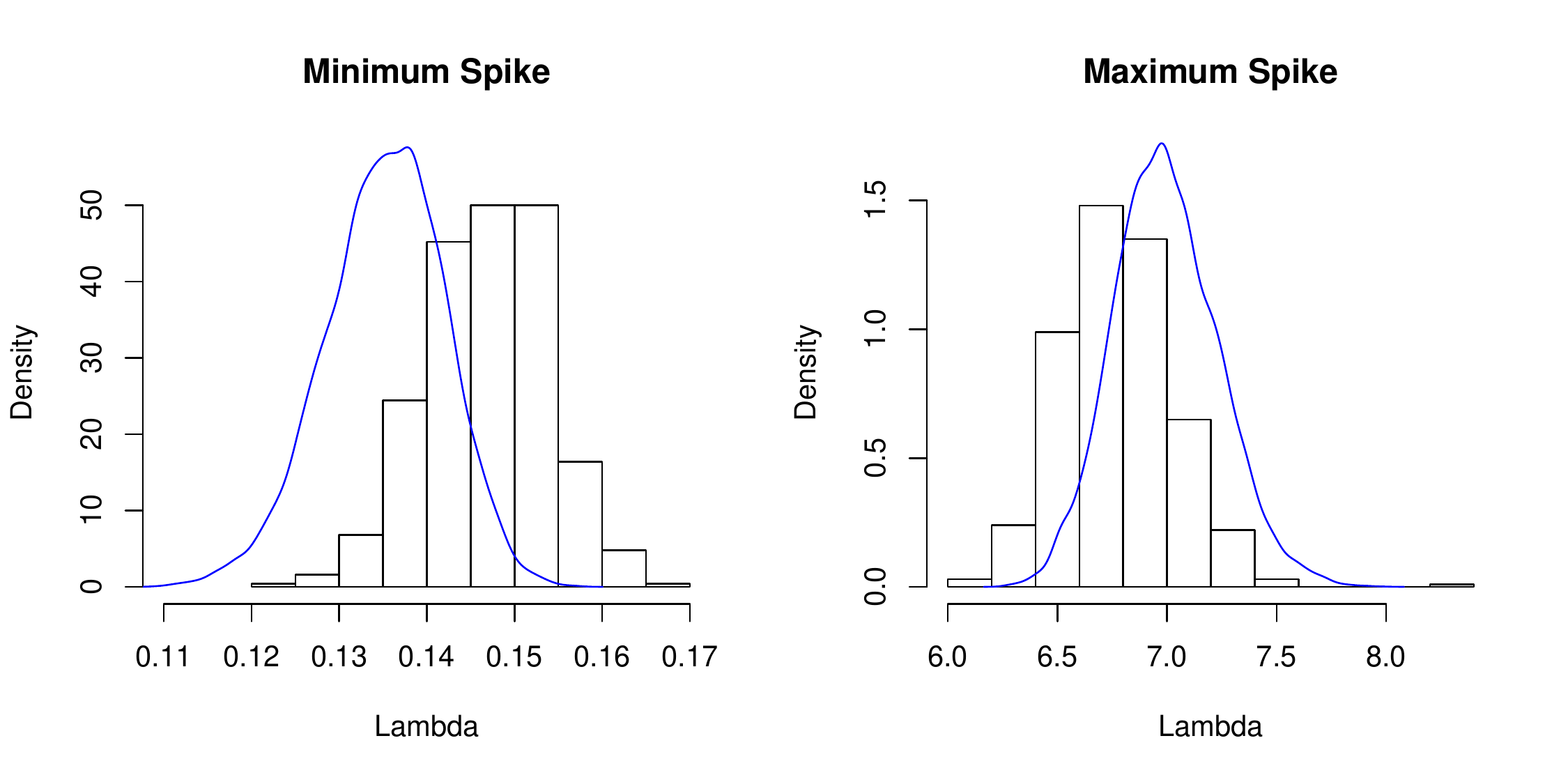}   
    \end{minipage}\\
\begin{minipage}{.45\textwidth} \centering
\scalebox{0.8}{
\begin{tabular}{c}
Scenario 5 \\
\begin{tabular}{c|c|c}
$\rho=0.5$ & $c_X=0.5$ & $c_Y=2$  \\ 
\hline 
$m=1000$ & $n_X=2000$ & $n_Y=500$ \\ 
\end{tabular} \\
 k=15, $\vec{\theta}=\left(5'000,5'000,...,5'000\right)$.
\end{tabular}}
\end{minipage}
 & 
\begin{minipage}{.45\textwidth} \centering
\scalebox{0.8}{
\begin{tabular}{c}
Scenario 6 \\
\begin{tabular}{c|c|c}
$\rho=0.5$ & $c_X=0.5$ & $c_Y=2$  \\ 
\hline 
$m=1000$ & $n_X=2000$ & $n_Y=500$ \\ 
\end{tabular} \\
 k=4, $\vec{\theta}=\left(15'000,5'000,2'000,500\right)$\\
 $k_{est}=7$.
\end{tabular}}
\end{minipage}\\
\begin{minipage}{.45\textwidth} \centering
      \includegraphics[width=1\textwidth]{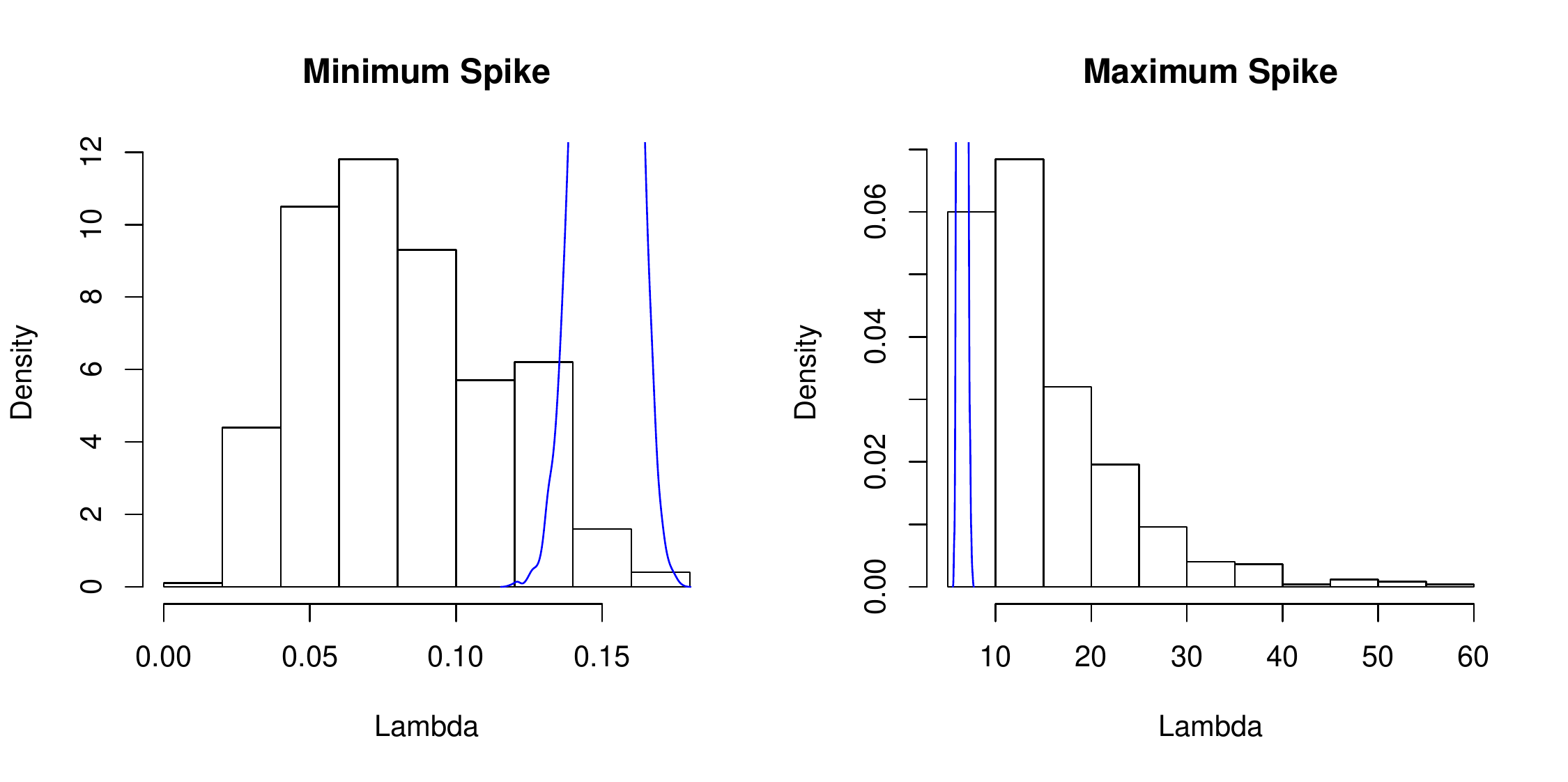}   
    \end{minipage} 
    &
    \begin{minipage}{.45\textwidth} \centering
      \includegraphics[width=1\textwidth]{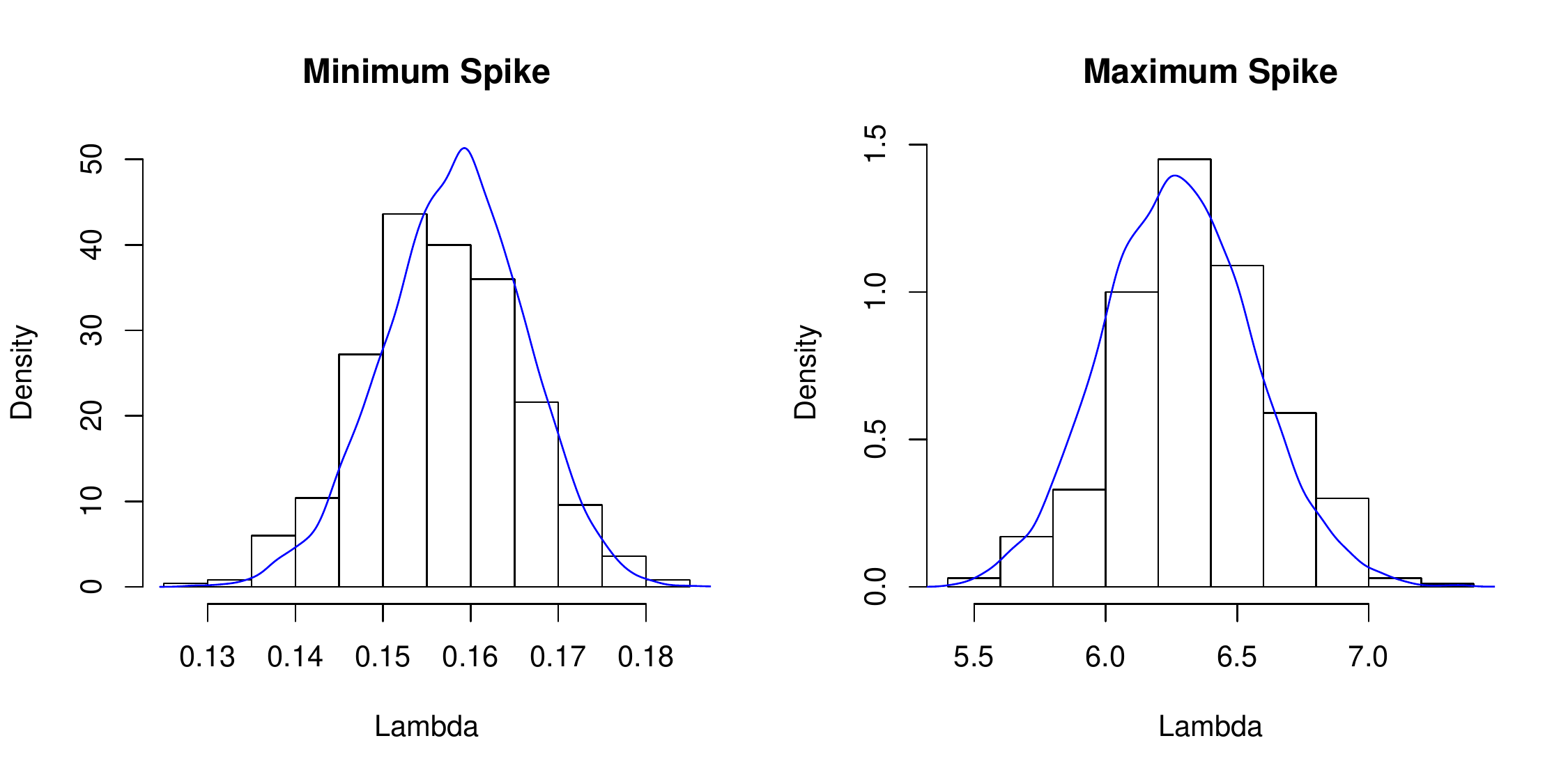}   
    \end{minipage}\\
\begin{minipage}{.45\textwidth} \centering
\scalebox{0.8}{
\begin{tabular}{c}
Scenario 7 \\
\begin{tabular}{c|c|c}
$\rho=0.5$ & $c_X=0.5$ & $c_Y=2$  \\ 
\hline 
$m=1000$ & $n_X=2000$ & $n_Y=500$ \\ 
\end{tabular} \\
 k=4, $\vec{\theta}=\left(15'000,5'000,2'000,500\right)$\\
 $k_{est}=3$.
\end{tabular}}
\end{minipage}
 & 
\begin{minipage}{.45\textwidth} \centering
\scalebox{0.8}{
\begin{tabular}{c}
Scenario 8 \\
\begin{tabular}{c|c|c}
$\rho=0.5$ & $c_X=0.5$ & $c_Y=2$  \\ 
\hline 
$m=1000$ & $n_X=2000$ & $n_Y=500$ \\ 
\end{tabular} \\
k=4, $\vec{\theta}=\left(15'000,3'000,8,6\right)$\\
 $k_{est}=2$.
\end{tabular}}
\end{minipage}

\end{tabular}
\caption{Empirical distributions of the residual spikes together with the Gaussian densities from the theorem \ref{TH=Main} (in blue).  
} \label{tab:residual}
\end{table}

\paragraph*{Multiple eigenvalues}  \ \\
Despite the lack of a proof, the maximum residual distribution when the eigenvalues of the perturbations are multiple is well approximated by our Theorem. This can be seen in Table \ref{tab:residual} but also in Appendix \ref{appendix:Tablecomplete}. 
\paragraph*{Different values of $k$}  \ \\
Scenario 3, 4 and 5 of Table \ref{tab:residual} shows that the result holds for different values of $k$. The less accurate result of scenario 5 is due to the relatively large $k$, whereas the theorem is based on an approximation which considers $k$ to be small compared to $m$. The precision of the asymptotic approximation would be better for $k=15$ when $m=10'000$ with the same $c_X$ and $c_Y$, for example.
\paragraph*{Wrong estimation of $k$}  \ \\
Scenario 6, 7 and 8 of Table \ref{tab:residual} shows the impact of using wrong values of $k$. We see in Scenario 6 that a small overestimation of $k$ leads to a small overestimation of the maximum and small underestimation the minimum. This will lead to conservative tests. Scenario 7 shows that underestimation of $k$ can lead to a bad approximation but in this scenario we neglect perturbations of size $500$! Scenario 8 shows that neglecting two small perturbations of size $6$ and $8$, as we could easily do by mistake, still leads to very accurate approximations.\\
The simulation of Table \ref{tab:residual} are done to convince the reader of the usefulness of Theorem \ref{TH=Main}. In practice, we must estimate the parameters needed in the approximation. An arguments based on the Cauchy-interlacing theorem can convinced the reader that we can estimate 
\begin{eqnarray*}
M_{X,s}=\frac{1}{m}\sum_{i=1}^m \lambda_{W_X}^s
\end{eqnarray*}
by 
\begin{eqnarray*}
\hat{M}_{X,s}=\frac{1}{m-k}\sum_{i=k+1}^m \lambda_{\hat{\Sigma}_X}^s.
\end{eqnarray*}
The impact of using a wrong value for  $k_{est}$ is investigate in Table \ref{tab=simulationkestimation} in Appendix \ref{appendix:Simulationkest}.

\paragraph*{Other simulations}
In appendix \ref{appendix:Simulationkest} we also investigate the approximation with estimated spectra for data with distributions that are  not invariant by rotation. In some scenarios the approximation succeeds to estimate the location but failed to correctly estimate the variance. In others, the location of the maximum residual spike is overestimated and the minimum residual spike is underestimated. This would again lead to conservative tests, but suggests a lack of power.

\section{An application}
In this section, we apply our procedure developed from Theorem \ref{TH=Main} to data $\rX$ and $\rY$. First, each step is briefly explained. Then, an analysis is presented on simulated data together with the mathematical work and the important plots.\\ 
This procedure is not unique and other solutions better adapted to the problem could be implemented. For example, the choice of $k$ and the number of perturbations, could certainly be improved. The goal of this section is to provide a procedure as conservative as possible with reasonably good asymptotic power.

\begin{enumerate}
\item First, we center the data with regard to the rows and columns. 
\item Then, we need to estimate $k$ and rescale the variance. These two tasks are interconnected. One intuitive way to choose $k$ for each matrix ($k_X$ and $k_Y$) consists in looking at the spectra for spikes and keeping in mind that overestimation is preferable to underestimation of the actual value.\\
Using $k=\max(k_X,k_Y)$, we can then rescale the matrices $\rX$ and $\rY$ to create $\bX$ and $\bY$. 
\item Next, we apply the procedure to $\bX$ and $\bY$, using the above $k$. In our case, this leads to two observed extreme residual spikes.
\item We compute the distribution of the residual spike by assuming $k$ perturbation and estimating $M_{s,X}$. 
\item Finally, we can compare the extreme values with their distribution under $\text{H}_0$ for testing purposes.
\end{enumerate}

\begin{Rem} \
Our simulations in Appendix \ref{appendix:Simulationkest}, show that the choice of $k$ does not affect the conservative nature of the test. A strong underestimation of $k$, however, greatly reduces the power. This explains the advice to overestimate $k$.
\end{Rem}

\subsection{Analysis} \label{sec:simulAnalysis}
We observe data $\rX\in \mathbf{R}^{m \times n_X}$ and $\rY\in \mathbf{R}^{m \times n_Y}$ that we suppose is already centred by rows and columns. We choose $k$ by looking at the histogram of the matrices in Figure \ref{Figexample1} where $m=1000$, $n_X=2000$ et $n_Y=500$.
\begin{figure}[hbtp]
\includegraphics[width=0.5\paperwidth]{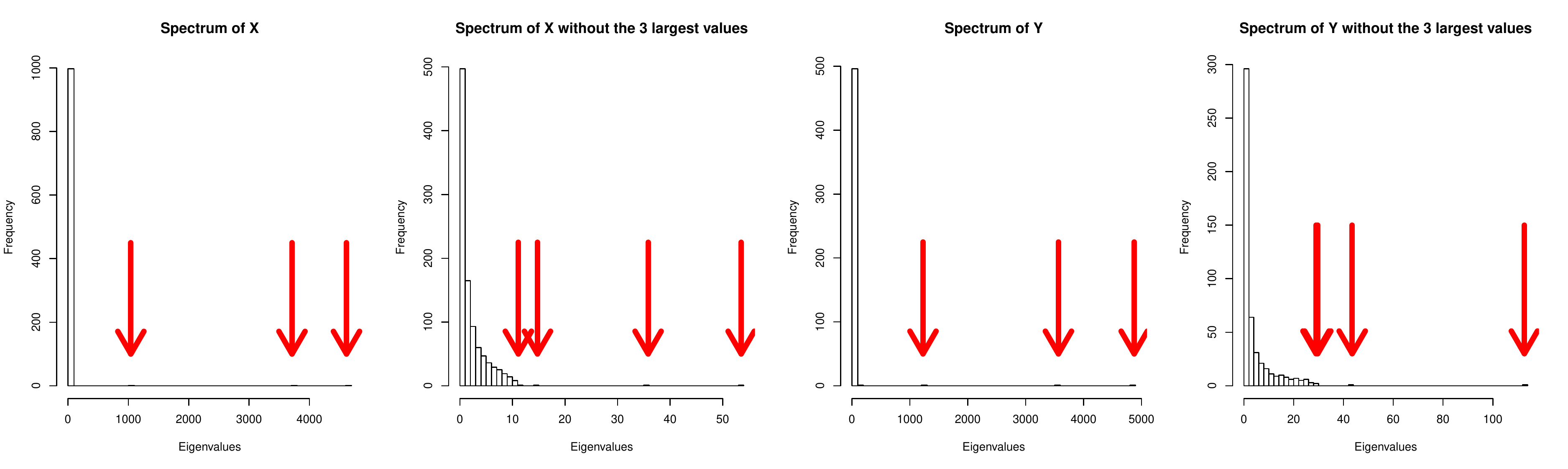} 
\caption{Spectra of $\rX$ and $\rY$ with largest isolated eigenvalues indicated by arrows.}\label{Figexample1}
\end{figure} 
We try to overestimate a lower bound on  $k$ based on Figure \ref{Figexample1}. The spectrum of $\rX$ seems to have $6$ isolated eigenvalues, but we could argue that two other eigenvalues are perturbations. The spectrum $\rY$ clearly shows $5$ isolated eigenvalues and at most $2$ additional ones. We thus set $k=8$ knowing that we probably overestimate the true value. Next, we estimate the variances, 
\begin{eqnarray*}
\hat{\sigma}_X^2=  \frac{1}{m-k}\sum_{i=k+1}^m \lambda_i\left( \frac{1}{n_X}\rX\rX^t \right),\\
\hat{\sigma}_Y^2= \frac{1}{m-k} \sum_{i=k+1}^m \lambda_i\left( \frac{1}{n_Y}\rY\rY^t \right).
\end{eqnarray*} 
We can then rescale the matrices $\rX$ and $\rY$ by $\hat{\sigma}_X$ and $\hat{\sigma}_Y$, respectively, to create the covariance matrices
\begin{eqnarray*}
\hat{\Sigma}_X=\frac{1}{n_X \hat{\sigma}_X^2} \rX \rX^t \text{ and }
\hat{\Sigma}_Y=\frac{1}{n_Y \hat{\sigma}_Y^2} \rY \rY^t.
\end{eqnarray*}
Finally, we filter the matrices as in definition \ref{Def=unbiased}. 
\begin{eqnarray*}
&&\hat{\hat{\Sigma}}_X = \I_m + \sum_{i=1}^k \left(\hat{\hat{\theta}}_{X,i}-1 \right) \hat{u}_{\hat{\Sigma}_X,i}\hat{u}_{\hat{\Sigma}_X,i}^t,\\
&&\hat{\hat{\theta}}_{X,i}= 1+\frac{1}{\frac{1}{m-k} \sum_{j=k+1}^m \frac{\hat{\lambda}_{ \hat{\Sigma}_X,j}}{\hat{\lambda}_{ \hat{\Sigma}_X,i}-\hat{\lambda}_{ \hat{\Sigma}_X,j}}}.
\end{eqnarray*}
The computed residual spikes of $\hat{\hat{\Sigma}}_X^{-1} \hat{\hat{\Sigma}}_Y$ are shown in Table \ref{Tabexample1}.

\begin{table}[hbtp]
\begin{tabular}{l|cccccccc}
$\lambda_{\max}$ & 56.03 & 10.25 & 9.88 & 8.96 & 8.29 & 7.27 & 5.71 & 5.10 \\
$\lambda_{\min}$ & 0.04 & 0.10 & 0.13 & 0.13 & 0.18 & 0.21 & 0.34 & 0.36 
\end{tabular}
\caption{Observed residual spikes.}\label{Tabexample1}
\end{table} 

Using Figure \ref{Figexample2} a, these values are compared to the theoretical distributions of the extreme residual spikes assuming equality of the perturbations of order $k$. The distribution in blue uses the usual estimator of the spectra and the distribution in orange uses the conservative estimator introduced in Appendix \ref{appendix:Spectrumestimation}.
The moments of the spectra are summarize in Figure \ref{Figexample2} b. 
\begin{figure}[hbtp]
\begin{center}
\begin{tabular}{cc}
\begin{minipage}{0.3\paperwidth}
\includegraphics[width=0.3\paperwidth]{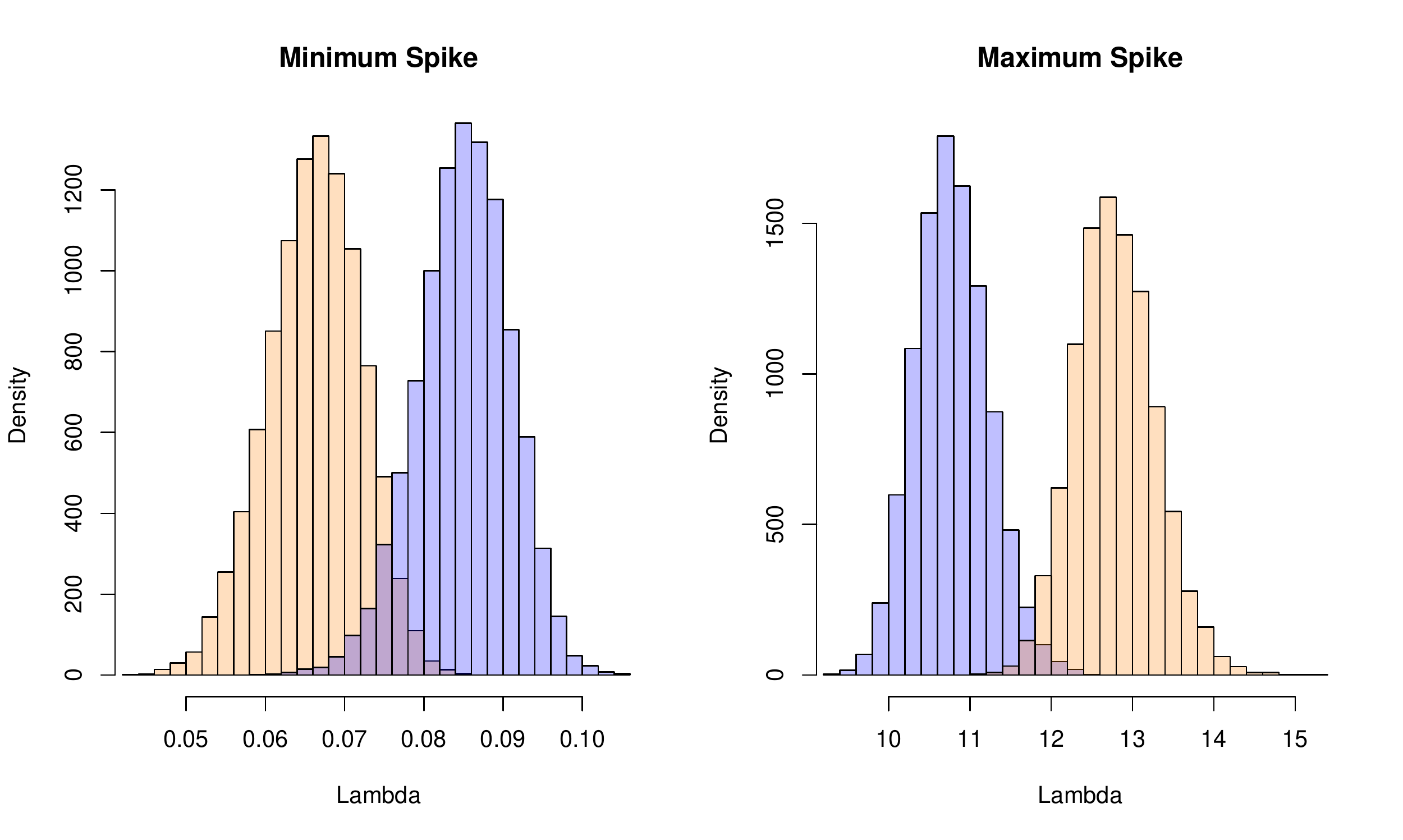}
 \end{minipage}
& 
\begin{minipage}{0.3\paperwidth}
\scalebox{0.7}{
\begin{tabular}{l|cc}
&$\lambda_{\min}$ & $\lambda_{\max}$\\
\color{blue} Usual & (0.085, 0.006) & (10.75, 0.45)\\
 \color{orange} Robust & (0.067, 0.006) & (12.67, 0.47)
\end{tabular}}
\end{minipage}\\
a & b
\end{tabular}
\end{center}
\caption{a: Distribution of the extreme residual spike assuming equality of the covariance, $k=8$ and $\theta_i$ large. (Robust estimation of the spectra in orange.) b: Estimated residual spikes moments, $(\mu,\sigma)$ using usual or robust estimators of the spectral moments. }\label{Figexample2}
\end{figure}

We finally clearly detect two residual spikes. Figure \ref{Figexample3} presents the residual eigenvectors of the residual eigenvalues.

\begin{figure}[hbtp]
\begin{center}
\begin{tabular}{cc}
$\lambda=56.03$ & $\lambda=0.04 $\\
\includegraphics[width=0.25\paperwidth]{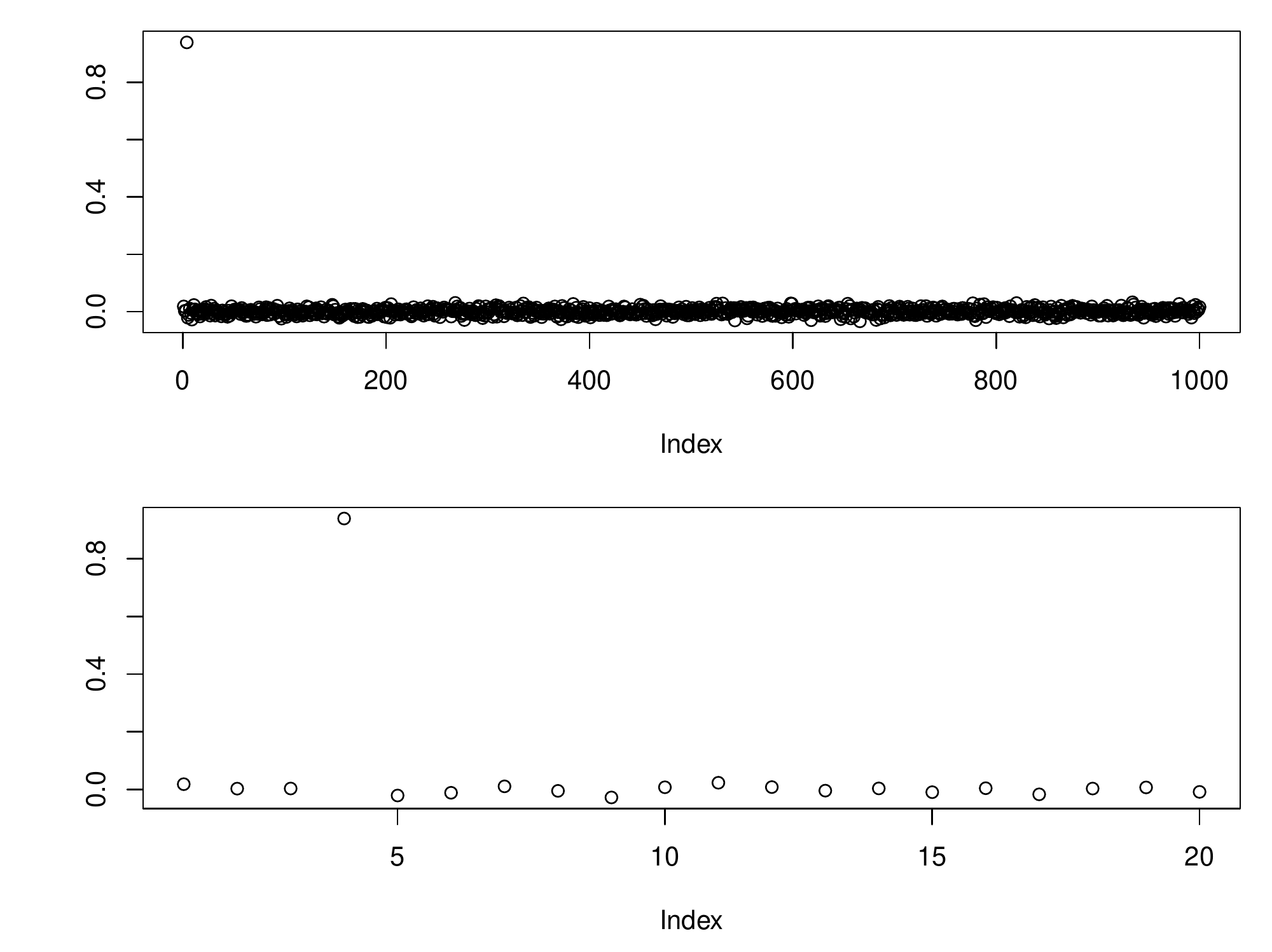} & \includegraphics[width=0.25\paperwidth]{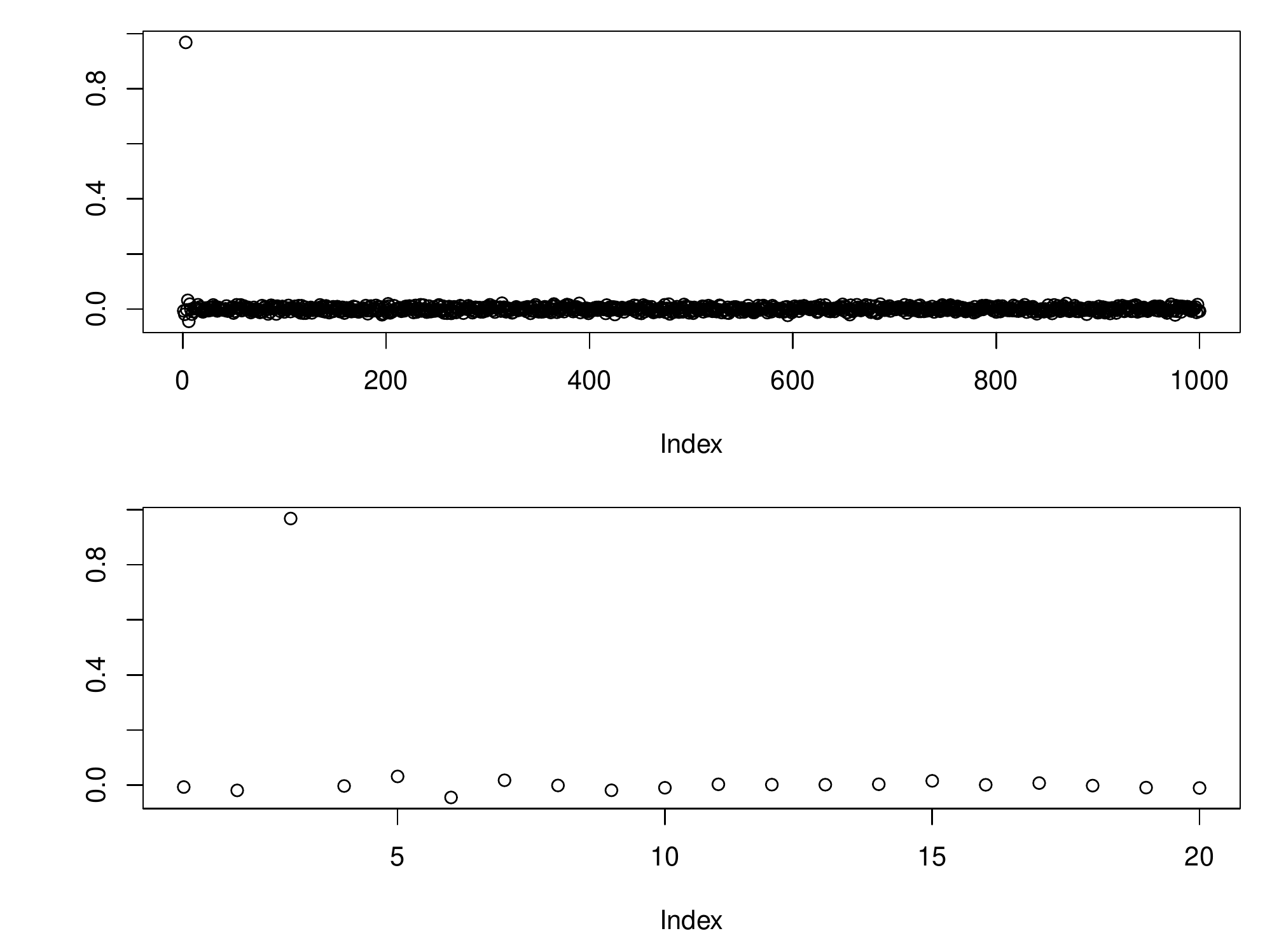} \\
$\lambda=10.25$ & $\lambda=0.10$ \\
\includegraphics[width=0.25\paperwidth]{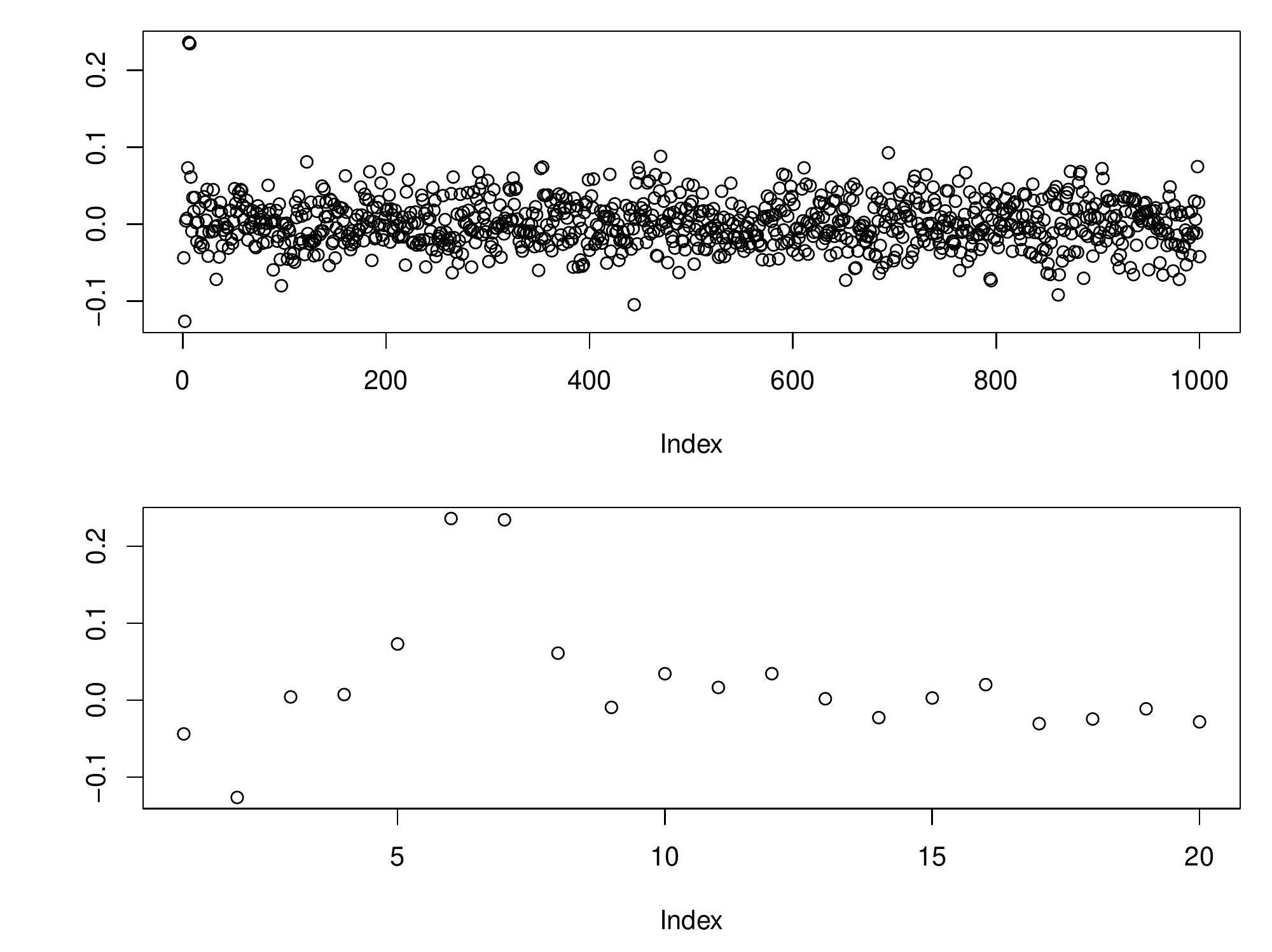} & \includegraphics[width=0.25\paperwidth]{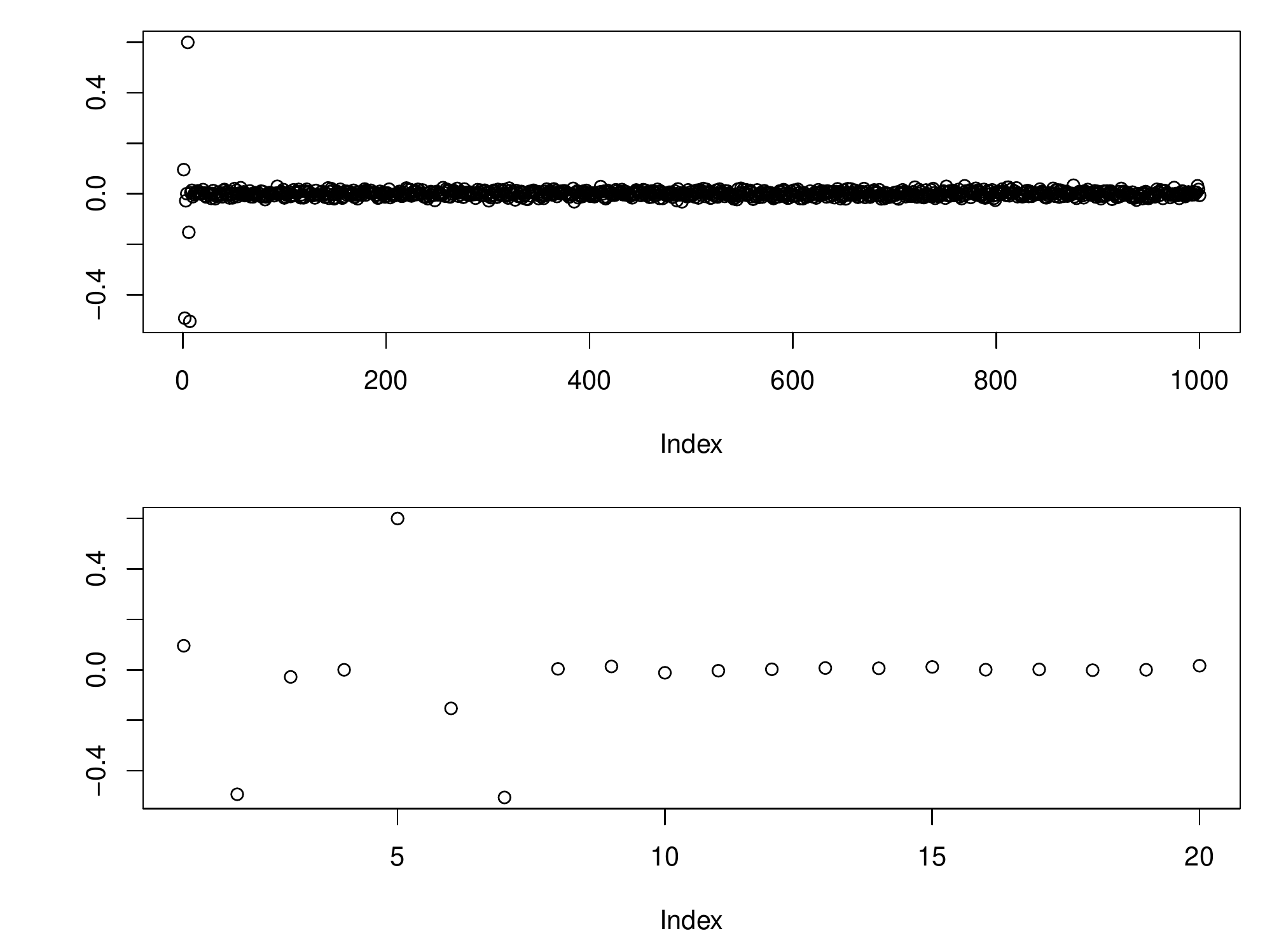} 
\end{tabular}
\end{center}
\caption{Representation of the entiere residual eigenvectors and only the 20 first entries. }\label{Figexample3}
\end{figure}

We conclude that the differences are in direction $e_3$ and $e_4$. As we see in the figure, two other eigenvectors also exhibit a structure. Without our test, we could have concluded that they also represent significant differences, but this residual structure is merely due to the biased estimation of the eigenvectors.

\begin{center}
\textbf{A structure in a residual eigenvector does not imply a real difference!}
\end{center}

\subsection{Conclusion}
By studying perturbation of order $1$ in \cite{mainarticle} and perturbation of order $k$ in \cite{mainarticle2}, we highlighted the lack of power of the usual procedure to detect differences between two groups. This paper extended the residual spike to perturbations of order $k$ by using tools introduced in previous papers. While this test has weaker power than in  \cite{mainarticle2}, it has the important advantage do be able to deal with multiple equal eigenvalues. Additional simulations investigating the robustness of this new procedure are contained in the thesis \cite{mathese} and seems promising.

\pagebreak

\appendix

\section{Table} \label{appendix:Tablecomplete}
We extend the simulations of Section \ref{sec:smallsimulationmainth}. We test our Main Theorem \ref{TH=Main} under different hypotheses on $\X\in\mathbf{R}^{m \times n_X}$ and $\Y\in\mathbf{R}^{m \times n_Y}$ (recall that $W_X=\frac{1}{n_X}\X \X^t$ and $W_Y=\frac{1}{n_Y}\Y \Y^t$):
\begin{enumerate}
\item The matrices $\X$ and $\Y$ contain independent standard normal entries.
\item The columns of the matrices $\X$ and $\Y$ are i.i.d. with a multivariate Student's distribution with $8$ degrees of freedom. For $i=1,2,...,n_X$ and $j=1,2...,n_Y$,
$$X_{\cdot,i} \overset{i.i.d.}{\sim} \frac{\Normal\left(\vec{0},\I_m\right)}{\sqrt{\frac{\chi^2_8}{8}}} \text{ and } Y_{\cdot,j} \overset{i.i.d.}{\sim} \frac{\Normal\left(\vec{0},\I_m\right)}{\sqrt{\frac{\chi^2_8}{8}}}  $$ 
\item The rows of $\X$ and $\Y$ are i.i.d. Gaussian ARMA entries of parameters $\text{AR}=(0.6,0.2)$ and $\text{MA}=(0.5,0.2)$. Moreover, the traces of the matrices are standardised by the estimated variance.
\end{enumerate}

Theorem \ref{TH=Main} is investigate through Table \ref{tab=simulationmainth1} and \ref{tab=simulationmainth2}. 
The estimates of the mean and the standard error of the residual spikes $(\hat{\mu},\hat{\sigma})$ are compare to their empirical values $(\mu,\sigma)$. The simulations are computed for the three scenarios described above. The perturbation $P=\I_m+ \sum_{i=1}^k (\theta_i-1) u_i u_i$ is without loss of generality assumed canonical and the eigenvalue $\theta_i$ are fixed and equal to 5000. \\

These simulation confirms that our Theorem is valid despite some less accurate results in {\color{red} red}. Even though the values of $n_X$ and $n_Y$ are large, this lack of accuracy is probably due to the temporal correlation of the data that reduces the equivalent number of independent columns.

\begin{table}
\begin{center}
\begin{adjustbox}{addcode={\begin{minipage}{\width}}{\caption{Simulations of the maximum residual spikes. The values $(\hat{\mu},\hat{\sigma})$ and $(\mu,\sigma)$ are respectively the estimations of the mean and the standard error of the residual spikes obtained by the Main Theorem \ref{TH=Main} and empirical methods using 500 replicates and $\theta_i=5000$, respectively.  } \label{tab=simulationmainth1}
\end{minipage}},rotate=90,center}
\scalebox{0.55}{
\begin{tabular}{c c}
&$\lambda_{\max}\left(\hat{\hat{\Sigma}}_{X,P_k}^{-1/2} \hat{\hat{\Sigma}}_{Y,P_k} \hat{\hat{\Sigma}}_{X,P_k}^{-1/2}\right)$ \\ 
 \rotatebox{90}{\hspace{-2.2cm} 1. Normal entries. }&\tiny 
\begin{tabular}{c @{\,\vrule width 1.5mm\,}  c @{\,\vrule width 1.5mm\,} c c c c @{\,\vrule width 0.5mm\,}  c c c c @{\,\vrule width 0.5mm\,} c c c c @{\,\vrule width 0.5mm\,} c c c c }
\backslashbox{$n_Y$}{$n_X$} & & \multicolumn{4}{c}{100} & \multicolumn{4}{c}{500} & \multicolumn{4}{c}{1000} & \multicolumn{4}{c}{2000}\\
\specialrule{1.5mm}{0pt}{0pt}
 & \backslashbox{$k$}{$m$} & \multicolumn{2}{c}{100}  & \multicolumn{2}{c}{1000}  &  \multicolumn{2}{c}{100}  & \multicolumn{2}{c}{1000} & \multicolumn{2}{c}{100}  & \multicolumn{2}{c}{1000} & \multicolumn{2}{c}{100}  & \multicolumn{2}{c}{1000} \\
 \specialrule{1.5mm}{0pt}{0pt}
 & &  $(\hat{\mu},\hat{\sigma})$  & $(\mu,\sigma)$ &  $(\hat{\mu},\hat{\sigma})$  & $(\mu,\sigma)$ & $(\hat{\mu},\hat{\sigma})$  & $(\mu,\sigma)$ & $(\hat{\mu},\hat{\sigma})$  & $(\mu,\sigma)$ & $(\hat{\mu},\hat{\sigma})$  & $(\mu,\sigma)$ & $(\hat{\mu},\hat{\sigma})$  & $(\mu,\sigma)$ & $(\hat{\mu},\hat{\sigma})$  & $(\mu,\sigma)$ & $(\hat{\mu},\hat{\sigma})$  & $(\mu,\sigma)$ \\
 \specialrule{1.5mm}{0pt}{0pt}
\multirow{3}{*}{100} 
& 1   & $(3.69,0.47 )$ & $(3.77,0.48 )$ & $(21.91,2.38 )$ & $(22.32,2.36 )$ & $(2.88,0.25 )$ & $(2.94,0.25 )$ & $(13.93,0.71 )$ & $(14.15,0.69 )$ & $(2.76,0.21 )$ & $(2.78,0.21 )$ & $(13.01,0.60 )$ & $(13.14,0.57 )$ & $(2.73,0.22 )$ & $(2.75,0.23 )$ & $(12.39,0.52 )$ & $(12.59,0.50)$ \\ 
 & 4 & $(4.67,0.41 )$ & $(4.85,0.49 )$ & $(26.66,1.92 )$ & $(28.21,2.84 )$ & $(3.37,0.20 )$ & $(3.39,0.22 )$ & $(15.32,0.56 )$ & $(15.50,0.60 )$ & $(3.18,0.18 )$ & $(3.19,0.18 )$ & $(14.12,0.45 )$ & $(14.30,0.44 )$ & $(3.15,0.19 )$ & $(3.17,0.20 )$ & $(13.40,0.38 )$ & $(13.52,0.41)$ \\ 
 & 10 & $(5.64,0.40 )$ & $(6.29,0.71 )$ & $(31.28,1.77 )$ & $(35.98,3.38 )$ & $(3.87,0.18 )$ & $(3.88,0.21 )$ & $(16.66,0.47 )$ & $(16.95,0.55 )$ & $(3.60,0.16 )$ & $(3.56,0.15 )$ & $(15.24,0.40 )$ & $(15.39,0.45 )$ & $(3.59,0.16 )$ & $(3.56,0.16 )$ & $(14.37,0.34 )$ & $(14.56,0.34)$ \\
\specialrule{0.5mm}{0pt}{0pt}
\multirow{3}{*}{500} 
& 1 &  &  &  &  & $( 1.87,0.12 )$ & $(1.87,0.12 )$ & $(5.85,0.33 )$ & $(5.86,0.34 )$ & $(1.72,0.09 )$ & $(1.73,0.09 )$ & $(4.79,0.20 )$ & $(4.81,0.21 )$ & $(1.63,0.07 )$ & $(1.63,0.08 )$ & $(4.28,0.15 )$ & $(4.29,0.16)$ \\ 
&  4 &  &  &  &   & $( 2.10,0.09 )$ & $(2.11,0.10 )$ & $(6.45,0.24 )$ & $(6.58,0.30 )$ & $(1.90,0.07 )$ & $(1.90,0.08 )$ & $(5.19,0.16 )$ & $(5.22,0.17 )$ & $(1.78,0.06 )$ & $(1.77,0.06 )$ & $(4.56,0.11 )$ & $(4.56,0.12)$ \\ 
&  10 &  &  &  &   & $(2.32,0.09 )$ & $(2.33,0.10 )$ & $(7.06,0.21 )$ & $(7.21,0.29 )$ & $(2.07,0.07 )$ & $(2.06,0.07 )$ & $(5.57,0.14 )$ & $(5.63,0.17 )$ & $(1.92,0.05 )$ & $(1.90,0.05 )$ & $(4.85,0.10 )$ & $(4.87,0.11)$ \\
\specialrule{0.5mm}{0pt}{0pt}
\multirow{3}{*}{1000} 
& 1 &  &  &  &  &  &  &  &  & $(1.57,0.07 )$ & $(1.57,0.07 )$ & $(3.73,0.15 )$ & $(3.74,0.15 )$ & $(1.48,0.06 )$ & $(1.47,0.06 )$ & $(3.19,0.10 )$ & $(3.19,0.11)$ \\ 
 & 4 &  &  &  &  &  &  &  &  & $(1.71,0.06 )$ & $(1.70,0.06 )$ & $(4.02,0.11 )$ & $(4.04,0.13 )$ & $(1.58,0.04 )$ & $(1.58,0.04 )$ & $(3.39,0.08 )$ & $(3.39,0.08)$ \\ 
 & 10 &  &  &  &  &  &  &  &  & $(1.84,0.05 )$ & $(1.84,0.05 )$ & $(4.31,0.10 )$ & $(4.36,0.12 )$ & $(1.69,0.04 )$ & $(1.68,0.04 )$ & $(3.58,0.07 )$ & $(3.59,0.08)$ \\
\specialrule{0.5mm}{0pt}{0pt}
\multirow{3}{*}{2000} 
& 1&  & &  &  &  &  &  &  &  &  &  &  & $(1.37,0.04 )$ & $(1.37,0.04 )$ & $(2.62,0.08 )$ & $(2.63,0.08)$ \\ 
 & 4  & &  &  &  &  &  &  &  &  &  &  &  & $(1.46,0.03 )$ & $(1.45,0.03 )$ & $(2.77,0.06 )$ & $(2.77,0.06)$ \\ 
 & 10  &  &  &  &  &  &  &  &  &  &  &  &  & $(1.54,0.03 )$ & $(1.53,0.03 )$ & $(2.92,0.05 )$ & $(2.93,0.06)$ \\
 \hline  \hline  
\end{tabular} \\

\rotatebox{90}{\hspace{-2.2cm} 2. Multivariate Student. } & \tiny \begin{tabular}{c @{\,\vrule width 1.5mm\,}  c @{\,\vrule width 1.5mm\,} c c c c @{\,\vrule width 0.5mm\,}  c c c c @{\,\vrule width 0.5mm\,} c c c c @{\,\vrule width 0.5mm\,} c c c c }
\backslashbox{$n_Y$}{$n_X$} & & \multicolumn{4}{c}{100} & \multicolumn{4}{c}{500} & \multicolumn{4}{c}{1000} & \multicolumn{4}{c}{2000}\\
\specialrule{1.5mm}{0pt}{0pt}
 & \backslashbox{$k$}{$m$} & \multicolumn{2}{c}{100}  & \multicolumn{2}{c}{1000}  &  \multicolumn{2}{c}{100}  & \multicolumn{2}{c}{1000} & \multicolumn{2}{c}{100}  & \multicolumn{2}{c}{1000} & \multicolumn{2}{c}{100}  & \multicolumn{2}{c}{1000} \\
 \specialrule{1.5mm}{0pt}{0pt}
 & &  $(\hat{\mu},\hat{\sigma})$  & $(\mu,\sigma)$ &  $(\hat{\mu},\hat{\sigma})$  & $(\mu,\sigma)$ & $(\hat{\mu},\hat{\sigma})$  & $(\mu,\sigma)$ & $(\hat{\mu},\hat{\sigma})$  & $(\mu,\sigma)$ & $(\hat{\mu},\hat{\sigma})$  & $(\mu,\sigma)$ & $(\hat{\mu},\hat{\sigma})$  & $(\mu,\sigma)$ & $(\hat{\mu},\hat{\sigma})$  & $(\mu,\sigma)$ & $(\hat{\mu},\hat{\sigma})$  & $(\mu,\sigma)$ \\
 \specialrule{1.5mm}{0pt}{0pt}
\multirow{3}{*}{100} 
&1 & $(5.00,1.23)$ & $(5.37,1.24)$ & $(31.66,7.70)$ & $(33.68,7.67)$ & $(3.27,0.37)$ & $(3.35,0.36)$ & $(18.45,2.87)$ & $(19.15,2.96)$ & $(3.24,0.39)$ & $(3.33,0.40)$ & $(16.33,1.70)$ & $(16.69,1.52)$ & $(3.16,0.42)$ & $(3.27,0.40)$ & $(18.30,3.61)$ & $(18.72,3.24)$ \\ 
 & 4 & $(7.38,1.03)$ & $(7.56,1.32)$ & $(47.52,6.91)$ & $(49.82,9.18)$ & $(4.00,0.31)$ & $(4.07,0.35)$ & $(24.02,2.34)$ & $(24.75,3.40)$ & $(4.04,0.33)$ & $(4.13,0.38)$ & $(19.53,1.30)$ & $(19.89,1.65)$ & $(3.95,0.33)$ & $(3.98,0.41)$ & $(25.24,3.02)$ & $(25.23,3.76)$ \\ 
 & 10 & $(9.76,1.03)$ & $(10.40,1.71)$ & $(62.93,6.28)$ & $(69.96,9.63)$ & $(4.75,0.29)$ & $(4.76,0.33)$ & $(29.61,2.13)$ & $(31.50,3.54)$ & $(4.86,0.32)$ & $(4.86,0.36)$ & $(22.77,1.18)$ & $(23.40,1.70)$ & $(4.78,0.32)$ & $(4.83,0.41)$ & $(32.12,2.80)$ & $(32.94,3.51)$ \\
\specialrule{0.5mm}{0pt}{0pt}
\multirow{3}{*}{500} 
& 1 &  &  &  &  & $(2.16,0.17)$ & $(2.15,0.16)$ & $(7.56,0.62)$ & $(7.62,0.57)$ & $(1.92,0.13)$ & $(1.91,0.13)$ & $(6.12,0.42)$ & $(6.22,0.43)$ & $(1.81,0.11)$ & $(1.84,0.11)$ & $(6.15,0.90)$ & $(6.32,1.02)$ \\ 
&  4 &  &  &  &  & $(2.49,0.14)$ & $(2.47,0.16)$ & $(8.76,0.50)$ & $(8.87,0.61)$ & $(2.16,0.10)$ & $(2.14,0.10)$ & $(6.96,0.33)$ & $(7.00,0.40)$ & $(2.03,0.09)$ & $(2.03,0.10)$ & $(8.03,0.81)$ & $(8.12,1.14)$ \\ 
&  10 &  &  &  &  & $(2.83,0.12)$ & $(2.79,0.15)$ & $(9.99,0.45)$ & $(10.17,0.63)$ & $(2.40,0.09)$ & $(2.36,0.10)$ & $(7.77,0.30)$ & $(7.99,0.43)$ & $(2.24,0.08)$ & $(2.24,0.10)$ & $(9.99,0.72)$ & $(10.69,1.41)$ \\ 
\specialrule{0.5mm}{0pt}{0pt}
\multirow{3}{*}{1000} 
& 1 &  &  &  &  &  &  &  &  & $(1.72,0.10)$ & $(1.71,0.09)$ & $(5.28,0.87)$ & $(5.34,0.84)$ & $(1.70,0.17)$ & $(1.70,0.16)$ & $(3.94,0.21)$ & $(3.92,0.21)$ \\ 
  & 4 &  &  &  &  &  &  &  & & $(1.91,0.08)$ & $(1.89,0.08)$ & $(7.05,0.78)$ & $(7.22,1.50)$ & $(2.03,0.15)$ & $(2.01,0.21)$ & $(4.36,0.17)$ & $(4.41,0.24)$ \\ 
  & 10 &  &  &  &  &  &  &  &  & $(2.09,0.07)$ & $(2.07,0.08)$ & $(8.80,0.71)$ & $(9.92,1.83)$ & $(2.37,0.14)$ & $(2.42,0.24)$ & $(4.77,0.15)$ & $(4.91,0.29)$ \\ 
  \specialrule{0.5mm}{0pt}{0pt}
\multirow{3}{*}{2000} 
& 1 &  &  &  &  &  &  &  &  &  &  &  &  & $(1.46,0.06)$ & $(1.46,0.05)$ & $(3.17,0.13)$ & $(3.18,0.12)$ \\ 
  & 4 &  &  &  &  &  &  &  &  &  &  &  &  & $(1.57,0.04)$ & $(1.56,0.05)$ & $(3.43,0.10)$ & $(3.43,0.10)$ \\ 
  & 10 &  &  &  &  &  &  &  &  &  &  &  &  & $(1.68,0.04)$ & $(1.66,0.05)$ & $(3.66,0.09)$ & $(3.67,0.10)$ \\
  \hline \hline 
\end{tabular}\\ 
\rotatebox{90}{\hspace{-2.3cm}3. ARMA \scriptsize $\big((0.6,0.2),(0.5,0.2)\big)$. }& \tiny

\begin{tabular}{c @{\,\vrule width 1.5mm\,}  c @{\,\vrule width 1.5mm\,} c c c c @{\,\vrule width 0.5mm\,}  c c c c @{\,\vrule width 0.5mm\,} c c c c @{\,\vrule width 0.5mm\,} c c c c }
\backslashbox{$n_Y$}{$n_X$} & & \multicolumn{4}{c}{100} & \multicolumn{4}{c}{500} & \multicolumn{4}{c}{1000} & \multicolumn{4}{c}{2000}\\
\specialrule{1.5mm}{0pt}{0pt}
 & \backslashbox{$k$}{$m$} & \multicolumn{2}{c}{100}  & \multicolumn{2}{c}{1000}  &  \multicolumn{2}{c}{100}  & \multicolumn{2}{c}{1000} & \multicolumn{2}{c}{100}  & \multicolumn{2}{c}{1000} & \multicolumn{2}{c}{100}  & \multicolumn{2}{c}{1000} \\
 \specialrule{1.5mm}{0pt}{0pt}
 & &  $(\hat{\mu},\hat{\sigma})$  & $(\mu,\sigma)$ &  $(\hat{\mu},\hat{\sigma})$  & $(\mu,\sigma)$ & $(\hat{\mu},\hat{\sigma})$  & $(\mu,\sigma)$ & $(\hat{\mu},\hat{\sigma})$  & $(\mu,\sigma)$ & $(\hat{\mu},\hat{\sigma})$  & $(\mu,\sigma)$ & $(\hat{\mu},\hat{\sigma})$  & $(\mu,\sigma)$ & $(\hat{\mu},\hat{\sigma})$  & $(\mu,\sigma)$ & $(\hat{\mu},\hat{\sigma})$  & $(\mu,\sigma)$ \\
 \specialrule{1.5mm}{0pt}{0pt}
\multirow{3}{*}{100} 
&1  & $(28.96,12.69)$ & $ (36.79,16.64)$ & $ (251.42,105.05)$ & $ (333.04,156.12)$ & $ (19.48,7.79)$ & $ (25.23,9.14)$ & $ (154.16,50.52)$ & $ (187.38,61.43)$ & $ (15.59,5.62)$ & $ (18.96,5.99)$ & $ (139.57,46.00)$ & $ (170.16,49.33)$ & $ (15.10,5.05)$ & $ (18.11,5.88)$ & $ (132.97,47.80)$ & $ (163.83,51.00)$ \\ 
  & 4 & $(54.47,13.21)$ & $ (65.76,25.01)$ & $ (466.97,110.66)$ & $ (570.47,218.52)$ & $ (35.83,7.83)$ & $ (41.70,8.71)$ & $ (254.86,48.94)$ & $ (290.63,51.26)$ & $ (27.23,5.62)$ & $ (30.07,5.66)$ & $ (231.95,44.43)$ & $ (261.05,38.51)$ & $ (25.11,4.72)$ & $ (28.40,4.30)$ & $ (227.38,43.68)$ & $ (251.50,37.94)$ \\ 
  & 10 & $(84.26,13.36)$ & $ (92.84,28.21)$ & $ (703.57,108.72)$ & $ (791.92,255.82)$ & $ (51.87,7.68)$ & $ (55.35,8.47)$ & $ (369.98,50.60)$ & $ (388.93,55.27)$ & $ (39.91,6.00)$ & $ (40.14,4.46)$ & $ (332.81,44.09)$ & $ (331.52,33.62)$ & $ (36.21,4.74)$ & $ (34.65,2.84)$ & $ (325.40,43.65)$ & $ (313.54,24.40)$ \\ 
\specialrule{0.5mm}{0pt}{0pt}
\multirow{3}{*}{500} 
& 1 &  &  &  &  & $(7.24,1.55)$ & $ (7.70,1.67)$ & $ (55.71,11.52)$ & $ (58.73,12.43)$ & $ (5.86,1.06)$ & $ (6.15,1.21)$ & $ (42.51,7.07)$ & $ (45.19,8.13)$ & $ (5.29,0.96)$ & $ (5.57,1.01)$ & $ (36.00,5.82)$ & $ (36.68,5.54)$ \\ 
  & 4 &  &  &  &  & $(10.36,1.37)$ & $ (10.66,1.68)$ & $ (78.46,9.99)$ & $ (82.43,13.21)$ & $ (8.09,0.97)$ & $ (8.38,1.13)$ & $ (57.47,6.27)$ & $ (59.21,7.58)$ & $ (7.28,0.85)$ & $ (7.40,0.96)$ & $ (46.98,4.82)$ & $ (48.43,4.76)$ \\ 
  & 10 &  &  &  &  & $(13.38,1.31)$ & $ (14.76,1.97)$ & $ (100.86,8.87)$ & $ (112.61,16.04)$ & $ (10.30,0.93)$ & $ (10.47,1.20)$ & $ (71.75,5.73)$ & $ (74.17,6.73)$ & $ (9.25,0.82)$ & $ (9.26,0.90)$ & $ (58.59,4.59)$ & $ (58.84,4.30)$ \\
\specialrule{0.5mm}{0pt}{0pt}
\multirow{3}{*}{1000} 
& 1  &  &  &  &  &  &  &  &  & $(4.46,0.66)$ & $ (4.52,0.73)$ & $ (28.97,3.98)$ & $ (29.43,4.26)$ & $ (3.83,0.48)$ & $ (3.93,0.54)$ & $ (22.44,2.76)$ & $ (23.08,2.87)$ \\ 
  & 4 &  &  &  &  &  &  &  &  & $(5.77,0.57)$ & $ (5.82,0.61)$ & $ (36.95,3.37)$ & $ (37.93,3.71)$ & $ (4.79,0.42)$ & $ (4.83,0.47)$ & $ (28.01,2.35)$ & $ (28.29,2.38)$ \\ 
  & 10 &  &  &  &  &  &  &  &  & $(7.10,0.54)$ & $ (7.34,0.73)$ & $ (44.83,3.18)$ & $ (47.80,4.46)$ & $ (5.81,0.39)$ & $ (5.84,0.46)$ & $ (33.31,2.08)$ & $ (33.98,2.39)$ \\
\specialrule{0.5mm}{0pt}{0pt}
\multirow{3}{*}{2000} 
& 1 &  &  &  &  &  &  &  &  &  &  &  &  & \color{red} $(3.09,0.37)$ & \color{red} $ (3.17,0.34)$ & \color{red} $ (15.35,1.61)$ & \color{red} $ (15.58,1.60)$ \\ 
  & 4 &  &  &  &  &  &  &  &  &  &  &  &  &\color{red} $(3.78,0.29)$ &\color{red} $ (3.82,0.34)$ &\color{red} $ (18.45,1.25)$ &\color{red} $ (18.69,1.43)$ \\ 
  & 10 &  &  &  &  &  &  &  &  &  &  &  &  &\color{red} $(4.51,0.26)$ &\color{red} $ (4.62,0.36)$ &\color{red} $ (21.45,1.13)$ &\color{red} $ (22.16,1.49)$ \\ 
\end{tabular}
\normalsize 
\end{tabular}}
\end{adjustbox}
\end{center}
\end{table}

\begin{table}
\begin{center}
\begin{adjustbox}{addcode={\begin{minipage}{\width}}{\caption{
Simulations of the minimum residual spikes. The values $(\hat{\mu},\hat{\sigma})$ and $(\mu,\sigma)$ are the estimates of the mean and the standard error of the residual spikes obtained by the Main Theorem \ref{TH=Main} and empirical methods using 500 replicates and $\theta_i=5000$, respectively. } \label{tab=simulationmainth2}\end{minipage}},rotate=90,center}
\scalebox{0.55}{
\begin{tabular}{c c}
&$\lambda_{\min}\left(\hat{\hat{\Sigma}}_{X,P_k}^{-1/2} \hat{\hat{\Sigma}}_{Y,P_k} \hat{\hat{\Sigma}}_{X,P_k}^{-1/2}\right)$ \\ 
 \rotatebox{90}{\hspace{-2.2cm} 1. Normal entries. }&\tiny 
\begin{tabular}{c @{\,\vrule width 1.5mm\,}  c @{\,\vrule width 1.5mm\,} c c c c @{\,\vrule width 0.5mm\,}  c c c c @{\,\vrule width 0.5mm\,} c c c c @{\,\vrule width 0.5mm\,} c c c c }
\backslashbox{$n_Y$}{$n_X$} & & \multicolumn{4}{c}{100} & \multicolumn{4}{c}{500} & \multicolumn{4}{c}{1000} & \multicolumn{4}{c}{2000}\\
\specialrule{1.5mm}{0pt}{0pt}
 & \backslashbox{$k$}{$m$} & \multicolumn{2}{c}{100}  & \multicolumn{2}{c}{1000}  &  \multicolumn{2}{c}{100}  & \multicolumn{2}{c}{1000} & \multicolumn{2}{c}{100}  & \multicolumn{2}{c}{1000} & \multicolumn{2}{c}{100}  & \multicolumn{2}{c}{1000} \\
 \specialrule{1.5mm}{0pt}{0pt}
 & &  $(\hat{\mu},\hat{\sigma})$  & $(\mu,\sigma)$ &  $(\hat{\mu},\hat{\sigma})$  & $(\mu,\sigma)$ & $(\hat{\mu},\hat{\sigma})$  & $(\mu,\sigma)$ & $(\hat{\mu},\hat{\sigma})$  & $(\mu,\sigma)$ & $(\hat{\mu},\hat{\sigma})$  & $(\mu,\sigma)$ & $(\hat{\mu},\hat{\sigma})$  & $(\mu,\sigma)$ & $(\hat{\mu},\hat{\sigma})$  & $(\mu,\sigma)$ & $(\hat{\mu},\hat{\sigma})$  & $(\mu,\sigma)$ \\
 \specialrule{1.5mm}{0pt}{0pt}
\multirow{3}{*}{100} &
1 & $(0.269,0.035)$ & $(0.278,0.035)$ & $(0.045,0.005)$ & $(0.045,0.005)$ & $(0.347,0.031)$ & $(0.348,0.042)$ & $(0.072,0.003)$ & $(0.070,0.009)$ & $(0.363,0.028)$ & $(0.353,0.043)$ & $(0.077,0.003)$ & $(0.076,0.010)$ & $(0.368,0.031)$ & $(0.362,0.042)$ & $(0.080,0.003)$ & $(0.079,0.011)$ \\ 
&  4 & $(0.201,0.026)$ & $(0.222,0.022)$ & $(0.036,0.004)$ & $(0.036,0.003)$ & $(0.271,0.027)$ & $(0.272,0.028)$ & $(0.056,0.005)$ & $(0.055,0.005)$ & $(0.286,0.027)$ & $(0.276,0.029)$ & $(0.060,0.005)$ & $(0.058,0.006)$ & $(0.289,0.029)$ & $(0.283,0.030)$ & $(0.063,0.006)$ & $(0.060,0.006)$ \\ 
 & 10 & $(0.138,0.022)$ & $(0.170,0.016)$ & $(0.027,0.003)$ & $(0.028,0.003)$ & $(0.192,0.026)$ & $(0.207,0.021)$ & $(0.039,0.005)$ & $(0.041,0.004)$ & $(0.204,0.025)$ & $(0.215,0.021)$ & $(0.042,0.005)$ & $(0.042,0.004)$ & $(0.204,0.027)$ & $(0.214,0.021)$ & $(0.043,0.006)$ & $(0.044,0.005)$ \\
  \specialrule{1.5mm}{0pt}{0pt}
\multirow{3}{*}{500}  &
  1 & & &  &  & $(0.535,0.034)$ & $(0.544,0.035)$ & $(0.171,0.009)$ & $(0.171,0.010)$ & $(0.579,0.030)$ & $(0.579,0.033)$ & $(0.208,0.009)$ & $(0.208,0.011)$ & $(0.612,0.028)$ & $(0.619,0.030)$ & $(0.234,0.008)$ & $(0.233,0.013)$ \\ 
 & 4 &  &  &  &  & $(0.471,0.026)$ & $(0.484,0.021)$ & $(0.154,0.007)$ & $(0.154,0.007)$ & $(0.519,0.025)$ & $(0.519,0.023)$ & $(0.188,0.008)$ & $(0.187,0.007)$ & $(0.554,0.023)$ & $(0.566,0.023)$ & $(0.212,0.007)$ & $(0.210,0.009)$ \\ 
 & 10 &  &  &  &  & $(0.409,0.023)$ & $(0.436,0.018)$ & $(0.136,0.006)$ & $(0.138,0.005)$ & $(0.456,0.022)$ & $(0.469,0.020)$ & $(0.166,0.007)$ & $(0.168,0.007)$ & $(0.495,0.021)$ & $(0.517,0.018)$ & $(0.187,0.007)$ & $(0.187,0.007)$ \\ 
   \specialrule{1.5mm}{0pt}{0pt}
\multirow{3}{*}{1000}  &
  1 &  &  &  &  &  &  &  &  & $(0.637,0.028)$ & $(0.641,0.028)$ & $(0.269,0.011)$ & $(0.268,0.010)$ & $(0.676,0.025)$ & $(0.682,0.026)$ & $(0.314,0.010)$ & $(0.313,0.012)$ \\ 
 & 4 & & &  &  &  &  &  &  & $(0.582,0.022)$ & $(0.587,0.019)$ & $(0.247,0.008)$ & $(0.248,0.008)$ & $(0.625,0.020)$ & $(0.635,0.021)$ & $(0.291,0.008)$ & $(0.289,0.009)$ \\ 
 & 10 &  &  &  &  &  &  &  &  & $(0.528,0.020)$ & $(0.547,0.017)$ & $(0.227,0.007)$ & $(0.230,0.006)$ & $(0.573,0.018)$ & $(0.594,0.017)$ & $(0.268,0.007)$ & $(0.269,0.007)$ \\
   \specialrule{1.5mm}{0pt}{0pt}
\multirow{3}{*}{2000}  & 
  1 &  &  &  &  &  &  &  &  &  &  &  &  & $(0.730,0.023)$ & $(0.729,0.021)$ & $(0.382,0.011)$ & $(0.383,0.011)$ \\ 
 & 4 &  &  &  &  &  &  &  &  &  &  &  &  & $(0.684,0.018)$ & $(0.689,0.016)$ & $(0.360,0.009)$ & $(0.362,0.008)$ \\ 
 & 10 &  &  &  &  &  &  &  &  &  &  &  &  & $(0.641,0.016)$ & $(0.655,0.013)$ & $(0.339,0.007)$ & $(0.342,0.007)$ \\
 \hline  \hline
\end{tabular}\\   
 
\rotatebox{90}{\hspace{-2.2cm} 2. Multivariate Student. } &\tiny \begin{tabular}{c @{\,\vrule width 1.5mm\,}  c @{\,\vrule width 1.5mm\,} c c c c @{\,\vrule width 0.5mm\,}  c c c c @{\,\vrule width 0.5mm\,} c c c c @{\,\vrule width 0.5mm\,} c c c c }
\backslashbox{$n_Y$}{$n_X$} & & \multicolumn{4}{c}{100} & \multicolumn{4}{c}{500} & \multicolumn{4}{c}{1000} & \multicolumn{4}{c}{2000}\\
\specialrule{1.5mm}{0pt}{0pt}
 & \backslashbox{$k$}{$m$} & \multicolumn{2}{c}{100}  & \multicolumn{2}{c}{1000}  &  \multicolumn{2}{c}{100}  & \multicolumn{2}{c}{1000} & \multicolumn{2}{c}{100}  & \multicolumn{2}{c}{1000} & \multicolumn{2}{c}{100}  & \multicolumn{2}{c}{1000} \\
 \specialrule{1.5mm}{0pt}{0pt}
 & &  $(\hat{\mu},\hat{\sigma})$  & $(\mu,\sigma)$ &  $(\hat{\mu},\hat{\sigma})$  & $(\mu,\sigma)$ & $(\hat{\mu},\hat{\sigma})$  & $(\mu,\sigma)$ & $(\hat{\mu},\hat{\sigma})$  & $(\mu,\sigma)$ & $(\hat{\mu},\hat{\sigma})$  & $(\mu,\sigma)$ & $(\hat{\mu},\hat{\sigma})$  & $(\mu,\sigma)$ & $(\hat{\mu},\hat{\sigma})$  & $(\mu,\sigma)$ & $(\hat{\mu},\hat{\sigma})$  & $(\mu,\sigma)$ \\
 \specialrule{1.5mm}{0pt}{0pt}
\multirow{3}{*}{100} 
&1  & $(0.203,0.047)$ & $(0.217,0.034)$ & $(0.031,0.007)$ & $(0.035,0.005)$ & $(0.309,0.035)$ & $(0.311,0.042)$ & $(0.054,0.008)$ & $(0.058,0.009)$ & $(0.307,0.038)$ & $(0.308,0.042)$ & $(0.062,0.006)$ & $(0.056,0.008)$ & $(0.318,0.040)$ & $(0.299,0.040)$ & $(0.055,0.011)$ & $(0.054,0.009)$ \\ 
 & 4 & $(0.120,0.034)$ & $(0.162,0.021)$ & $(0.019,0.005)$ & $(0.026,0.004)$ & $(0.229,0.029)$ & $(0.242,0.026)$ & $(0.038,0.006)$ & $(0.043,0.005)$ & $(0.226,0.032)$ & $(0.235,0.026)$ & $(0.046,0.005)$ & $(0.043,0.005)$ & $(0.232,0.031)$ & $(0.233,0.025)$ & $(0.036,0.007)$ & $(0.039,0.005)$ \\ 
 & 10 & $(0.049,0.028)$ & $(0.126,0.014)$ & $(0.008,0.004)$ & $(0.020,0.002)$ & $(0.152,0.026)$ & $(0.185,0.021)$ & $(0.022,0.005)$ & $(0.033,0.004)$ & $(0.148,0.027)$ & $(0.181,0.019)$ & $(0.030,0.005)$ & $(0.032,0.003)$ & $(0.151,0.029)$ & $(0.177,0.019)$ & $(0.019,0.006)$ & $(0.029,0.004)$ \\
\specialrule{0.5mm}{0pt}{0pt}
\multirow{3}{*}{500} 
& 1 &   &   &   &   & $(0.465,0.036)$ & $(0.440,0.039)$ & $(0.132,0.011)$ & $(0.132,0.011)$ & $(0.522,0.034)$ & $(0.535,0.039)$ & $(0.163,0.012)$ & $(0.147,0.011)$ & $(0.550,0.032)$ & $(0.515,0.032)$ & $(0.161,0.025)$ & $(0.160,0.017)$ \\ 
 & 4  &   &   &   &   & $(0.386,0.030)$ & $(0.375,0.025)$ & $(0.110,0.008)$ & $(0.111,0.009)$ & $(0.446,0.029)$ & $(0.464,0.032)$ & $(0.141,0.009)$ & $(0.129,0.007)$ & $(0.484,0.026)$ & $(0.458,0.024)$ & $(0.119,0.018)$ & $(0.131,0.014)$ \\ 
 & 10 &   &   &   &   & $(0.312,0.027)$ & $(0.322,0.021)$ & $(0.088,0.007)$ & $(0.093,0.008)$ & $(0.368,0.026)$ & $(0.404,0.029)$ & $(0.120,0.007)$ & $(0.114,0.005)$ & $(0.419,0.023)$ & $(0.409,0.020)$ & $(0.082,0.015)$ & $(0.108,0.009)$ \\
\specialrule{0.5mm}{0pt}{0pt}
\multirow{3}{*}{1000} 
& 1 &   &   &   &   &   &   &   &   & $(0.583,0.033)$ & $(0.562,0.034)$ & $(0.190,0.031)$ & $(0.196,0.022)$ & $(0.587,0.059)$ & $(0.628,0.036)$ & $(0.255,0.014)$ & $(0.249,0.013)$\\ 
 & 4 &   &   &   &   &   &   &   &   & $(0.517,0.026)$ & $(0.504,0.026)$ & $(0.134,0.022)$ & $(0.159,0.021)$ & $(0.494,0.045)$ & $(0.560,0.032)$ & $(0.228,0.011)$ & $(0.225,0.009)$ \\ 
 & 10 &   &   &   &   &   &   &   &   & $(0.450,0.023)$ & $(0.452,0.023)$ & $(0.085,0.018)$ & $(0.126,0.015)$ & $(0.409,0.037)$ & $(0.515,0.023)$ & $(0.203,0.008)$ & $(0.204,0.009)$ \\
\specialrule{0.5mm}{0pt}{0pt}
\multirow{3}{*}{2000} 
& 1 &   &   &   &   &   &   &   &   &   &   &   &   & $(0.684,0.026)$ & $(0.676,0.027)$ & $(0.315,0.013)$ & $(0.310,0.013)$ \\ 
 & 4 &   &   &   &   &   &   &   &   &   &   &   &   & $(0.631,0.021)$ & $(0.631,0.020)$ & $(0.289,0.010)$ & $(0.285,0.010)$ \\ 
 & 10  &   &   &   &   &   &   &   &   &   &   &   &   & $(0.580,0.019)$ & $(0.589,0.019)$ & $(0.264,0.008)$ & $(0.262,0.009)$ \\
 \hline \hline
\end{tabular}\\  
\rotatebox{90}{\hspace{-2.3cm}3. ARMA \scriptsize$\big((0.6,0.2),(0.5,0.2)\big)$. }& \tiny

\begin{tabular}{c @{\,\vrule width 1.5mm\,}  c @{\,\vrule width 1.5mm\,} c c c c @{\,\vrule width 0.5mm\,}  c c c c @{\,\vrule width 0.5mm\,} c c c c @{\,\vrule width 0.5mm\,} c c c c }
\backslashbox{$n_Y$}{$n_X$} & & \multicolumn{4}{c}{100} & \multicolumn{4}{c}{500} & \multicolumn{4}{c}{1000} & \multicolumn{4}{c}{2000}\\
\specialrule{1.5mm}{0pt}{0pt}
 & \backslashbox{$k$}{$m$} & \multicolumn{2}{c}{100}  & \multicolumn{2}{c}{1000}  &  \multicolumn{2}{c}{100}  & \multicolumn{2}{c}{1000} & \multicolumn{2}{c}{100}  & \multicolumn{2}{c}{1000} & \multicolumn{2}{c}{100}  & \multicolumn{2}{c}{1000} \\
 \specialrule{1.5mm}{0pt}{0pt}
 & &  $(\hat{\mu},\hat{\sigma})$  & $(\mu,\sigma)$ &  $(\hat{\mu},\hat{\sigma})$  & $(\mu,\sigma)$ & $(\hat{\mu},\hat{\sigma})$  & $(\mu,\sigma)$ & $(\hat{\mu},\hat{\sigma})$  & $(\mu,\sigma)$ & $(\hat{\mu},\hat{\sigma})$  & $(\mu,\sigma)$ & $(\hat{\mu},\hat{\sigma})$  & $(\mu,\sigma)$ & $(\hat{\mu},\hat{\sigma})$  & $(\mu,\sigma)$ & $(\hat{\mu},\hat{\sigma})$  & $(\mu,\sigma)$ \\
 \specialrule{1.5mm}{0pt}{0pt}
\multirow{3}{*}{100} 
&1  & $(0.035,0.015)$ & $(0.033,0.017)$ & $(0.004,0.002)$ & $(0.004,0.002)$ & $(0.052,0.020)$ & $(0.041,0.017)$ & $(0.007,0.002)$ & $(0.005,0.002)$ & $(0.065,0.025)$ & $(0.054,0.020)$ & $(0.007,0.002)$ & $(0.006,0.002)$ & $(0.067,0.022)$ & $(0.057,0.021)$ & $(0.007,0.003)$ & $(0.006,0.002)$ \\ 
  & 4 & $(0.005,0.012)$ & $(0.016,0.005)$ & $(0.001,0.001)$ & $(0.002,0.001)$ & $(0.012,0.015)$ & $(0.023,0.007)$ & $(0.002,0.002)$ & $(0.003,0.001)$ & $(0.017,0.019)$ & $(0.030,0.008)$ & $(0.002,0.002)$ & $(0.003,0.001)$ & $(0.021,0.018)$ & $(0.030,0.008)$ & $(0.002,0.002)$ & $(0.003,0.001)$ \\ 
  & 10 & $(-0.025,0.010)$ & $(0.011,0.003)$ & $(-0.002,0.001)$ & $(0.001,0.000)$ & $(-0.026,0.012)$ & $(0.016,0.004)$ & $(-0.002,0.001)$ & $(0.002,0.001)$ & $(-0.027,0.016)$ & $(0.020,0.005)$ & $(-0.003,0.002)$ & $(0.002,0.001)$ & $(-0.024,0.014)$ & $(0.021,0.005)$ & $(-0.003,0.002)$ & $(0.003,0.001)$ \\ 
\specialrule{0.5mm}{0pt}{0pt}
\multirow{3}{*}{500} 
& 1 & & & & & $(0.137,0.030)$ & $(0.136,0.027)$ & $(0.018,0.004)$ & $(0.018,0.004)$ & $(0.170,0.032)$ & $(0.164,0.032)$ & $(0.023,0.004)$ & $(0.023,0.004)$ & $(0.189,0.035)$ & $(0.184,0.034)$ & $(0.028,0.005)$ & $(0.028,0.005)$ \\ 
  & 4 & & & & & $(0.082,0.022)$ & $(0.095,0.014)$ & $(0.011,0.003)$ & $(0.013,0.002)$ & $(0.106,0.025)$ & $(0.119,0.018)$ & $(0.016,0.003)$ & $(0.017,0.002)$ & $(0.121,0.026)$ & $(0.131,0.017)$ & $(0.019,0.003)$ & $(0.020,0.003)$ \\ 
 & 10 & & & & & $(0.027,0.019)$ & $(0.070,0.009)$ & $(0.004,0.002)$ & $(0.009,0.001)$ & $(0.047,0.021)$ & $(0.089,0.012)$ & $(0.008,0.003)$ & $(0.012,0.002)$ & $(0.057,0.023)$ & $(0.098,0.012)$ & $(0.010,0.003)$ & $(0.015,0.002)$ \\ 
\specialrule{0.5mm}{0pt}{0pt}
\multirow{3}{*}{1000} 
& 1 & & & & & & & & & $(0.225,0.034)$ & $(0.223,0.032)$ & $(0.035,0.005)$ & $(0.035,0.005)$ & $(0.265,0.035)$ & $(0.258,0.033)$ & $(0.045,0.006)$ & $(0.044,0.006)$ \\ 
  & 4 & & & & & & & & & $(0.158,0.026)$ & $(0.173,0.019)$ & $(0.025,0.004)$ & $(0.027,0.003)$ & $(0.194,0.026)$ & $(0.203,0.022)$ & $(0.033,0.004)$ & $(0.035,0.003)$ \\ 
  & 10 & & & & & & & & & $(0.093,0.023)$ & $(0.135,0.013)$ & $(0.016,0.003)$ & $(0.021,0.002)$ & $(0.123,0.023)$ & $(0.162,0.014)$ & $(0.023,0.004)$ & $(0.028,0.002)$ \\ 

\specialrule{0.5mm}{0pt}{0pt}
\multirow{3}{*}{2000} 
& 1 & & & & & & & & & & & & & $(0.322,0.036)$ & $(0.325,0.037)$ & $(0.065,0.007)$ & $(0.065,0.007)$ \\ 
  & 4 & & & & & & & & & & & & & $(0.250,0.028)$ & $(0.263,0.023)$ & $(0.052,0.005)$ & $(0.053,0.004)$ \\ 
  & 10 & & & & & & & & & & & & & $(0.180,0.025)$ & $(0.220,0.017)$ & $(0.040,0.004)$ & $(0.046,0.003)$ 
\end{tabular}
\normalsize 
\end{tabular}
}

\end{adjustbox}
\end{center}
\end{table}

\section{Estimation of $k$}\label{appendix:Simulationkest}
In this section we show in Table \ref{tab=simulationkestimation} that a small underestimation or overestimation of $k$ does not affect the estimation of residual spikes. \\

We assume the form of $W_X=\frac{1}{n}\X \X^t$ and $W_Y=\frac{1}{n}\Y \Y^t$ where the entries of $\X$ and $\Y$ are i.i.d. Normal. Then we apply a perturbation $P=I_m+\sum_{i=1}^k (\theta_i-1)u_i u_i^t$ to $W_X$ and $W_Y$ to create 
$$\hat{\Sigma}_X=P^{1/2} W_X P^{1/2} \text{ and }\hat{\Sigma}_Y=P^{1/2} W_Y P^{1/2}.$$
Recall that in the simulations of Section \ref{appendix:Tablecomplete}, we assumed the spectra of $W_X$ and $W_Y$ are known. In this appendix, the simulations estimate the spectra parameter $M_{X,s}$ and $M_{Y,s}$ by $\hat{M}_{X,s}$ and $\hat{M}_{Y,s}$ using the observed spectra of $\hat{\Sigma}_X$ and $\hat{\Sigma}_X$,
\begin{eqnarray*}
\hat{M}_{X,s}=\frac{1}{m-k}\sum_{i=k+1}^m \lambda_{\hat{\Sigma}_X}^s \text{ and }\hat{M}_{Y,s}=\frac{1}{m-k}\sum_{i=k+1}^m \lambda_{\hat{\Sigma}_Y}^s.
\end{eqnarray*}
In this case it seems that, 
\begin{center}
\textbf{Faulty values of $k$ lead to conservative procedures in all cases. However, underestimation of $k$ can lead to a large loss of power!}
\end{center}
Conservative procedure are obtained when underestimating the minimum and overestimating the maximum residual spike.

\paragraph*{Detail of the simulations of Table \ref{tab=simulationkestimation}}
We apply a perturbation of order $k=4$ to normal data. \\
Then we use the procedure with different $k_{est}=1,2,3,4,5,6$.\\
The moments are estimated using the usual estimators of the spectra assuming $k=k_{est}$ and $n=1000$ replicates of the experiment.\\
The correct perturbation, $P_4$, that is applied to the data has eigenvalues, $\theta_1=1000,$ $\theta_2=200,$ $\theta_3=16,$ $\theta_4=2.1.$

\begin{table}
\begin{center}
\scalebox{0.65}{
\begin{adjustbox}{addcode={\begin{minipage}{\width}}{\caption{Residual spike moment for normal entries, $k=4$ and $\theta_1=1000,$ $\theta_2=200$, $\theta_3=16$, $\theta_4=2.1$. Then, the value $k$ is estimated by $k_{est}$. (Number of replicates: 1000)} \label{tab=simulationkestimation}
\end{minipage}},rotate=90,center}
\begin{tabular}{c}
$\lambda_{\max}\left(\hat{\hat{\Sigma}}_{X,P_k}^{-1/2} \hat{\hat{\Sigma}}_{Y,P_k} \hat{\hat{\Sigma}}_{X,P_k}^{-1/2}\right)$ \\ 
 \tiny 
\begin{tabular}{c @{\,\vrule width 1.5mm\,}  c @{\,\vrule width 1.5mm\,} c c c c @{\,\vrule width 0.5mm\,} c c c c @{\,\vrule width 0.5mm\,} c c c c  }
\backslashbox{$n_Y$}{$n_X$} & & \multicolumn{4}{c}{200} & \multicolumn{4}{c}{400} & \multicolumn{4}{c}{800} \\
\specialrule{1.5mm}{0pt}{0pt}
 & \backslashbox{$k_{est}$}{$m$} & \multicolumn{2}{c}{200}  & \multicolumn{2}{c}{400}  &  \multicolumn{2}{c}{200}  &  \multicolumn{2}{c}{400}   &  \multicolumn{2}{c}{200} & \multicolumn{2}{c}{400}    \\
 \specialrule{1.5mm}{0pt}{0pt}
 & &    $(\hat{\mu},\hat{\sigma})$  & $(\mu,\sigma)$ & $(\hat{\mu},\hat{\sigma})$  & $(\mu,\sigma)$ &  $(\hat{\mu},\hat{\sigma})$  & $(\mu,\sigma)$ &  $(\hat{\mu},\hat{\sigma})$  & $(\mu,\sigma)$  &  $(\hat{\mu},\hat{\sigma})$  & $(\mu,\sigma)$  &  $(\hat{\mu},\hat{\sigma})$  & $(\mu,\sigma)$ \\
 \specialrule{1.5mm}{0pt}{0pt}
\multirow{5}{*}{200} 
& 1 & $\left(88.78,55.48\right)$ & $\left(3.86,1.92\right)$ & $\left(85.85,57.82\right)$ & $\left(6.23,2.84\right)$ & $\left(87.34,55.85\right)$ & $\left(3.44,1.54\right)$ & $\left(106.76,67.23\right)$ & $\left(5.04,1.98\right)$ & $\left(94.28,58.81\right)$ & $\left(3.17,1.40\right)$ & $\left(91.32,58.26\right)$ & $\left(4.51,1.83\right)$ \\ 
 & 2 & $\left(7.86,2.12\right)$ & $\left(4.11,0.42\right)$ & $\left(8.15,1.23\right)$ & $\left(6.40,0.57\right)$ & $\left(7.30,2.14\right)$ & $\left(3.43,0.30\right)$ & $\left(6.60,0.93\right)$ & $\left(5.16,0.37\right)$ & $\left(6.63,1.88\right)$ & $\left(3.13,0.26\right)$ & $\left(6.48,1.07\right)$ & $\left(4.52,0.27\right)$ \\ 
 & 3 & $\left(4.27,0.30\right)$ & $\left(4.23,0.31\right)$ & $\left(6.62,0.43\right)$ & $\left(6.46,0.48\right)$ & $\left(3.55,0.19\right)$ & $\left(3.53,0.22\right)$ & $\left(5.29,0.26\right)$ & $\left(5.19,0.29\right)$ & $\left(3.20,0.15\right)$ & $\left(3.18,0.16\right)$ & $\left(4.65,0.19\right)$ & $\left(4.58,0.21\right)$ \\ 
 & 4 & $\left(4.42,0.28\right)$ & $\left(4.29,0.33\right)$ & $\left(6.90,0.40\right)$ & $\left(6.59,0.47\right)$ & $\left(3.63,0.18\right)$ & $\left(3.57,0.22\right)$ & $\left(5.44,0.24\right)$ & $\left(5.27,0.31\right)$ & $\left(3.25,0.14\right)$ & $\left(3.21,0.16\right)$ & $\left(4.77,0.19\right)$ & $\left(4.64,0.22\right)$ \\ 
 & 5 & $\left(4.53,0.28\right)$ & $\left(4.33,0.32\right)$ & $\left(7.05,0.39\right)$ & $\left(6.55,0.49\right)$ & $\left(3.71,0.17\right)$ & $\left(3.59,0.21\right)$ & $\left(5.55,0.23\right)$ & $\left(5.30,0.30\right)$ & $\left(3.29,0.14\right)$ & $\left(3.21,0.15\right)$ & $\left(4.85,0.18\right)$ & $\left(4.69,0.22\right)$ \\ 
 & 6 & $\left(4.63,0.27\right)$ & $\left(4.34,0.33\right)$ & $\left(7.25,0.39\right)$ & $\left(6.70,0.49\right)$ & $\left(3.77,0.16\right)$ & $\left(3.64,0.22\right)$ & $\left(5.66,0.24\right)$ & $\left(5.36,0.29\right)$ & $\left(3.38,0.14\right)$ & $\left(3.32,0.17\right)$ & $\left(4.90,0.18\right)$ & $\left(4.71,0.22\right)$ \\ 
\specialrule{0.5mm}{0pt}{0pt}
\multirow{5}{*}{400} 
& 1 &  &  &  &  & $\left(89.97,58.12\right)$ & $\left(2.74,1.04\right)$ & $\left(93.15,58.71\right)$ & $\left(3.91,1.40\right)$ & $\left(102.49,64.57\right)$ & $\left(2.42,0.84\right)$ & $\left(85.13,55.71\right)$ & $\left(3.39,1.11\right)$ \\ 
 & 2 &  &  &  &  & $\left(6.71,2.18\right)$ & $\left(2.82,0.22\right)$ & $\left(5.82,1.05\right)$ & $\left(4.00,0.27\right)$ & $\left(6.14,1.96\right)$ & $\left(2.46,0.16\right)$ & $\left(5.19,0.98\right)$ & $\left(3.35,0.18\right)$ \\ 
 & 3 &  &  &  &  & $\left(2.89,0.15\right)$ & $\left(2.86,0.16\right)$ & $\left(4.10,0.20\right)$ & $\left(4.06,0.22\right)$ & $\left(2.54,0.11\right)$ & $\left(2.50,0.11\right)$ & $\left(3.44,0.13\right)$ & $\left(3.41,0.15\right)$ \\ 
 & 4 &  &  &  &  & $\left(2.97,0.14\right)$ & $\left(2.90,0.17\right)$ & $\left(4.21,0.19\right)$ & $\left(4.06,0.21\right)$ & $\left(2.55,0.10\right)$ & $\left(2.51,0.11\right)$ & $\left(3.50,0.12\right)$ & $\left(3.41,0.14\right)$ \\ 
& 5 &  &  &  &  & $\left(3.02,0.14\right)$ & $\left(2.92,0.16\right)$ & $\left(4.29,0.18\right)$ & $\left(4.09,0.23\right)$ & $\left(2.61,0.10\right)$ & $\left(2.54,0.11\right)$ & $\left(3.56,0.12\right)$ & $\left(3.43,0.14\right)$ \\ 
 & 6 &  &  &  &  & $\left(3.06,0.13\right)$ & $\left(2.92,0.16\right)$ & $\left(4.37,0.18\right)$ & $\left(4.12,0.22\right)$ & $\left(2.62,0.09\right)$ & $\left(2.57,0.11\right)$ & $\left(3.59,0.12\right)$ & $\left(3.45,0.14\right)$ \\ 
\specialrule{0.5mm}{0pt}{0pt}
\multirow{5}{*}{800} 
& 1 &  &  &  &  &  &  &  &  & $\left(91.91,56.71\right)$ & $\left(2.06,0.53\right)$ & $\left(90.85,59.77\right)$ & $\left(2.69,0.74\right)$ \\ 
 & 2 &  &  &  &  &  &  &  &  & $\left(5.52,1.85\right)$ & $\left(2.10,0.13\right)$ & $\left(5.00,1.22\right)$ & $\left(2.73,0.14\right)$ \\ 
 & 3 &  &  &  &  &  &  &  &  & $\left(2.17,0.08\right)$ & $\left(2.14,0.08\right)$ & $\left(2.81,0.10\right)$ & $\left(2.79,0.11\right)$ \\ 
 & 4 &  &  &  &  &  &  &  &  & $\left(2.19,0.08\right)$ & $\left(2.17,0.08\right)$ & $\left(2.85,0.10\right)$ & $\left(2.79,0.11\right)$ \\ 
 & 5 &  &  &  &  &  &  &  &  & $\left(2.22,0.08\right)$ & $\left(2.17,0.08\right)$ & $\left(2.90,0.10\right)$ & $\left(2.80,0.11\right)$ \\ 
 & 6 &  &  &  &  &  &  &  &  & $\left(2.24,0.07\right)$ & $\left(2.18,0.09\right)$ & $\left(2.92,0.09\right)$ & $\left(2.80,0.11\right)$ 
\end{tabular}\\ 
$\lambda_{\min}\left(\hat{\hat{\Sigma}}_{X,P_k}^{-1/2} \hat{\hat{\Sigma}}_{Y,P_k} \hat{\hat{\Sigma}}_{X,P_k}^{-1/2}\right)$ \\ 
 \tiny 
\begin{tabular}{c @{\,\vrule width 1.5mm\,}  c @{\,\vrule width 1.5mm\,} c c c c  @{\,\vrule width 0.5mm\,} c c c c  @{\,\vrule width 0.5mm\,} c c c c   }
\backslashbox{$n_Y$}{$n_X$} & & \multicolumn{4}{c}{200} & \multicolumn{4}{c}{400} & \multicolumn{4}{c}{800} \\
\specialrule{1.5mm}{0pt}{0pt}
 & \backslashbox{$k_{est}$}{$m$} & \multicolumn{2}{c}{200}  & \multicolumn{2}{c}{400}  &  \multicolumn{2}{c}{200}  &  \multicolumn{2}{c}{400}   &  \multicolumn{2}{c}{200} & \multicolumn{2}{c}{400}    \\
 \specialrule{1.5mm}{0pt}{0pt}
 & &  $(\hat{\mu},\hat{\sigma})$  & $(\mu,\sigma)$ &  $(\hat{\mu},\hat{\sigma})$  & $(\mu,\sigma)$ & $(\hat{\mu},\hat{\sigma})$  & $(\mu,\sigma)$ & $(\hat{\mu},\hat{\sigma})$  & $(\mu,\sigma)$ & $(\hat{\mu},\hat{\sigma})$  & $(\mu,\sigma)$ & $(\hat{\mu},\hat{\sigma})$  & $(\mu,\sigma)$  \\
 \specialrule{1.5mm}{0pt}{0pt}
\multirow{5}{*}{200} 
& 1 & $\left(0.012,0.008\right)$ & $\left(0.297,0.093\right)$ & $\left(0.011,0.008\right)$ & $\left(0.181,0.053\right)$ & $\left(0.012,0.008\right)$ & $\left(0.333,0.102\right)$ & $\left(0.010,0.007\right)$ & $\left(0.219,0.060\right)$ & $\left(0.011,0.007\right)$ & $\left(0.360,0.105\right)$ & $\left(0.011,0.008\right)$ & $\left(0.246,0.065\right)$ \\ 
 & 2 & $\left(0.113,0.057\right)$ & $\left(0.245,0.024\right)$ & $\left(0.119,0.024\right)$ & $\left(0.159,0.014\right)$ & $\left(0.120,0.066\right)$ & $\left(0.290,0.027\right)$ & $\left(0.146,0.028\right)$ & $\left(0.193,0.016\right)$ & $\left(0.133,0.070\right)$ & $\left(0.318,0.030\right)$ & $\left(0.149,0.036\right)$ & $\left(0.217,0.018\right)$ \\ 
 & 3 & $\left(0.231,0.020\right)$ & $\left(0.239,0.019\right)$ & $\left(0.149,0.011\right)$ & $\left(0.156,0.011\right)$ & $\left(0.274,0.019\right)$ & $\left(0.281,0.021\right)$ & $\left(0.183,0.012\right)$ & $\left(0.189,0.014\right)$ & $\left(0.302,0.019\right)$ & $\left(0.306,0.022\right)$ & $\left(0.206,0.012\right)$ & $\left(0.210,0.015\right)$ \\ 
 & 4 & $\left(0.220,0.018\right)$ & $\left(0.235,0.017\right)$ & $\left(0.141,0.011\right)$ & $\left(0.153,0.011\right)$ & $\left(0.262,0.019\right)$ & $\left(0.279,0.021\right)$ & $\left(0.174,0.011\right)$ & $\left(0.188,0.013\right)$ & $\left(0.291,0.019\right)$ & $\left(0.305,0.021\right)$ & $\left(0.197,0.013\right)$ & $\left(0.210,0.016\right)$ \\ 
 & 5 & $\left(0.211,0.018\right)$ & $\left(0.233,0.017\right)$ & $\left(0.137,0.010\right)$ & $\left(0.153,0.011\right)$ & $\left(0.254,0.018\right)$ & $\left(0.276,0.021\right)$ & $\left(0.168,0.011\right)$ & $\left(0.187,0.014\right)$ & $\left(0.281,0.019\right)$ & $\left(0.306,0.022\right)$ & $\left(0.190,0.012\right)$ & $\left(0.210,0.016\right)$ \\ 
 & 6 & $\left(0.205,0.017\right)$ & $\left(0.230,0.018\right)$ & $\left(0.131,0.010\right)$ & $\left(0.151,0.011\right)$ & $\left(0.244,0.018\right)$ & $\left(0.275,0.020\right)$ & $\left(0.162,0.011\right)$ & $\left(0.186,0.014\right)$ & $\left(0.271,0.018\right)$ & $\left(0.305,0.023\right)$ & $\left(0.183,0.013\right)$ & $\left(0.210,0.015\right)$ \\ 
\specialrule{0.5mm}{0pt}{0pt}
\multirow{5}{*}{400} 
& 1 &  &  &  &  & $\left(0.011,0.008\right)$ & $\left(0.402,0.105\right)$ & $\left(0.011,0.008\right)$ & $\left(0.277,0.066\right)$ & $\left(0.010,0.007\right)$ & $\left(0.448,0.108\right)$ & $\left(0.013,0.009\right)$ & $\left(0.317,0.073\right)$ \\ 
  & 2 &  &  &  &  & $\left(0.130,0.083\right)$ & $\left(0.357,0.028\right)$ & $\left(0.162,0.043\right)$ & $\left(0.251,0.016\right)$ & $\left(0.131,0.085\right)$ & $\left(0.407,0.028\right)$ & $\left(0.183,0.053\right)$ & $\left(0.296,0.018\right)$ \\ 
 & 3 &  &  &  &  & $\left(0.342,0.021\right)$ & $\left(0.349,0.020\right)$ & $\left(0.241,0.013\right)$ & $\left(0.247,0.013\right)$ & $\left(0.389,0.021\right)$ & $\left(0.397,0.021\right)$ & $\left(0.286,0.014\right)$ & $\left(0.291,0.015\right)$ \\ 
 & 4 &  &  &  &  & $\left(0.332,0.020\right)$ & $\left(0.347,0.018\right)$ & $\left(0.235,0.013\right)$ & $\left(0.246,0.013\right)$ & $\left(0.382,0.019\right)$ & $\left(0.395,0.020\right)$ & $\left(0.278,0.013\right)$ & $\left(0.290,0.015\right)$ \\ 
 & 5 &  &  &  &  & $\left(0.323,0.019\right)$ & $\left(0.344,0.019\right)$ & $\left(0.228,0.012\right)$ & $\left(0.245,0.013\right)$ & $\left(0.372,0.019\right)$ & $\left(0.393,0.019\right)$ & $\left(0.271,0.013\right)$ & $\left(0.289,0.015\right)$ \\ 
 & 6 &  &  &  &  & $\left(0.316,0.019\right)$ & $\left(0.343,0.018\right)$ & $\left(0.223,0.012\right)$ & $\left(0.244,0.013\right)$ & $\left(0.365,0.018\right)$ & $\left(0.392,0.021\right)$ & $\left(0.265,0.012\right)$ & $\left(0.289,0.014\right)$ \\ 
\specialrule{0.5mm}{0pt}{0pt}
\multirow{5}{*}{800} 
 & 1 &  &  &  &  &  &  &  &  & $\left(0.011,0.008\right)$ & $\left(0.509,0.103\right)$ & $\left(0.012,0.008\right)$ & $\left(0.391,0.074\right)$ \\ 
 & 2 &  &  &  &  &  &  &  &  & $\left(0.148,0.103\right)$ & $\left(0.477,0.027\right)$ & $\left(0.184,0.075\right)$ & $\left(0.367,0.017\right)$ \\ 
 & 3 &  &  &  &  &  &  &  &  & $\left(0.459,0.020\right)$ & $\left(0.467,0.018\right)$ & $\left(0.355,0.014\right)$ & $\left(0.361,0.014\right)$ \\ 
 & 4 &  &  &  &  &  &  &  &  & $\left(0.453,0.018\right)$ & $\left(0.463,0.018\right)$ & $\left(0.349,0.013\right)$ & $\left(0.359,0.013\right)$ \\ 
 & 5 &  &  &  &  &  &  &  &  & $\left(0.447,0.018\right)$ & $\left(0.465,0.018\right)$ & $\left(0.340,0.013\right)$ & $\left(0.356,0.013\right)$ \\ 
 & 6 &  &  &  &  &  &  &  &  & $\left(0.439,0.017\right)$ & $\left(0.460,0.018\right)$ & $\left(0.337,0.013\right)$ & $\left(0.358,0.014\right)$ 
\end{tabular}\\
\end{tabular}
\end{adjustbox}}

\end{center}
\end{table}

\section{Robust estimation of the spectrum}\label{appendix:Spectrumestimation}
In the procedure, when $\hat{\Sigma}=P^{1/2}WP^{1/2}$, we estimate 
\begin{equation*}
\frac{1}{m}\sum_{i=1}^m f\left( \lambda_{W,i} \right)
\end{equation*}
by 
\begin{equation*}
\frac{1}{m-k}\sum_{i=k+1}^m f\left( \lambda_{\hat{\Sigma},i} \right).
\end{equation*}
The estimation always underestimates the true value. When the estimation is used for a second moment, this has no real impact. However, the first moment could lead to a loss of the conservative properties of the procedure. Even if simulations show that this loss is very small when $m$ is large, we propose a more conservative way to estimate the expectation of the residual spike that only uses $f(x)=x^2$. When $\hat{\Sigma}=P_k^{1/2}WP_k^{1/2}$ is as defined above in \ref{section:teststat}, then
\begin{eqnarray*}
\frac{1}{m}\sum_{i=1}^m \left( \hat{\lambda}_{W,i} \right)^2 
\end{eqnarray*}
can be estimated by
\begin{eqnarray*}
\frac{\sum_{i=k+1}^m \left( \hat{\lambda}_{\hat{\Sigma},i}(k) \right)^2 + 2 k \left( \hat{\lambda}_{\hat{\Sigma},k+1}  \right)^2}{(m-k) M_{1,\hat{\Sigma}}+ 2 k \hat{\lambda}_{\hat{\Sigma},k+1} } 
\end{eqnarray*}

\pagebreak
\section{Statistical applications of Random matrix theory:\\ comparison of two populations III,\\
Supplement} \label{appendixproof}

\subsection{Introduction} 

This appendix contains the supplemental material presenting the proofs of the theorems and lemmas of the paper. The theorems are first introduced with the same notation as in the main paper and directly proved. Additional notations and assumptions are introduced at the start.

\subsection{Notations, definitions, assumptions and previous theorems}
The notation in \cite{mainarticle} are as follows.

\begin{Not}\ \label{ANot=Theorem}\\
Although we use a precise notation to enunciate the theorems, the proofs often use a simpler notation when no confusion is possible. This difference is always specified at the beginning of the proofs. 
\begin{itemize}
\item If $W$ is a symmetric random matrix, we denote by $\left(\hat{\lambda}_{W,i}, \hat{u}_{W,i} \right)$ its $i^{\rm th}$ eigenvalue and eigenvector.
\item A finite perturbation of order $k$ is denoted by $P_k= \I_m  +\sum_{i=1}^k (\theta_i-1) u_i u_i^t \in \mathbb{R}^{m \times m}$ with $u_1,u_2,...,u_k \in \mathbb{R}^{m \times m}$ orthonormal vectors.
\item We denote by $W \in  \mathbb{R}^{m\times m}$ an invariant by rotation random matrix as defined in Assumption \ref{AAss=matrice}. Moreover, the estimated covariance matrix is $\hat{\Sigma}=P_k^{1/2} W P_k^{1/2}$. \\
When comparing two groups, we use $W_X$, $W_Y$ and $\hat{\Sigma}_X$, $\hat{\Sigma}_Y$.
\item When we consider only one group, $\hat{\Sigma}_{P_r}=P_r^{1/2} W P_r^{1/2}$ is the perturbation of order $r$ of the matrix $W$ and: 
\begin{itemize}
\item $\hat{u}_{P_r,i}$ is its $i^{\text{th}}$ eigenvector. When $r=k$ we just use the simpler notation $\hat{u}_{i}=\hat{u}_{P_k,i}$ after an explicit statement.
\item  $\hat{u}_{P_r,i,j}$ is the $j^{\rm th}$ component of the $i^{\rm th}$ eigenvector.
\item $\hat{\lambda}_{P_r,i}$ is its $i^{\text{th}}$ eigenvalue. If $\theta_1>\theta_2>...>\theta_r$, then for $i=1,2,...,r$ we use also the notation $\hat{\theta}_{P_r,i}=\hat{\lambda}_{P_r,i}$. We call these eigenvalues the spikes. When $r=k$, we just use the simpler notation $\hat{\theta}_i=\hat{\theta}_{P_k,i}$ after an explicit statement.
\item $\hat{\alpha}_{P_r,i}^2=\sum_{j=1}^r \left\langle \hat{u}_{P_r,i},u_j \right\rangle^2$ is called the \textbf{general angle}.
\end{itemize}  
With this notation, we have $\hat{\Sigma}=\hat{\Sigma}_{P_k}=P_k^{1/2} W P_k^{1/2}$.
\item When we look at two groups $X$ and $Y$, we use a notation similar to the above. The perturbation of order $r$ of the matrices $W_X$ and $W_Y$ are $\hat{\Sigma}_{X,P_r}=P_r^{1/2} W_X P_r^{1/2}$ and $\hat{\Sigma}_{Y,P_r}=P_r^{1/2} W_Y P_r^{1/2}$ respectively. Then, we define for the group $\hat{\Sigma}_{X,P_r}$ (and similarly for $\hat{\Sigma}_{Y,P_r}$):
\begin{itemize}
\item $\hat{u}_{\hat{\Sigma}_{X,P_r},i}$ is its $i^{\text{th}}$ eigenvector. When $r=k$ we just use the simpler notation $\hat{u}_{X,i}=\hat{u}_{\hat{\Sigma}_{X,P_k},i}$ after an explicit statement.
\item  $\hat{u}_{\hat{\Sigma}_{X,P_r},i,j}$ is the $j^{\rm th}$ component of the $i^{\rm th}$ eigenvector.
\item $\hat{\lambda}_{\hat{\Sigma}_{X,P_r},i}$ is its $i^{\text{th}}$ eigenvalue. If $\theta_1>\theta_2>...>\theta_r$, then for $i=1,2,...,r$ we use the notation $\hat{\theta}_{\hat{\Sigma}_{X,P_r},i}=\hat{\lambda}_{\hat{\Sigma}_{X,P_r},i}$. When $r=k$, we just use the simpler notation $\hat{\theta}_{X,i}=\hat{\theta}_{\hat{\Sigma}_{X,P_k},i}$ after an explicit statement.
\item $\hat{\alpha}_{\hat{\Sigma}_{X,P_r},i}^2=\sum_{j=1}^r \left\langle \hat{u}_{\hat{\Sigma}_{X,P_r},i},u_j \right\rangle^2$. 
\item $\hat{\alpha}_{X,Y,P_r,i}^2=\sum_{j=1}^r \left\langle \hat{u}_{\hat{\Sigma}_{X,P_r},i},\hat{u}_{\hat{\Sigma}_{Y,P_r},j} \right\rangle^2$ is the \textbf{double angle} and, when no confusion is possible, we use the simpler notation $\hat{\alpha}_{P_r,i}^2$. When this simpler notation is used, it is stated explicitly.
\end{itemize} 
\item The proofs can assume either the sign convention 
\begin{eqnarray*}
\hat{u}_{P_s,i,i}>0, \text{ for $s=1,2,...,k$ and $i=1,2,...,s$,}
\end{eqnarray*}
or they may use 
\begin{eqnarray*}
\hat{u}_{P_s,i,s}>0, \text{ for $s=1,2,...,k$ and $i=1,2,...,s$,}.
\end{eqnarray*}
Which convention is adopted will be indicated in the proofs when confusion is possible.
\item We define the function $M_{s_1,s_2,X}(\rho_X)$, $M_{s_1,s_2,Y}(\rho_Y)$ and $M_{s_1,s_2}(\rho_X,\rho_Y)$ as 
\begin{eqnarray*}
M_{s_1,s_2,X}(\rho_X)&=&\frac{1}{m} \sum_{i=1}^m \frac{\hat{\lambda}_{W_X,i}^{s_1}}{\left( \rho_X -\hat{\lambda}_{W_X,i} \right)^2},\\
M_{s_1,s_2,Y}(\rho_Y)&=&\frac{1}{m} \sum_{i=1}^m \frac{\hat{\lambda}_{W_Y,i}^{s_1}}{\left( \rho_Y -\hat{\lambda}_{W_Y,i} \right)^2},\\
M_{s_1,s_2}(\rho_X,\rho_Y)&=& \frac{M_{s_1,s_2,X}(\rho_X)+M_{s_1,s_2,Y}(\rho_Y)}{2}.
\end{eqnarray*}
In particular, when $s_2=0$, we use $M_{s_1,X}=M_{s_1,0,X}$. When studying a single group, we use the simpler notation $M_{s_1,s_2}(\rho)$.
\item We use two transforms inspired by the T-transform:
\begin{itemize}
\item $T_{W,u}(z)= \sum_{i=1}^m \frac{\hat{\lambda}_{W,i}}{z-\hat{\lambda}_{W,i}} \left\langle \hat{u}_{W,i},u \right\rangle^2$ is the T-transform in direction $u$ using the random matrix $W$.
\item $\hat{T}_{\hat{\Sigma}_X}(z)= \frac{1}{m}\sum_{i=k+1}^m \frac{\hat{\lambda}_{\hat{\Sigma}_X,i}}{z-\hat{\lambda}_{\hat{\Sigma}_X,i}}$,
 and  $\hat{T}_{W_X}(z)= \frac{1}{m} \sum_{i=1}^m \frac{\hat{\lambda}_{W_X,i}}{z-\hat{\lambda}_{W_X,i}}$ are the estimated T-transforms using $\hat{\Sigma}_X$ and $W$, respectively.
\end{itemize}
\end{itemize}
\end{Not}

We recall the assumptions and definitions used throughout the paper.
\begin{Ass}\label{AAss=matrice} 
Let $W_X=O_X \Lambda_X O_X$ and $W_Y=O_Y \Lambda_Y O_Y$ with 
\begin{eqnarray*}
&&O_X, O_Y  \text{ being unit orthonormal invariant and independent random matrices,}\\
&&\Lambda_X,\Lambda_Y \text{ being diagonal bounded matrices and independent of } O_X, O_Y,\\
&&\text{and }\Tr\left( W_X \right)= \Tr\left( W_Y \right)=m.
\end{eqnarray*}

\noindent Assume $P_X= \I_m + \sum_{i=1}^k(\theta_{X,i}-1) e_i e_i^t$ and $P_Y= \I_m+ \sum_{i=1}^k(\theta_{Y,i}-1) e_i e_i^t$. Then 
\begin{eqnarray*}
\hat{\Sigma}_X=P_X^{1/2} W_X P_X^{1/2} \text{ and } \hat{\Sigma}_Y=P_Y^{1/2} W_Y P_Y^{1/2}.
\end{eqnarray*}
\end{Ass}

\begin{Ass}\label{AAss=theta}
\begin{itemize} \
\item[(A1)] $\frac{\theta}{\sqrt{m} }\rightarrow \infty.$
\item[(A2)] $\theta_i=p_i \theta$, where $p_i$ is fixed and different from $1$.
\end{itemize}
\end{Ass}

\begin{Def} \label{ADef=unbiased} 
Suppose $\hat{\Sigma}$ satisfies Assumption \ref{AAss=matrice}. Then, the \textbf{unbiased estimator of }$\theta$ is defined as 
$$ \hat{\hat{\theta}}=1+\frac{1}{\frac{1}{m-k} \sum_{i=k+1}^{m} \frac{\hat{\lambda}_{\hat{\Sigma},i}}{\hat{\theta}-\hat{\lambda}_{\hat{\Sigma},i}}},$$
where $\hat{\lambda}_{\hat{\Sigma},i}$ is the $i^{\text{th}}$ of $\hat{\Sigma}$.\\
With $\hat{\theta}$ and $\hat{u}_i$ the $i^{\text{th}}$ eigenvalue and eigenvector of $\hat{\Sigma}$, the \textbf{filtered estimated covariance matrix } is defined as 
$$\hat{\hat{\Sigma}}= \I_m+\sum_{i=1}^k (\hat{\hat{\theta}}_i-1) \hat{u}_i \hat{u}_i^t.$$
\end{Def}
\begin{Def} \label{ADef:Invariant}
Let $W$ is a random matrix and let $P_1=\I_m+(\theta_1-1) u_1 u_1^t$ and $P_k=\I_m+\sum_{i=1}^k (\theta_i-1)u_i u_i^t$ be perturbations of order $1$ and $k$, respectively. We say that a statistic $T\left( W_m,P_1\right)$ is \textbf{invariant} with respect to $k$, if $T\left( W_m,P_k\right)$ is such that 
\begin{eqnarray*}
\resizebox{.9\hsize}{!}{$T\left( W_m,P_k\right) = T\left( W_m,P_1\right) + \epsilon_m, \text{ where }
\max\left(\frac{\epsilon_m}{\E\left[ T\left( W,P_1\right)\right]},\frac{\epsilon_m^2}{\var\left( T\left( W,P_1\right) \right)}\right) \rightarrow 0.$}
\end{eqnarray*}
\end{Def}
We are now ready to state the main Theorems of this paper.

\begin{Th}
 \label{ATH=Main} 
Suppose $W_X,W_Y \in \mathbb{R}^{m\times m}$ satisfy Assumption \ref{AAss=matrice} and $(m,\theta_i)$ are according to Assumption \ref{AAss=theta}.
\begin{enumerate}
\item \cite{mainarticle} has investigated the asymptotics of large $m$ in the case of a perturbation of order 1, that is, 
$P_1= \I_m+(\theta-1)e_1 e_1^t \in  \mathbb{R}^{m\times m}$ with $\frac{\sqrt{m}}{\theta}=o(1)$. Let  
\begin{eqnarray*}
\hat{\Sigma}_{X,P_1}=P_1^{1/2} W_X P_1^{1/2} \text{ and } \hat{\Sigma}_{Y,P_1}=P_1^{1/2} W_Y P_1^{1/2}.
\end{eqnarray*}
and $\hat{\hat{\Sigma}}_{X,P_1}$, $\hat{\hat{\Sigma}}_{Y,P_1}$ as described above (see, \ref{ADef=unbiased}). 

Then, conditional on the spectra $S_{W_X}=\left\lbrace \hat{\lambda}_{W_X,1},\hat{\lambda}_{W_X,2},...,\hat{\lambda}_{W_X,m} \right\rbrace$ and $S_{W_Y}=\left\lbrace \hat{\lambda}_{W_Y,1},\hat{\lambda}_{W_Y,2},...,\hat{\lambda}_{W_Y,m} \right\rbrace$ of $W_X$ and $W_Y$,
\begin{eqnarray*}
\left. \sqrt{m} \frac{\left(\lambda_{\max}\left(\hat{\hat{\Sigma}}_{X,P_1}^{-1/2} \hat{\hat{\Sigma}}_{Y,P_1} \hat{\hat{\Sigma}}_{X,P_1}^{-1/2} \right) -\lambda^+ \right)}{\sigma^+} \right| S_{W_X}, S_{W_Y} \sim \Normal(0,1)+o_{p}(1), 
\end{eqnarray*}
where
\begin{eqnarray*}
&&\lambda^+= \sqrt{M_2^2-1}+M_2, \hspace{30cm}
\end{eqnarray*}
\begin{eqnarray*}
&&{\sigma^+}^2 \equiv {\sigma^+}^2 \left(M_{2,X},M_{3,X},M_{4,X},M_{2,Y},M_{3,Y},M_{4,Y}\right)  \hspace{30cm}
\end{eqnarray*}
\begin{eqnarray*}
&&M_{s,X}= \frac{1}{m} \sum_{i=1}^m \hat{\lambda}_{W_X,i}^s, \ 
M_{s,Y}= \frac{1}{m} \sum_{i=1}^m \hat{\lambda}_{W_Y,i}^s, \ 
M_s=\frac{M_{s,X}+M_{s,Y}}{2}.\hspace{30cm}
\end{eqnarray*}
Moreover, 
\begin{eqnarray*}
\left.\sqrt{m} \frac{\left(\lambda_{\min}\left(\hat{\hat{\Sigma}}_{X,P_1}^{-1/2} \hat{\hat{\Sigma}}_{Y,P_1} \hat{\hat{\Sigma}}_{X,P_1}^{-1/2} \right) -\lambda^- \right)}{\sigma^-}\right| S_{W_X}, S_{W_Y} \sim \Normal(0,1)+o_{m}(1), 
\end{eqnarray*}
where 
\begin{eqnarray*}
&&\lambda^-= -\sqrt{M_2^2-1}+M_2, \hspace{30cm}\\
&&{\sigma^-}^2=\left(\lambda^-\right)^4 {\sigma^+}^2.
\end{eqnarray*}
The error $o_p(1)$ in the approximation is with regard to large values of $m$.

\item Suppose that $P_k= \I_m+\sum_{s=1}^k(\theta_s-1)e_s e_s^t  \mathbb{R}^{m\times m}$ with $\theta_s=p_s \theta$, $p_s>0$ and $\frac{\sqrt{m}}{\theta}=o(1)$ with regard to large $m$.
\begin{eqnarray*}
\hat{\Sigma}_{X,P_k}=P_k^{1/2} W_X P_k^{1/2} \text{ and } \hat{\Sigma}_{Y,P_k}=P_k^{1/2} W_Y P_k^{1/2},
\end{eqnarray*}
and $\hat{\hat{\Sigma}}_{X,P_k}$, $\hat{\hat{\Sigma}}_{Y,P_k}$ as described above (see, \ref{ADef=unbiased}). 
\noindent Then, conditioning on the spectra $S_{W_X}$ and $S_{W_Y}$,\\
\scalebox{0.82}{
\begin{minipage}{1\textwidth}
\begin{eqnarray*}
\left. \lambda_{\max}\left(\hat{\hat{\Sigma}}_{X,P_k}^{-1/2} \hat{\hat{\Sigma}}_{Y,P_k} \hat{\hat{\Sigma}}_{X,P_k}^{-1/2}\right) \right| S_{W_X},S_{W_Y}=\lambda_{\max}\left( H^+ \right)+1+O_p\left(\frac{1}{m}\right)+O_p\left( \frac{1}{ \theta \sqrt{m}}\right),\\
\left.\lambda_{\min}\left(\hat{\hat{\Sigma}}_{X,P_k}^{-1/2} \hat{\hat{\Sigma}}_{Y,P_k} \hat{\hat{\Sigma}}_{X,P_k}^{-1/2}\right)\right| S_{W_X},S_{W_Y}=\lambda_{\max}\left( H^- \right)+1+O_p\left(\frac{1}{m}\right)+O_p\left( \frac{1}{ \theta \sqrt{m}}\right),
\end{eqnarray*}
\end{minipage}}
where
\begin{eqnarray*}
H^\pm= \zeta_{\infty}^\pm \begin{pmatrix}
\hat{\zeta}_1^\pm / \zeta_{\infty}^\pm &  w_{1,2}^\pm & w_{1,3}^\pm & \cdots &  w_{1,k}^\pm \\ 
 w_{2,1}^\pm  & \hat{\zeta}_2^\pm / \zeta_{\infty}^\pm &  w_{2,3}^\pm & \cdots  & w_{2,k}^\pm \\ 
 w_{3,1}^\pm & w_{3,2}^\pm & \hat{\zeta}_3^\pm /\zeta_{\infty}^\pm & \cdots  &  w_{3,k}^\pm \\ 
\vdots & \vdots & \ddots & \ddots &  \vdots \\ 
  w_{k,1}^\pm &  w_{k,2}^\pm & w_{k,3}^\pm & \cdots  & \hat{\zeta}_k^\pm /\zeta_{\infty}^\pm  \\ 
\end{pmatrix} ,
\end{eqnarray*}
and \\
\scalebox{0.73}{
\begin{minipage}{1\textwidth}
\begin{eqnarray*}
&&\hat{\zeta}_i^+= \left.\lambda_{\max}\left(\hat{\hat{\Sigma}}_{X,\tilde{P}_i}^{1/2} \hat{\hat{\Sigma}}_{Y,\tilde{P}_i} \hat{\hat{\Sigma}}_{X,\tilde{P}_i}^{1/2} \right)-1 \right| S_{W_X},S_{W_Y},\hspace{30cm}\\
&&\hat{\zeta}_i^-= \left.\lambda_{\min}\left(\hat{\hat{\Sigma}}_{X,\tilde{P}_i}^{1/2} \hat{\hat{\Sigma}}_{Y,\tilde{P}_i} \hat{\hat{\Sigma}}_{X,\tilde{P}_i}^{1/2} \right)-1 \right| S_{W_X},S_{W_Y},\\
&&\zeta_{\infty}^\pm =\underset{m \rightarrow \infty}{\lim} \hat{\zeta}_i^\pm = \lambda^\pm -1,\\
&&w_{i,j}^\pm \sim { \color{red} \Normal}\left(0,
\frac{1}{m}\frac{2 (M_{2,X}-1) (M_{2,Y}-1)+B_{X}^\pm +B_{Y}^\pm}{\left((\zeta_{\infty}^\pm-2 M_2+1)^2+2 (M_2-1)\right)^2}
  \right)+o_p\left( \frac{1}{\sqrt{m}} \right),
\end{eqnarray*}
\end{minipage}}\\
\scalebox{0.73}{
\begin{minipage}{1\textwidth}
\begin{eqnarray*}
&&B_{X}^+=\left(1-M_2+2 M_{2,X}+\sqrt{M_2^2-1}\right)^2 (M_{2,X}-1)\hspace{20cm}\\
&&\hspace{2cm} +2\left(-1+M_2-2M_{2,x}-\sqrt{M_2^2-1}\right)(M_{3,X}-M_{2,X})+(M_{4,X}-M_{2,X}^2),\\
&&B_{Y}^+=\left(1+M_2+M_{2,Y}-M_{2,X}-\sqrt{M_2^2-1}\right)^2 (M_{2,Y}-1)\\
&&\hspace{2cm} +2\left(-1-M_2-M_{2,Y}-M_{2,X}-\sqrt{M_2^2-1}\right)(M_{3,Y}-M_{2,Y})+(M_{4,Y}-M_{2,Y}^2),\\
&&B_{X}^-=\left(1-M_2+2 M_{2,X}-\sqrt{M_2^2-1}\right)^2 (M_{2,X}-1)\hspace{20cm}\\
&&\hspace{2cm} +2\left(-1+M_2-2M_{2,x}+\sqrt{M_2^2-1}\right)(M_{3,X}-M_{2,X})+(M_{4,X}-M_{2,X}^2),\\
&&B_{Y}^-=\left(1+M_2+M_{2,Y}-M_{2,X}+\sqrt{M_2^2-1}\right)^2 (M_{2,Y}-1)\\
&&\hspace{2cm} +2\left(-1-M_2-M_{2,Y}+M_{2,X}-\sqrt{M_2^2-1}\right)(M_{3,Y}-M_{2,Y})+(M_{4,Y}-M_{2,Y}^2).
\end{eqnarray*}
\end{minipage}}\\
\noindent The matrices $H^+$ and $H^-$ are strongly correlated. However, within a matrix, all the entries are uncorrelated.
\end{enumerate}
\end{Th}

\subsection{Proof}
Results necessary to prove the Main Theorem \ref{ATH=Main} are separated into subsection. 
First, we present a sketch of the proof without detail to explain the main idea. 
Then, we prove some useful lemmas from linear algebra and finally, we put the things together.

\subsubsection{Rough sketch of the idea behind the proof} \label{Asec=sketch}

\paragraph*{Residual spike for a perturbation of order $1$} 
The first part of Theorem \ref{ATH=Main} concerns perturbations of order $1$ and is proved in \cite{mainarticle}.

\paragraph*{Decomposition of the matrix} 
The generalisation of the previous result to perturbations of order $k > 1$ is not straightforward. 
We want to study the largest eigenvalue of 
\begin{eqnarray*}
\hat{\hat{\Sigma}}_{P_k,X}^{-1/2} \hat{\hat{\Sigma}}_{P_k,Y} \hat{\hat{\Sigma}}_{P_k,X}^{-1/2},
\end{eqnarray*}
where $\hat{\hat{\Sigma}}_{P_k,X}$ is the filtered estimator of the random covariance matrix $\hat{\Sigma}_{P_k,X}= P_k^{1/2} W_X P_k^{1/2}$ defined in Def. \ref{ADef=unbiased}. First, we define a rotation matrix $\hat{U}_{P_k,X}$ such that $\hat{U}_{P_k,X}^t \hat{u}_{P_k,X,i}=e_i$ for $i=1,2,...,k$ and $\hat{U}_{P_k,X}^t \hat{u}_{P_k,Y,i}=\tilde{u}_{P_k}$. Then, we study the transformed matrices which retains the same eigenvalues,

\scalebox{0.77}{
\begin{minipage}{1\textwidth}
\begin{eqnarray*}
\Sigma_{P_k,X}^{-1/2} \tilde{\Sigma}_{P_k,Y} \Sigma_{P_k,X}^{-1/2}&=& \I_m + \sum_{i=1}^k \left[ \Sigma_{P_k,X}^{-1/2}  \left(\hat{\hat{\theta}}_{P_k,Y,i}-1 \right) \tilde{u}_{P_k,i} \tilde{u}_{P_k,i}^t \Sigma_{P_k,X}^{-1/2} +\left(\frac{1}{\hat{\hat{\theta}}_{P_k,X,i}}-1\right)e_i e_i^t \right],
\end{eqnarray*}
\end{minipage}}\\
where $\Sigma_{P_k,X}=\I_m + \sum_{i=1}^k \left(\hat{\hat{\theta}}_{P_k,X,i}-1 \right) e_i e_i^t$ and \linebreak $\tilde{\Sigma}_{P_k,Y}=\I_m + \sum_{i=1}^k \left(\hat{\hat{\theta}}_{P_k,Y,i}-1 \right) \tilde{u}_i \tilde{u}_i^t$. The details are presented in Section \ref{Asec:decomp}.

\paragraph*{Pseudo Invariance of the residual spike}
When $k$ grows, the residual spike is not invariant in the sense of Def. \ref{ADef:Invariant}. We define for $i=1,2,...,k$, $\hat{\Sigma}_{\tilde{P}_i,X}= \tilde{P}_i^{1/2} W_X \tilde{P}_i^{1/2}$, where $\tilde{P}_i= \I_m + (\theta_i-1) e_i e_i^t$ is a perturbation of order $1$. 
The following invariance then holds,\\
\scalebox{0.72}{
\begin{minipage}{1\textwidth}
\begin{eqnarray*}
 \lambda \left( \Sigma_{P_k,X}^{-1/2}  \left(\hat{\hat{\theta}}_{P_k,Y,i}-1 \right) \tilde{u}_{P_k,i} \tilde{u}_{P_k,i}^t \Sigma_{P_k,X}^{-1/2} +\left(\frac{1}{\hat{\hat{\theta}}_{P_k,X,i}}-1\right)e_i e_i^t \right)
= \lambda\left( \hat{\hat{\Sigma}}_{\tilde{P}_i,X}^{-1/2} \hat{\hat{\Sigma}}_{\tilde{P}_i,Y} \hat{\hat{\Sigma}}_{\tilde{P}_i,X}^{-1/2} \right)-1 + O_p\left( \frac{1}{m} \right),
\end{eqnarray*} 
\end{minipage}}\\
where $\lambda()$ provides the non null eigenvalues. This result is proven assuming either that $\theta_i$ is large or that the two non-trivial residual spikes of the perturbation of order $1$ are distinct. The second condition is easy to show, because dot products between eigenvectors do not tend to $1$. However, when $n_X,n_Y>>m$, this could create some imprecision. The details are presented in Section \ref{Asec:pseudoinvresidual}.
 
\paragraph*{Pseudo residual eigenvectors}
The previous part demonstrates an invariance property for some eigenvalues related to the residual spike. The next step studies eigenvectors corresponding to these eigenvalues. For $s=1,2,...,k$, we set 
\begin{eqnarray*}
\hat{\zeta}_s^{\pm}=\lambda \left( \Sigma_{P_k,X}^{-1/2}  \left(\hat{\hat{\theta}}_{P_k,Y,s}-1 \right) \tilde{u}_{P_k,s} \tilde{u}_{P_k,s}^t \Sigma_{P_k,X}^{-1/2} +\left(\frac{1}{\hat{\hat{\theta}}_{P_k,X,s}}-1\right)e_s e_s^t \right)
\end{eqnarray*}
and 
\begin{eqnarray*}
w_s^{\pm}&=&u\Bigg(\Sigma_{P_k,X}^{-1/2} \left( (\hat{\theta}_{P_k,Y,s}-1) \tilde{u}_{P_k,s}\tilde{u}_{P_k,Y,s}^t \right) \Sigma_{P_k,X}^{-1/2}+ \left(\frac{1}{\hat{\theta}_{P_k,X,s}}-1\right) e_s e_s^t \Bigg),
\end{eqnarray*}
its corresponding eigenvector. The notation $\pm$ allows $\hat{\zeta}_s^{-}$ and $\hat{\zeta}_s^{+}$ to be distinguished. We define for $s=1,2,...,k$,
$$\zeta_{\infty}^{\pm}(\theta_s)= \underset{m \rightarrow \infty}{\lim} \hat{\zeta}_s^{\pm} \text{ and }
\zeta_{\infty}^{\pm}= \underset{m,\theta_s \rightarrow \infty}{\lim} \hat{\zeta}_s^{\pm}.$$
Then,
\begin{eqnarray*}
w_{s,s}^{\pm}
&=&\frac{\left(\zeta_{\infty}^{\pm}-2 \left(M_2 -1\right)\right)}{\sqrt{\left(\zeta_{\infty}^{\pm}-2 \left(M_2 -1\right)\right)^2+2 \left(M_2 -1\right) }}+ O_p\left( \frac{1}{\sqrt{m}} \right) + o_{p;\theta}\left( 1 \right), \hspace{20cm}
\end{eqnarray*} 
\scalebox{0.72}{
\begin{minipage}{1\textwidth}
\begin{eqnarray*}
w_{s,2:m \setminus s}^{\pm}
&=&\frac{\sqrt{\hat{\theta}_{P_k,Y,1}-1} \left(\frac{\tilde{u}_{P_k,s,1}}{\sqrt{\hat{\theta}_{P_k,X,1}}},...,\frac{\tilde{u}_{P_k,s,s-1}}{\sqrt{\hat{\theta}_{P_k,X,s-1}}},\frac{\tilde{u}_{P_k,s,s+1}}{\sqrt{\hat{\theta}_{P_k,X,s+1}}},...,\frac{\tilde{u}_{P_k,s,k}}{\sqrt{\hat{\theta}_{P_k,X,k}}},\tilde{u}_{P_k,s,k+1},...,\tilde{u}_{P_k,s,m}\right) }{\sqrt{\left(\zeta_{\infty}^{\pm}-2 \left(M_2 -1\right)\right)^2+2 \left(M_2 -1\right)}+O_p\left(\frac{1}{\sqrt{m}} \right)+o_{p;\theta} (1)},\hspace{20cm}
\end{eqnarray*} 
\end{minipage}}\vspace{0.3cm}\\
The details of the proof are presented in Section \ref{Asec:pseudoinvresidualvector}.

\begin{Rem}
These results concerning pseudo residual structure are valid for large $\theta$. The paper is based on extracts from the Thesis \cite{mathese} that prove similar formulas for all $\theta$. 
\end{Rem}

\paragraph*{Dimension reduction}
The three previous parts showed that

\scalebox{0.72}{
\begin{minipage}{1\textwidth}
\begin{eqnarray*}
\lambda\left( \hat{\hat{\Sigma}}_{P_k,X}^{-1/2} \hat{\hat{\Sigma}}_{P_k,Y} \hat{\hat{\Sigma}}_{P_k,X}^{-1/2}\right)
&=&\lambda\left( \I_m + \sum_{i=1}^k \left[ \Sigma_{P_k,X}^{-1/2}  \left(\hat{\hat{\theta}}_{P_k,Y,i}-1 \right) \tilde{u}_{P_k,i} \tilde{u}_{P_k,i}^t \Sigma_{P_k,X}^{-1/2} +\left(\frac{1}{\hat{\hat{\theta}}_{P_k,X,i}}-1\right)e_i e_i^t \right] \right)
\end{eqnarray*}
\end{minipage}}\vspace{0.3cm}\\
and for $i=1,2,...,k$, \\
\scalebox{0.83}{
\begin{minipage}{1\textwidth}
\begin{eqnarray*}
\hat{\zeta}_{i}^+ w_{i}^+ {w_{i}^+}^t +\hat{\zeta}_{i}^- w_{i}^- {w_{i}^-}^t=\left[ \Sigma_{P_k,X}^{-1/2}  \left(\hat{\hat{\theta}}_{P_k,Y,i}-1 \right) \tilde{u}_{P_k,i} \tilde{u}_{P_k,i}^t \Sigma_{P_k,X}^{-1/2} +\left(\frac{1}{\hat{\hat{\theta}}_{P_k,X,i}}-1\right)e_i e_i^t \right],\\
\end{eqnarray*}
\end{minipage}}\vspace{0.3cm}\\
where $\pm$ allows the distinction of the pseudo eigenvalues and eigenvectors such that $+$ is the largest and $-$ is the smallest. 
Easy arguments from linear algebra lead to 
\begin{eqnarray*}
\lambda_{\max}\left( \hat{\hat{\Sigma}}_{P_k,X}^{-1/2} \hat{\hat{\Sigma}}_{P_k,Y} \hat{\hat{\Sigma}}_{P_k,X}^{-1/2} \right) = \lambda_{\max}\left(H^+ \right)+1+O_p\left(\frac{1}{m} \right) ,\\
\lambda_{\min}\left( \hat{\hat{\Sigma}}_{P_k,X}^{-1/2} \hat{\hat{\Sigma}}_{P_k,Y} \hat{\hat{\Sigma}}_{P_k,X}^{-1/2} \right) = \lambda_{\min}\left(H^- \right)+1+O_p\left(\frac{1}{m} \right) ,
\end{eqnarray*}
where $H^{\pm}$ are matrices of dimension $k$ such that\\
\scalebox{0.75}{
\begin{minipage}{1\textwidth}
\begin{eqnarray*}
H^{\pm} &=& \begin{pmatrix}
\hat{\zeta}_1^{\pm} & \sqrt{\hat{\zeta}_1^{\pm} \hat{\zeta}_2^{\pm}} \left\langle w_1^{\pm},w_2^{\pm} \right\rangle & \sqrt{\hat{\zeta}_1^{\pm} \hat{\zeta}_3^{\pm}} \left\langle w_1^{\pm},w_3^{\pm} \right\rangle & \cdots & \sqrt{\hat{\zeta}_k^{\pm} \hat{\zeta}_2^{\pm}} \left\langle w_1^{\pm},w_k^{\pm} \right\rangle \\ 
\sqrt{\hat{\zeta}_2^{\pm} \hat{\zeta}_1^{\pm}} \left\langle w_2^{\pm},w_1^{\pm} \right\rangle  & \hat{\zeta}_2^{\pm} & \sqrt{\hat{\zeta}_2^{\pm} \hat{\zeta}_3^{\pm}} \left\langle w_2^{\pm},w_3^{\pm} \right\rangle & \cdots  & \sqrt{\hat{\zeta}_2^{\pm} \hat{\zeta}_k^{\pm}} \left\langle w_2^{\pm},w_k^{\pm} \right\rangle \\ 
\sqrt{\hat{\zeta}_3^{\pm} \hat{\zeta}_1^{\pm}} \left\langle w_3^{\pm},w_1^{\pm} \right\rangle &\sqrt{\hat{\zeta}_3^{\pm} \hat{\zeta}_2^{\pm}} \left\langle w_3^{\pm},w_2^{\pm} \right\rangle  & \hat{\zeta}_3^{\pm} & \cdots  & \sqrt{\hat{\zeta}_3^{\pm} \hat{\zeta}_k^{\pm}} \left\langle w_3^{\pm},w_k^{\pm} \right\rangle \\ 
\vdots & \vdots & \ddots & \ddots &  \vdots \\ 
 \sqrt{\hat{\zeta}_k^{\pm} \hat{\zeta}_1^{\pm}} \left\langle w_k^{\pm},w_1^{\pm} \right\rangle & \sqrt{\hat{\zeta}_k^{\pm} \hat{\zeta}_2^{\pm}} \left\langle w_k^{\pm},w_2^{\pm} \right\rangle & \sqrt{\hat{\zeta}_k^{\pm} \hat{\zeta}_3^{\pm}} \left\langle w_k^{\pm},w_3^{\pm} \right\rangle & \cdots  & \hat{\zeta}_k^{\pm}  \\ 
\end{pmatrix}.
\end{eqnarray*}
\end{minipage}}\vspace{0.3cm}\\
The details are explained in Section \ref{Asec:dimreduce}.

\paragraph*{Elements of $H$}
The matrices $H^{\pm}$ are functions of $\hat{\zeta}_{i}^{\pm}$ and $\left\langle w_{i}^\pm,w_{j}^\pm \right\rangle$ for $i,j=1,2,...,k$. By the pseudo invariance of the residual spike, we know that $\hat{\zeta}_{i}^\pm$ behaves like residual spikes of a perturbation of order $1$. Using a theorem of \cite{mainarticle2}, we can express $\left\langle w_{i}^\pm,w_{j}^\pm \right\rangle$ as a function of well-known statistics. We directly see that $\left\langle w_{s}^\pm,w_{s}^\pm \right\rangle=1$. Moreover, for $s \neq t$,
\begin{eqnarray*}
&&\left\langle w_{i}^\pm,w_{j}^\pm \right\rangle \sim {\rm RV}\left(0,\frac{1}{m}
\frac{2 (M_{2,X}-1) (M_{2,Y}-1)+B_{X}^\pm +B_{Y}^\pm}{\left((\zeta_{\infty}^\pm-2 M_2+1)^2+2 (M_2-1)\right)^2}
  \right)
  \end{eqnarray*}
where  $B_{X}^\pm$ and $B_{Y}^\pm$ are defined in the Theorem \ref{ATH=Main}. The details of the computation are presented in Section \ref{Asec:elementH}.

\paragraph*{Normality discussion}\ \\
When the perturbation is of order $1$, we have already proved the normality in \cite{mainarticle}. When the perturbation is of order $k$ and $n_X>>n_Y$, then the joint normality is a straightforward consequence of the unit statistic Theorem of \cite{mainarticle}. In the general case we can only express the entries of $H^{\pm}$ as a function of marginally Normal statistics. The details of the computations are presented in Section \ref{Asec:HNormal}.

\subsubsection{Prerequisite Lemmas}
In order to prove the Main Theorem \ref{ATH=Main} we need some preliminary results, which we present in the form of lemmas.

\begin{Lem} \label{ALemreducedim} 
Suppose $w_1,...,w_k \in \mathbb{R}^m$ and $\lambda_1,...,\lambda_k \in \mathbb{R}^*$, then if the function $\lambda()$ provides the non-trivial eigenvalues,
\begin{eqnarray*}
\lambda\Bigg( \sum_{i=1}^k \lambda_i w_i w_i^t \Bigg) = \lambda\Bigg( H \Bigg) ,
\end{eqnarray*}
where\\
\scalebox{0.9}{
\begin{minipage}{1\textwidth}
\begin{eqnarray*}
H=\begin{pmatrix}
{\lambda}_1 & \sqrt{{\lambda}_1 {\lambda}_2} \left\langle w_1,w_2 \right\rangle & \sqrt{{\lambda}_1 {\lambda}_3} \left\langle w_1,w_3 \right\rangle & \cdots & \sqrt{{\lambda}_k {\lambda}_2} \left\langle w_1,w_k \right\rangle \\ 
\sqrt{{\lambda}_2 {\lambda}_1} \left\langle w_2,w_1 \right\rangle  & {\lambda}_2 & \sqrt{{\lambda}_2 {\lambda}_3} \left\langle w_2,w_3 \right\rangle & \cdots  & \sqrt{{\lambda}_2 {\lambda}_k} \left\langle w_2,w_k \right\rangle \\ 
\sqrt{{\lambda}_3 {\lambda}_1} \left\langle w_3,w_1 \right\rangle &\sqrt{{\lambda}_3 {\lambda}_2} \left\langle w_3,w_2 \right\rangle  & {\lambda}_3 & \cdots  & \sqrt{{\lambda}_3 {\lambda}_k} \left\langle w_3,w_k \right\rangle \\ 
\vdots & \vdots & \ddots & \ddots &  \vdots \\ 
 \sqrt{{\lambda}_k {\lambda}_1} \left\langle w_k,w_1 \right\rangle & \sqrt{{\lambda}_k {\lambda}_2} \left\langle w_k,w_2 \right\rangle & \sqrt{{\lambda}_k {\lambda}_3} \left\langle w_k,w_3 \right\rangle & \cdots  & {\lambda}_k  \\ 
\end{pmatrix} .
\end{eqnarray*}
\end{minipage}}
\end{Lem}

\begin{proof} 

\noindent We define
\begin{eqnarray*}
\Lambda=\begin{pmatrix}
\lambda_1 & 0 & \cdots & 0 \\ 
0 &\lambda_2 &  & 0 \\ 
\vdots &  &\ddots & \vdots \\ 
0 & 0 & \cdots & \lambda_k
\end{pmatrix} , \ 
W=\left( 
w_1 ,w_2 , ... , w_k
\right)  \in \mathbb{R}^{m \times k},
\end{eqnarray*}
then
\begin{eqnarray*}
\lambda\Bigg( \sum_{i=1}^k \lambda_i w_i w_i^t \Bigg)&=& \lambda\Bigg( W \Lambda W^t\Bigg)
= \lambda\Bigg( \left(W \Lambda^{1/2} \right)  \left(W \Lambda^{1/2} \right)^t\Bigg)\\
&=& \lambda\Bigg( \left(W \Lambda^{1/2} \right)^t  \left(W \Lambda^{1/2} \right)\Bigg) \text{ (for nonzero eigenvalues)}\\
&=& \lambda\Bigg(  \Lambda^{1/2} W^t W \Lambda^{1/2} \Bigg)
= \lambda\Bigg( H \Bigg).
\end{eqnarray*}

\end{proof}

\begin{Lem}\label{ALemmalinearalg}
Suppose $e_1,w \in \mathbf{R}^m$ and $a,b \in \mathbf{R}$. Then, if $||w||=1$, the two ($\pm$) non-trivial eigenvalues and eigenvectors are\\
\scalebox{0.82}{
\begin{minipage}{1\textwidth}
\begin{eqnarray*}
\lambda^{\pm}\Bigg( a e_1 e_1^t + b w w^t  \Bigg)&=& \frac{1}{2} \left(a+b \pm \sqrt{4 a b w_1^2+(a-b)^2}\right),\\
u^{\pm}\Bigg( a e_1 e_1^t + b w w^t  \Bigg)&=&\frac{1}{\rm{Norm}^{\pm}}\left(\frac{\lambda^{\pm}\Bigg( a e_1 e_1^t + b w w^t   \Bigg)+b \left(w1^2-1\right)}{ b w_1 },w_2,w_3,w_4,...,w_m\right),\\
\left({\rm Norm}^{\pm}\right)^2 &=& \frac{\left(  \lambda^{\pm}\Bigg( a e_1 e_1^t + b w w^t   \Bigg)+b \left(w1^2-1\right) \right)^2}{b^2 w_1^2 } +1-w_1^2.
\end{eqnarray*}
\end{minipage}}\vspace{0.3cm}\\
If $||w||\not = 1$,\\
\scalebox{0.82}{
\begin{minipage}{1\textwidth}
\begin{eqnarray*}
\lambda^{\pm}\Bigg( a e_1 e_1^t + w w^t  \Bigg)&=& \frac{1}{2} \left(\pm \sqrt{\left(a+ ||w||^2\right)^2-4 a \left(||w||^2-w_1^2\right)}+a+ ||w||^2\right),\\
u^{\pm}\Bigg(a e_1 e_1^t + w w^t \Bigg)
&=&\frac{1}{\rm{Norm}^{\pm}}\left(\frac{\lambda^{\pm}\Bigg( a e_1 e_1^t + w w^t   \Bigg)- ||w||^2+w_1^2}{  w_1},w_2,w_3,w_4,...,w_m\right),\\
\left({\rm Norm}^{\pm}\right)^2 &=& \frac{\left(\lambda^{\pm}\Bigg( a e_1 e_1^t + w w^t   \Bigg) - ||w||^2+w_1^2\right)^2}{ w_1^2 } +||w||^2-w_1^2.
\end{eqnarray*}
\end{minipage}}\vspace{0.3cm}\\
\end{Lem}

\begin{proof}
These results were computed with Wolfram Mathematica 11.1.1. and we invite the reader to check that
\begin{eqnarray*}
\left( a e_1 e_1^t + b w w^t \right) u^{\pm}=  \lambda^{\pm}u^{\pm}.
\end{eqnarray*}
\end{proof}

\begin{Lem} \label{ALemneglectvalue} 
Suppose $u_1,...,u_k\in \mathbb{R}^m$ are orthonormal and $\lambda_1 >...>\lambda_k \in \mathbb{R}^+$ with finite $k$. Suppose $v \in \mathbb{R}^m$ and $\mu \in \mathbb{R}^+$ such that $\left\langle u_i,v \right\rangle=O_p\left(1/\sqrt{m}\right)$ and $ \mu-\lambda_1  < d<0$ for a fixed $d$, then 
\begin{eqnarray*}
\lambda_{\max}\Bigg( \sum_{i=1}^k \lambda_i u_i u_i^t + \mu v v^t \Bigg)= \lambda_1+O_p\left(\frac{1}{m} \right).
\end{eqnarray*}
Moreover, if $ \mu-\lambda_k  > d_2 > 0$ for a fixed $d_2$,
\begin{eqnarray*}
\lambda_{\min}\Bigg( \sum_{i=1}^k \lambda_i u_i u_i^t + \mu v v^t \Bigg)= \lambda_k+O_p\left(\frac{1}{m} \right).
\end{eqnarray*}
\end{Lem}

\begin{proof}
\noindent Suppose $w$ is the maximum unit eigenvector of $\sum_{i=1}^k \lambda_i u_i u_i^t + \mu v v^t$. Then, 
\begin{eqnarray*}
&&w=\sum_{i=1}^k \alpha_i u_i  + \beta v,
\end{eqnarray*}
where
\begin{eqnarray*}
&&\sum_{i=1}^k \alpha_i^2+\beta^2+2 \sum_{i=1}^k \alpha_i \beta \left\langle u_i,v \right\rangle=1.
\end{eqnarray*}
If $\beta=O_p(1/\sqrt{m})$, then\\
\scalebox{0.9}{
\begin{minipage}{1\textwidth}
\begin{eqnarray*}
w^t \left(\sum_{i=1}^k \lambda_i u_i u_i^t + \mu v v^t \right) w &=& \sum_{i=1}^k \lambda_i \left(\alpha_i+\beta \left\langle v,u_i \right\rangle \right)^2+ \mu \left( \sum_{i=1}^k \alpha_i \left\langle u_i,v \right\rangle+\beta \right)^2\\
&=& \sum_{i=1}^k \lambda_i \alpha_i^2+ O_p\left(\frac{1}{m} \right)\\
&\leqslant & \lambda_1 + O_p\left(\frac{1}{m} \right).
\end{eqnarray*}
\end{minipage}}\vspace{0.3cm}\\
If $\beta$ is larger than $O_p(1/\sqrt{m})$,\\
\scalebox{0.7}{
\begin{minipage}{1\textwidth}
\begin{eqnarray*}
w^t \left(\sum_{i=1}^k \lambda_i u_i u_i^t + \mu v v^t \right) w &=& \sum_{i=1}^k \lambda_i \left(\alpha_i+\beta \left\langle v,u_i \right\rangle \right)^2+ \mu \left( \sum_{i=1}^k \alpha_i \left\langle u_i,v \right\rangle+\beta \right)^2\\
&\leqslant & \lambda_1 \left(\sum_{i=1}^k \left(\alpha_i+\beta \left\langle v,u_i \right\rangle \right)^2 \right) 
+ \mu \left( \sum_{i=1}^k \alpha_i \left\langle u_i,v \right\rangle+\beta \right)^2\\
&= & \lambda_1 \left(\sum_{i=1}^k \alpha_i^2+ 2 \sum_{i=1}^k \alpha_i \beta \left\langle v,u_i \right\rangle \right) 
+ \mu \left( 2 \beta \sum_{i=1}^k \alpha_i \left\langle u_i,v \right\rangle+ \beta^2 \right) + O_p\left( \frac{1}{m} \right)\\
&= & \lambda_1 + (\mu-\lambda_1) \beta^2 
+ 2 \mu  \beta \sum_{i=1}^k \alpha_i \left\langle u_i,v \right\rangle + O_p\left( \frac{1}{m} \right)\\
&\leqslant & \lambda_1 + O_p\left( \frac{1}{m} \right),
\end{eqnarray*} 
\end{minipage}}\vspace{0.3cm}\\
where the last two lines are obtained using $\sum_{i=1}^k \alpha_i^2 + 2 \sum_{i=1}^k \alpha_i \beta \left\langle u_i,v \right\rangle=1-\beta^2$ and
because 
\begin{eqnarray*}
P\left\lbrace (\mu-\lambda_1) \beta^2+2 \mu  \beta \sum_{i=1}^k \alpha_i \left\langle u_i,v \right\rangle <0 \right\rbrace \underset{m\rightarrow \infty}{\rightarrow} 1.
\end{eqnarray*} 
\noindent On the other hand, 
\begin{eqnarray*}
\lambda_{\max}\Bigg( \sum_{i=1}^k \lambda_i u_i u_i^t + \mu v v^t \Bigg) \geqslant u_1^t \left( \sum_{i=1}^k \lambda_i u_i u_i^t + \mu v v^t \right) u_1 = \lambda_1 +O_p\left(\frac{1}{m} \right).
\end{eqnarray*}
This concludes the proof.

\end{proof}

\subsubsection{Residual spike for perturbations of order $1$} \label{Asec:order1}
The proof of the first part of the Main Theorem \ref{ATH=Main} is in \cite{mainarticle}. 

\subsubsection{Decomposition of the difference matrix}\label{Asec:decomp}

As proposed in the Section \ref{Asec=sketch}, we can decompose the matrix $\hat{\hat{\Sigma}}_{P_k,X}^{-1/2} \hat{\hat{\Sigma}}_{P_k,Y} \hat{\hat{\Sigma}}_{P_k,X}^{-1/2}$ into a sum\\
\scalebox{0.7}{
\begin{minipage}{1\textwidth}
\begin{eqnarray*}
\hat{\hat{\Sigma}}_{P_k,X}^{-1/2} \hat{\hat{\Sigma}}_{P_k,Y} \hat{\hat{\Sigma}}_{P_k,X}^{-1/2}
&=&\I_m + \sum_{i=1}^k \left[ \hat{\hat{\Sigma}}_{P_k,X}^{-1/2}  \left(\hat{\hat{\theta}}_{P_k,Y,i}-1 \right) \hat{u}_{P_k,Y,i} \hat{u}_{P_k,Y,i}^t \hat{\hat{\Sigma}}_{P_k,X}^{-1/2} +\left(\frac{1}{\hat{\hat{\theta}}_{P_k,X,i}}-1\right) \hat{u}_{P_k,X,i} \hat{u}_{P_k,X,i}^t \right].
\end{eqnarray*}
\end{minipage}}\vspace{0.3cm}\\
Next, we define a rotation matrix $\hat{U}_{P_k,X}$ such that $\hat{U}_{P_k,X}^t \hat{u}_{P_k,X,i}=e_i$ and $\hat{U}_{P_k,X}^t \hat{u}_{P_k,Y,i}=\tilde{u}_{P_k}$ as in \cite[Theorem 4.2]{mainarticle2} of \cite{mainarticle2}. Because this rotation does not affect the eigenvalues, \\
\scalebox{0.7}{
\begin{minipage}{1\textwidth}
\begin{eqnarray*}
\lambda \left( \hat{\hat{\Sigma}}_{P_k,X}^{-1/2} \hat{\hat{\Sigma}}_{P_k,Y} \hat{\hat{\Sigma}}_{P_k,X}^{-1/2} \right)&=& \lambda \left( \hat{U}_{P_k,X}\hat{\hat{\Sigma}}_{P_k,X}^{-1/2}\hat{U}_{P_k,X}^t \hat{U}_{P_k,X} \hat{\hat{\Sigma}}_{P_k,Y}\hat{U}_{P_k,X}^t \hat{U}_{P_k,X} \hat{\hat{\Sigma}}_{P_k,X}^{-1/2}\hat{U}_{P_k,X}^t \right)\\
 &=& \lambda\left( \I_m + \sum_{i=1}^k \left[ \Sigma_{P_k,X}^{-1/2}  \left(\hat{\hat{\theta}}_{P_k,Y,i}-1 \right) \tilde{u}_{P_k,i} \tilde{u}_{P_k,i}^t \Sigma_{P_k,X}^{-1/2} +\left(\frac{1}{\hat{\hat{\theta}}_{P_k,X,i}}-1\right)e_i e_i^t \right] \right),
\end{eqnarray*}
\end{minipage}}\vspace{0.3cm}\\
where $\Sigma_{P_k,X}=\hat{U}_{P_k,X}^t \hat{\hat{\Sigma}}_{P_k,X} \hat{U}_{P_k,X}=\I_m + \sum_{i=1}^k \left(\hat{\hat{\theta}}_{P_k,X,i}-1 \right) e_i e_i^t$ and $\lambda()$ provides the eigenvalues of the matrices. 

\subsubsection{Pseudo invariant residual spike}\label{Asec:pseudoinvresidual}
With $\hat{\Sigma}_{P_k,X}= P_k^{1/2} W_X P_k^{1/2}$ and $\hat{\Sigma}_{P_k,Y}= P_k^{1/2} W_Y P_k^{1/2}$, we can define 
$\hat{\Sigma}_{\tilde{P}_i,X}= \tilde{P}_i^{1/2} W_X \tilde{P}_i^{1/2},$ where $\tilde{P}_i= \I_m + (\theta_i-1) e_i e_i^t$, and we can show that if $\theta_i$ is large or if we assume that the residual spikes of the perturbations of order $1$ are distinct, \\
\scalebox{0.7}{
\begin{minipage}{1\textwidth}
\begin{eqnarray*}
\lambda \left( \Sigma_{P_k,X}^{-1/2}  \left(\hat{\hat{\theta}}_{P_k,Y,i}-1 \right) \tilde{u}_{P_k,i} \tilde{u}_{P_k,i}^t \Sigma_{P_k,X}^{-1/2} +\left(\frac{1}{\hat{\hat{\theta}}_{P_k,X,i}}-1\right)e_i e_i^t \right)
= \lambda\left( \hat{\hat{\Sigma}}_{\tilde{P}_i,X}^{-1/2} \hat{\hat{\Sigma}}_{\tilde{P}_i,Y} \hat{\hat{\Sigma}}_{\tilde{P}_i,X}^{-1/2} \right)-1 + O_p\left( \frac{1}{m} \right).
\end{eqnarray*} 
\end{minipage}}\vspace{0.3cm}\\
The proof of this equality is computed in two steps. 
\begin{enumerate}
\item First we compute the non trivial eigenvalues and eigenvectors of 
$$\Sigma_{P_k,X}^{-1/2}\left(\hat{\hat{\theta}}_{P_k,Y,i}-1 \right) \tilde{u}_{P_k,i} \tilde{u}_{P_k,i}^t \Sigma_{P_k,X}^{-1/2}.$$
\item Then, using the Lemma \ref{ALemmalinearalg}, we establish the equality. 
\end{enumerate}
\begin{enumerate}
\item We define 
\begin{eqnarray*}
 \tilde{\Sigma}_{X} &=& \text{I}+ \sum_{i=1}^k(\hat{\theta}_{X,i}-1) \tilde{\tilde{u}}_i \tilde{\tilde{u}}_i^t,\\
 \tilde{\tilde{u}}_{i,1}&=& \tilde{u}_{1,i}.
\end{eqnarray*}
The vector $\tilde{\tilde{u}}_{i}$ is just $\hat{U}_Y^t \hat{u}_{X,i}$ or $\tilde{U} e_i$. Using the fact that for a matrix $M$, the non trivial eigenvector of $e_1 e_1^t M$ is $e_1$,  \\
\scalebox{0.78}{
\begin{minipage}{1\textwidth}
\begin{eqnarray*}
\lambda\Bigg(\Sigma_{X}^{-1/2} \left( (\hat{\theta}_{Y,1}-1) \tilde{u}_{1}\tilde{u}_{1}^t \right)\Sigma_{X}^{-1/2}\Bigg) &=& \lambda\Bigg(\tilde{\Sigma}_{X}^{-1/2} \left( (\hat{\theta}_{Y,1}-1) e_1 e_1^t \right) \tilde{\Sigma}_{X}^{-1/2}\Bigg)\\
&=& \lambda\Bigg( \left( (\hat{\theta}_{Y,1}-1) e_1 e_1^t \right) \tilde{\Sigma}_{X}^{-1}\Bigg)\\
&=&
e_1^t \left( (\hat{\theta}_{Y,1}-1) e_1 e_1^t \right) \tilde{\Sigma}_{X}^{-1} e_1 \\ 
&=&e_1^t  (\hat{\theta}_{Y,1}-1) \left(  \left( e_1 e_1^t \right)
+ \sum_{i=1}^k(\frac{1}{\hat{\theta}_{X,i}}-1) \tilde{\tilde{u}}_{i,1} e_1\tilde{\tilde{u}}_i^t \right) e_1\\
&=&(\hat{\theta}_{Y,1}-1) \left(  1
+ \sum_{i=1}^k(\frac{1}{\hat{\theta}_{X,i}}-1) \tilde{\tilde{u}}_{i,1}^2 \right)\\
&=&(\hat{\theta}_{Y,1}-1) \left(  1
+ \sum_{i=1}^k(\frac{1}{\hat{\theta}_{X,i}}-1) \tilde{u}_{1,i}^2 \right).
\end{eqnarray*}
\end{minipage}}\vspace{0.3cm}\\
The computation of the eigenvector leads to \\
\scalebox{0.8}{
\begin{minipage}{1\textwidth}
\begin{eqnarray*}
u\Bigg(\Sigma_{X}^{-1/2} \left( (\hat{\theta}_{Y,1}-1) \tilde{u}_{1}\tilde{u}_{1}^t \right)\Sigma_{X}^{-1/2}\Bigg) &\propto& \Sigma_X^{-1/2} u\Bigg( \left( (\hat{\theta}_{Y,1}-1) \tilde{u}_{1}\tilde{u}_{1}^t \right)\Sigma_{X}^{-1}\Bigg) \\
&\propto &  \Sigma_X^{-1/2} \tilde{u}_1\\
&\propto & \left( \frac{\tilde{u}_{1,1}}{\sqrt{\hat{\theta}_{X,1}}},\frac{\tilde{u}_{1,2}}{\sqrt{\hat{\theta}_{X,2}}},...,\frac{\tilde{u}_{1,k}}{\sqrt{\hat{\theta}_{X,k}}}, \tilde{u}_{1,k+1},...,\tilde{u}_{1,m} \right).
\end{eqnarray*} 
\end{minipage}}\vspace{0.3cm}\\
Because the previous eigenvector is not standardised, we compute its norm,\\
\scalebox{0.68}{
\begin{minipage}{1\textwidth}
\begin{eqnarray*}
\left|\left| \left( \frac{\tilde{u}_{1,1}}{\sqrt{\hat{\theta}_{X,1}}},\frac{\tilde{u}_{1,2}}{\sqrt{\hat{\theta}_{X,2}}},...,\frac{\tilde{u}_{1,k}}{\sqrt{\hat{\theta}_{X,k}}}, \tilde{u}_{1,k+1},...,\tilde{u}_{1,m} \right) \right|\right|^2 &=& 
\frac{\tilde{u}_{1,1}^2}{\hat{\theta}_{X,1}}+\frac{\tilde{u}_{1,2}^2}{\hat{\theta}_{X,2}}+...+\frac{\tilde{u}_{1,k}^2}{\hat{\theta}_{X,k}}+ \tilde{u}_{1,k+1}^2+...+\tilde{u}_{1,m}^2\\
&=&  1
+ \sum_{i=1}^k(\frac{1}{\hat{\theta}_{X,i}}-1) \tilde{u}_{1,i}^2 \\
&=& \frac{ \lambda\Bigg(\hat{\hat{\Sigma}}_{X}^{-1/2} \left( (\hat{\theta}_{Y,1}-1) \hat{u}_{Y,1}\hat{u}_{Y,1}^t \right) \hat{\hat{\Sigma}}_{X}^{-1/2}\Bigg)}{\hat{\theta}_{Y,1}-1}. 
\end{eqnarray*}
\end{minipage}}\vspace{0.3cm}\\
We conclude the first part with the two formulas:\\
\scalebox{0.85}{
\begin{minipage}{1\textwidth}
\begin{eqnarray*}
 \lambda\Bigg(\Sigma_{X}^{-1/2} \left( (\hat{\theta}_{Y,1}-1) \tilde{u}_{1}\tilde{u}_{1}^t \right)\Sigma_{X}^{-1/2}\Bigg)
 &=&\left(\hat{\theta}_{Y,1}-1 \right) \left( \sum_{i=1}^k \left(\frac{1}{\hat{\theta}_{X,i}} -1 \right) \tilde{u}_{1,i}^2+1 \right)
\end{eqnarray*} 
\end{minipage}}\vspace{0.3cm}\\
and \\
\scalebox{0.68}{
\begin{minipage}{1\textwidth}
\begin{eqnarray*}
&&u\Bigg(\Sigma_{X}^{-1/2} \left( (\hat{\theta}_{Y,1}-1) \tilde{u}_{1}\tilde{u}_{1}^t \right) \Sigma_{X}^{-1/2}\Bigg) \\
&&\hspace{1cm}=\frac{\sqrt{\hat{\theta}_{Y,1}-1}}{\sqrt{ \lambda\Bigg(\hat{\hat{\Sigma}}_{X}^{-1/2} \left( (\hat{\theta}_{Y,1}-1) \hat{u}_{Y,1}\hat{u}_{Y,1}^t \right) \hat{\hat{\Sigma}}_{X}^{-1/2}\Bigg)}}\left( \frac{\tilde{u}_{1,1}}{\sqrt{\hat{\theta}_{X,1}}},\frac{\tilde{u}_{1,2}}{\sqrt{\hat{\theta}_{X,2}}},...,\frac{\tilde{u}_{1,k}}{\sqrt{\hat{\theta}_{X,k}}},\tilde{u}_{1,k+1},...,\tilde{u}_{1,m} \right).
\end{eqnarray*}
\end{minipage}}\vspace{0.3cm}\\

\item The second part uses the Lemma \ref{ALemmalinearalg} to establish the following relation 
\begin{eqnarray*}
&&\Sigma_{P_k,X}^{-1/2}\left(\hat{\hat{\theta}}_{P_k,Y,1}-1 \right) \tilde{u}_{P_k,1} \tilde{u}_{P_k,1}^t \Sigma_{P_k,X}^{-1/2} +\left(\frac{1}{\hat{\hat{\theta}}_{P_k,X,1}}-1\right)e_1 e_1^t \\
&&\hspace{2cm}= \lambda \left( \hat{\hat{\Sigma}}_{P_1,X}^{-1/2} \hat{\hat{\Sigma}}_{P_1,X} \hat{\hat{\Sigma}}_{P_1,X}^{-1/2} \right)-1+O_p\left( \frac{1}{m}\right).
\end{eqnarray*}

\begin{enumerate}
\item We start with the order $1$, \\
\scalebox{0.68}{
\begin{minipage}{1\textwidth}
\begin{eqnarray*}
\lambda \left( \hat{\hat{\Sigma}}_{P_1,X}^{-1/2} \hat{\hat{\Sigma}}_{P_1,X} \hat{\hat{\Sigma}}_{P_1,X}^{-1/2} \right)
&=&\lambda\Bigg(\Sigma_{P_1,X}^{-1/2} \left( (\hat{\theta}_{P_1,Y,1}-1) \tilde{u}_{P_1,1}\tilde{u}_{P_1,1}^t \right) \Sigma_{P_1,X}^{-1/2}+ \left(\frac{1}{\hat{\theta}_{P_1,X,1}}-1\right) e_1 e_1^t + {\rm I} \Bigg)  \\
&=&\lambda\Bigg( \hat{\eta}_{P_1} u_{\hat{\eta}_{P_1}} u_{\hat{\eta}_{P_1}}^t+ \left(\frac{1}{\hat{\theta}_{P_1,X,1}}-1\right) e_1 e_1^t \Bigg)+1,
\end{eqnarray*}
\end{minipage}}\vspace{0.3cm}\\

where,\\
\scalebox{0.68}{
\begin{minipage}{1\textwidth}
\begin{eqnarray*}
\hat{\eta}_{P_1}&=& \lambda\Bigg(\Sigma_{X}^{-1/2} \left( (\hat{\theta}_{Y,1}-1) \tilde{u}_{1}\tilde{u}_{1}^t \right)\Sigma_{X}^{-1/2}\Bigg)
= \frac{\hat{\theta}_{P_1,Y,1}-1}{\hat{\theta}_{P_1,X,1}} \tilde{u}_{P_1,1,1}^2 +\left(\hat{\theta}_{P_1,Y,1}-1 \right) \left(1-\tilde{u}_{P_1,1,1}^2 \right),\\
u_{\hat{\eta}_{P_1}}&=&u\Bigg(\Sigma_{X}^{-1/2} \left( (\hat{\theta}_{Y,1}-1) \tilde{u}_{1}\tilde{u}_{1}^t \right)\Sigma_{X}^{-1/2}\Bigg)
= \frac{\sqrt{\hat{\theta}_{P_1,Y,1}-1} \left( \frac{\tilde{u}_{P_1,1,1}}{\sqrt{\hat{\theta}_{P_1,X,1}}},\tilde{u}_{P_1,1,2},...,\tilde{u}_{P_1,1,m} \right)}{\sqrt{\lambda\Bigg(\Sigma_{X}^{-1/2} \left( (\hat{\theta}_{Y,1}-1) \tilde{u}_{1}\tilde{u}_{1}^t \right)\Sigma_{X}^{-1/2}\Bigg)}}.
\end{eqnarray*}
\end{minipage}}\vspace{0.3cm}\\
By the Lemma \ref{ALemmalinearalg}, the non trivial eigenvalues are functions of different parameters. Using a similar notation to the Lemma we set
\begin{eqnarray*}
a_{P_1}&=&\frac{1}{\hat{\theta}_{P_1,X,1}}-1,\\
b_{P_1}&=&\frac{\hat{\theta}_{P_1,Y,1}-1}{\hat{\theta}_{P_1,X,1}} \tilde{u}_{P_1,1,1}^2 +\left(\hat{\theta}_{P_1,Y,1}-1 \right) \left(1-\tilde{u}_{P_1,1,1}^2 \right),\\
b_{P_1} w_{P_1}^2&=&\left(\hat{\theta}_{P_1,Y,1}-1\right)  \frac{\tilde{u}_{P_1,1,1}^2}{\hat{\theta}_{P_1,X,1}}.
\end{eqnarray*}
The Lemma \ref{ALemmalinearalg} provides a function $g$ such that
\begin{eqnarray*}
\lambda\Bigg( \hat{\eta}_{P_1} u_{\hat{\eta}_{P_1}} u_{\hat{\eta}_{P_1}}^t+ \left(\frac{1}{\hat{\theta}_{P_1,X,1}}-1\right) e_1 e_1^t \Bigg)=g^{\pm}(a_{P_1},b_{P_1},b_{P_1} w_{P_1}^2).
\end{eqnarray*}
Therefore
\begin{eqnarray*}
\lambda \left( \hat{\hat{\Sigma}}_{P_1,X}^{-1/2} \hat{\hat{\Sigma}}_{P_1,X} \hat{\hat{\Sigma}}_{P_1,X}^{-1/2} \right)
&=&g^{\pm}(a_{P_1},b_{P_1},b_{P_1} w_{P_1}^2).
\end{eqnarray*}

\item  For perturbations of order $k$, \\
\scalebox{0.78}{
\begin{minipage}{1\textwidth}
\begin{eqnarray*}
&&\lambda\Bigg(\hat{\hat{\Sigma}}_{P_k,X}^{-1/2} \left( (\hat{\theta}_{P_k,Y,1}-1) \hat{u}_{P_k,Y,1}\hat{u}_{P_k,Y,1}^t \right) \hat{\hat{\Sigma}}_{P_k,X}^{-1/2}+ \left(\frac{1}{\hat{\theta}_{P_k,X,1}}-1\right) \hat{u}_{P_k,X,1}\hat{u}_{P_k,X,1}^t \Bigg) \\
&&\hspace{1cm} = \lambda\Bigg(\Sigma_{P_k,X}^{-1/2} \left( (\hat{\theta}_{P_k,Y,1}-1) \tilde{u}_{P_k,1}\tilde{u}_{P_k,1}^t \right) \Sigma_{P_k,X}^{-1/2}+ \left(\frac{1}{\hat{\theta}_{P_k,X,1}}-1\right) e_1 e_1^t \Bigg) \\
&&\hspace{1cm} = \lambda\Bigg(\hat{\eta}_{P_k} u_{\hat{\eta}_{P_k}} u_{\hat{\eta}_{P_k}}^t + \left(\frac{1}{\hat{\theta}_{P_k,X,1}}-1\right) e_1 e_1^t \Bigg), 
\end{eqnarray*}
\end{minipage}}\vspace{0.3cm}\\
where\\
\scalebox{0.76}{
\begin{minipage}{1\textwidth}
\begin{eqnarray*}
\hat{\eta}_{P_k,1}&=& \left(\hat{\theta}_{P_k,Y,1}-1 \right) \left( \sum_{i=1}^k \left(\frac{1}{\hat{\theta}_{P_k,X,i}} -1 \right) \tilde{u}_{P_k,1,i}^2+1 \right)\\
&=& \frac{\hat{\theta}_{P_k,Y,1}-1}{\hat{\theta}_{P_k,X,1}} \left(\sum_{i=1}^k \tilde{u}_{P_k,1,i}^2 \right) +\left(\hat{\theta}_{P_k,Y,1}-1 \right) \left(1-\sum_{i=1}^k\tilde{u}_{P_k,1,i}^2 \right)+O_p\left( \frac{1}{m}\right),\\
u_{\hat{\eta}_{P_k,1}}&=&  \frac{\sqrt{\hat{\theta}_{P_k,Y,1}-1} \left( \frac{\tilde{u}_{P_k,1,1}}{\sqrt{\hat{\theta}_{P_k,X,1}}},\frac{\tilde{u}_{P_k,1,2}}{\sqrt{\hat{\theta}_{P_k,X,2}}},...,\frac{\tilde{u}_{P_k,1,k}}{\sqrt{\hat{\theta}_{P_k,X,k}}},\tilde{u}_{P_k,1,k+1},...,\tilde{u}_{P_k,1,m} \right) }{\sqrt{ \lambda\Bigg(\hat{\hat{\Sigma}}_{P_k,X}^{-1/2} \left( (\hat{\theta}_{P_k,Y,1}-1) \hat{u}_{P_k,Y,1}\hat{u}_{P_k,Y,1}^t \right) \hat{\hat{\Sigma}}_{P_k,X}^{-1/2}\Bigg)}}.
\end{eqnarray*}
\end{minipage}}\vspace{0.3cm}\\
As previously by the Lemma \ref{ALemmalinearalg}, the non-trivial eigenvalues are functions of different parameters. Using a similar notation to the Lemma we set\\
\scalebox{0.75}{
\begin{minipage}{1\textwidth}
\begin{eqnarray*}
a_{P_k}&=&\frac{1}{\hat{\theta}_{P_k,X,1}}-1,\\
b_{P_k}&=&\frac{\hat{\theta}_{P_k,Y,1}-1}{\hat{\theta}_{P_k,X,1}} \left(\sum_{i=1}^k \tilde{u}_{P_k,1,i}^2 \right) +\left(\hat{\theta}_{P_k,Y,1}-1 \right) \left(1-\sum_{i=1}^k\tilde{u}_{P_k,1,i}^2 \right)+O_p\left( \frac{1}{m}\right),\\
b_{P_k} w_{P_k,1}^2&=&\left(\hat{\theta}_{P_k,Y,1}-1\right)  \frac{\tilde{u}_{P_k,1,1}^2}{\hat{\theta}_{P_k,X,1}}=\left(\hat{\theta}_{P_k,Y,1}-1\right)  \frac{\sum_{i=1}^k \tilde{u}_{P_k,1,i}^2}{\hat{\theta}_{P_k,X,1}}+O_p\left( \frac{1}{m} \right).
\end{eqnarray*}
\end{minipage}}\vspace{0.3cm}\\
The Lemma \ref{ALemmalinearalg} provides the function $g$ such that \\
\scalebox{0.9}{
\begin{minipage}{1\textwidth}
\begin{eqnarray*}
\lambda\Bigg( \hat{\eta}_{P_k} u_{\hat{\eta}_{P_k}} u_{\hat{\eta}_{P_k}}^t+ \left(\frac{1}{\hat{\theta}_{P_k,X,1}}-1\right) e_1 e_1^t  \Bigg)
&=& g^{\pm}(a_{P_k},b_{P_k},b_{P_k} w_{P_k,1}^2).
\end{eqnarray*}
\end{minipage}}\vspace{0.3cm}\\
\item Finally we show that 
$$g^{\pm}(a_{P_k},b_{P_k},b_{P_k} w_{P_k,1}^2)=g^{\pm}(a_{P_1},b_{P_1},b_{P_1} w_{P_1,1}^2)+O_p\left( \frac{1}{m}\right).$$
By the Invariance Theorems,
\begin{eqnarray*}
&&a_{P_k}=a_{P_1}+O_p\left( \frac{1}{\theta m}\right),\\
&&b_{P_k}=b_{P_1}+O_p\left(\frac{1}{m} \right),\\
&&b_{P_k} w_{P_k,1}^2 = b_{P_1} w_{P_1}^2 +O_p\left( \frac{1}{m}\right).
\end{eqnarray*}
Moreover, the three values do not converge to $0$.\\
Because we know from Lemma \ref{ALemmalinearalg} that $g$ is continuous,
$$g^{\pm}(x,y,z)=\frac{1}{2} \left(x+y \pm \sqrt{4 x z +(x-y)^2}\right).$$
This function is Lipschitz if $4 x z +(x-y)^2$ is not closed to $0$. The reader can show that the perturbation creates two residual spikes different from $1$ when $\theta_1$ is detectable. (In other cases the covariance matrices are the same (CANNOT BE DISTINGUISHED?).) In particular we can show that when $\theta_1$ is large, the pseudo residual spike is distinct from $1$.\\
Therefore, using this property we conclude,\\
\scalebox{0.77}{
\begin{minipage}{1\textwidth}
\begin{eqnarray*}
&&\hspace{0cm} \left| \left(\lambda \left( \hat{\hat{\Sigma}}_{P_1,X}^{-1/2} \hat{\hat{\Sigma}}_{P_1,X} \hat{\hat{\Sigma}}_{P_1,X}^{-1/2} \right)-1 \right)- \left( \lambda\Bigg( \hat{\eta}_{P_k,1} u_{\hat{\eta}_{P_k,1}} u_{\hat{\eta}_{P_k,1}}^t+ \left(\frac{1}{\hat{\theta}_{P_k,X,1}}-1\right) e_1 e_1^t  \Bigg) \right) \right|\\
&&\hspace{2cm}= \left|g^{\pm}(a_{P_1},b_{P_1},b_{P_1}w_{P_1,1}^2) - g^{\pm}(a_{P_k},b_{P_k},b_{P_k} w_{P_k,1}^2)\right|\\
&&\hspace{2cm}=O_p\left( \frac{1}{m} \right)
\end{eqnarray*}
\end{minipage}}\vspace{0.3cm}\\
\begin{Rem}
The hypothesis assuming that $\lambda^+-\lambda^- \not \rightarrow 0$ is evident, except when $n_X,n_Y>>m$. Nevertheless, the Main Theorem \ref{ATH=Main} assumes proportional values and so avoids this critical case. 
\end{Rem}

\end{enumerate}
\end{enumerate}

\subsubsection{Pseudo residual eigenvectors} \label{Asec:pseudoinvresidualvector} 
Knowing the pseudo residual spike, it is not difficult to find the corresponding pseudo residual eigenvector. For $s=1,2,...,k$, suppose 
\begin{eqnarray*}
w_s^{\pm}&=&u\Bigg(\Sigma_{P_k,X}^{-1/2} \left( (\hat{\theta}_{P_k,Y,s}-1) \tilde{u}_{P_k,s}\tilde{u}_{P_k,Y,s}^t \right) \Sigma_{P_k,X}^{-1/2}+ \left(\frac{1}{\hat{\theta}_{P_k,X,s}}-1\right) e_s e_s^t \Bigg)
\end{eqnarray*}
are the pseudo residual eigenvectors corresponding to the eigenvalues 
\begin{eqnarray*}
\hat{\zeta}_s^{\pm}&=& \lambda \Bigg(\Sigma_{P_k,X}^{-1/2} \left( (\hat{\theta}_{P_k,Y,s}-1) \tilde{u}_{P_k,s}\tilde{u}_{P_k,Y,s}^t \right) \Sigma_{P_k,X}^{-1/2}+ \left(\frac{1}{\hat{\theta}_{P_k,X,s}}-1\right) e_s e_s^t \Bigg).
\end{eqnarray*}
We define
\begin{eqnarray*}
\zeta_{\infty}^{\pm}(\theta_s)&=& \underset{m \rightarrow \infty} {\lim}\hat{\zeta}_s^{\pm},\\
 \zeta_{\infty}^{\pm} &=& \underset{\theta_s,m \rightarrow \infty} {\lim} \hat{\zeta}_s^{\pm}.
\end{eqnarray*}
Then, \\
\scalebox{0.9}{
\begin{minipage}{1\textwidth}
\begin{eqnarray*}
w_{s,s}^{\pm}&=&\frac{\sqrt{\hat{\theta}_{P_k,X,s}}}{{\rm Norm}_s^{\pm} \sqrt{\hat{\theta}_{P_k,Y,s}-1} \tilde{u}_{P_k,s,s}   }\left(\hat{\zeta}_s^{\pm}- \left(\hat{\theta}_{P_k,Y,s}-1\right)\left(1-\hat{\alpha}^2_{P_k,s}  \right)  \right) + O_p\left( \frac{1}{m}\right)\\
&=&  \frac{\sqrt{\theta_s}}{{\rm Norm}_s^{\pm} \sqrt{\theta_s-1} \alpha_s } \left( \zeta_{\infty}^{\pm}(\theta_s) - \left( \theta_s-1 \right) \left(1-\alpha_s^2 \right) \right)+ O_p\left( \frac{1}{\sqrt{m}} \right)\\
&=&\frac{\left(\zeta_{\infty}^{\pm}-2 \left(M_2 -1\right)\right)}{\sqrt{\left(\zeta_{\infty}^{\pm}-2 \left(M_2 -1\right)\right)^2+2 \left(M_2 -1\right) }}+ O_p\left( \frac{1}{\sqrt{m}} \right) + o_{p;\theta_s}\left( 1 \right),
\end{eqnarray*}\end{minipage}}\vspace{0.3cm}\\
\scalebox{0.76}{
\begin{minipage}{1\textwidth}
\begin{eqnarray*}
w_{s,2:m \setminus s}^{\pm}&=&\frac{\sqrt{\hat{\theta}_{P_k,Y,s}-1} \left( \frac{ \tilde{u}_{P_k,s,1}}{\sqrt{\hat{\theta}_{P_k,X,1}}},...,\frac{ \tilde{u}_{P_k,s,s-1}}{\sqrt{\hat{\theta}_{P_k,X,s-1}}}
\frac{ \tilde{u}_{P_k,s,s+1}}{\sqrt{\hat{\theta}_{P_k,X,s+1}}},...,\frac{\tilde{u}_{P_k,s,k}}{\sqrt{\hat{\theta}_{P_k,X,k}}},\tilde{u}_{P_k,s,k+1},...,\tilde{u}_{P_k,s,m} \right)}{{\rm Norm}_s^{\pm}},\\
&=&\frac{\sqrt{\hat{\theta}_{P_k,Y,1}-1} \left(\frac{\tilde{u}_{P_k,s,1}}{\sqrt{\hat{\theta}_{P_k,X,1}}},...,\frac{\tilde{u}_{P_k,s,s-1}}{\sqrt{\hat{\theta}_{P_k,X,s-1}}},\frac{\tilde{u}_{P_k,s,s+1}}{\sqrt{\hat{\theta}_{P_k,X,s+1}}},...,\frac{\tilde{u}_{P_k,s,k}}{\sqrt{\hat{\theta}_{P_k,X,k}}},\tilde{u}_{P_k,s,k+1},...,\tilde{u}_{P_k,s,m}\right) }{\sqrt{\left(\zeta_{\infty}^{\pm}-2 \left(M_2 -1\right)\right)^2+2 \left(M_2 -1\right)}+O_p\left(\frac{1}{\sqrt{m}} \right)+o_{p;\theta_s} (1)},
\end{eqnarray*}
\end{minipage}}\vspace{0.3cm}\\
\scalebox{0.76}{
\begin{minipage}{1\textwidth}
\begin{eqnarray*}
\left({\rm Norm}_s^{\pm}\right)^2&=& \frac{\hat{\theta}_{P_k,X,s}\left(\hat{\zeta}_s^{\pm}- \left(\hat{\theta}_{P_k,Y,s}-1\right)\left(1-\hat{\alpha}^2_{P_k,s}  \right)  \right)^2 }{\left(\hat{\theta}_{P_k,Y,s}-1\right) \tilde{u}_{P_k,s,s}^2   }
+ \left(\hat{\theta}_{P_k,Y,s}-1\right) \left(1-\hat{\alpha}^2_{P_k,s}  \right)+ O_p\left( \frac{1}{m}\right)\\
&=&\frac{\theta_s}{\left(\theta_s-1\right) \alpha_s^2 } \left( \zeta_{\infty}^{\pm}\left( \theta_s\right) - \left( \theta_s-1 \right) \left(1-\alpha_s^2 \right) \right)^2+ \left( \theta_s-1 \right) \left(1-\alpha_s^2 \right)+ O_p\left( \frac{1}{\sqrt{m}} \right).
\end{eqnarray*}
\end{minipage}}\vspace{0.3cm}\\
We used the fact that the rate convergence of $\hat{\theta}_{P_k,X,s}$, $\hat{\theta}_{P_k,Y,s}$ and $\hat{\alpha}_{P_k,s}^2$  is in $1/\sqrt{m}$.\\
Moreover, when $\theta$ is large, $\left(1-\alpha_s^2 \right)=\frac{2\left(M_2-1\right)}{\theta_s}+O_p\left(1/\theta_s^2\right)$.

\subsubsection{Dimension reduction}\label{Asec:dimreduce}
The previous parts showed that the non-trivial eigenvalues of
\begin{eqnarray*}
\hat{\hat{\Sigma}}_{P_k,X}^{-1/2} \hat{\hat{\Sigma}}_{P_k,Y} \hat{\hat{\Sigma}}_{P_k,X}^{-1/2}-\I_m
\end{eqnarray*}
 are the same as the eigenvalues of
\begin{eqnarray*}
\sum_{i=1}^k \hat{\zeta}_i^+ w_i^+ {w_i^+}^t +\sum_{i=1}^k \hat{\zeta}_i^- w_i^- {w_i^-}^t \in \mathbf{R}^m \times  \mathbf{R}^m.
\end{eqnarray*}
We can use \ref{ALemreducedim} to show that for all non null eigenvalues,

\begin{eqnarray*}
\lambda_i\left(\sum_{i=1}^k \hat{\zeta}_i^+ w_i^+ {w_i^+}^t +\sum_{i=1}^k \hat{\zeta}_i^- w_i^- {w_i^-}^t \right) = \lambda_i\left(H \right),
\end{eqnarray*}
where \\
\scalebox{0.75}{
\begin{minipage}{1\textwidth}
\begin{eqnarray*}
H&=& \begin{pmatrix} H^+  & {H^b} \\ {H^b}^t & H^-
 \end{pmatrix}, \\
H^{\pm} &=& \begin{pmatrix}
\hat{\zeta}_1^{\pm} & \sqrt{\hat{\zeta}_1^{\pm} \hat{\zeta}_2^{\pm}} \left\langle w_1^{\pm},w_2^{\pm} \right\rangle & \sqrt{\hat{\zeta}_1^{\pm} \hat{\zeta}_3^{\pm}} \left\langle w_1^{\pm},w_3^{\pm} \right\rangle & \cdots & \sqrt{\hat{\zeta}_k^{\pm} \hat{\zeta}_2^{\pm}} \left\langle w_1^{\pm},w_k^{\pm} \right\rangle \\ 
\sqrt{\hat{\zeta}_2^{\pm} \hat{\zeta}_1^{\pm}} \left\langle w_2^{\pm},w_1^{\pm} \right\rangle  & \hat{\zeta}_2^{\pm} & \sqrt{\hat{\zeta}_2^{\pm} \hat{\zeta}_3^{\pm}} \left\langle w_2^{\pm},w_3^{\pm} \right\rangle & \cdots  & \sqrt{\hat{\zeta}_2^{\pm} \hat{\zeta}_k^{\pm}} \left\langle w_2^{\pm},w_k^{\pm} \right\rangle \\ 
\sqrt{\hat{\zeta}_3^{\pm} \hat{\zeta}_1^{\pm}} \left\langle w_3^{\pm},w_1^{\pm} \right\rangle &\sqrt{\hat{\zeta}_3^{\pm} \hat{\zeta}_2^{\pm}} \left\langle w_3^{\pm},w_2^{\pm} \right\rangle  & \hat{\zeta}_3^{\pm} & \cdots  & \sqrt{\hat{\zeta}_3^{\pm} \hat{\zeta}_k^{\pm}} \left\langle w_3^{\pm},w_k^{\pm} \right\rangle \\ 
\vdots & \vdots & \ddots & \ddots &  \vdots \\ 
 \sqrt{\hat{\zeta}_k^{\pm} \hat{\zeta}_1^{\pm}} \left\langle w_k^{\pm},w_1^{\pm} \right\rangle & \sqrt{\hat{\zeta}_k^{\pm} \hat{\zeta}_2^{\pm}} \left\langle w_k^{\pm},w_2^{\pm} \right\rangle & \sqrt{\hat{\zeta}_k^{\pm} \hat{\zeta}_3^{\pm}} \left\langle w_k^{\pm},w_3^{\pm} \right\rangle & \cdots  & \hat{\zeta}_k^{\pm}  \\ 
\end{pmatrix} ,\\
H^b &=& \begin{pmatrix}
0 & \sqrt{\hat{\zeta}_1^+ \hat{\zeta}_2^-} \left\langle w_1^+,w_2^- \right\rangle & \sqrt{\hat{\zeta}_1^+ \hat{\zeta}_3^-} \left\langle w_1^+,w_3^- \right\rangle & \cdots & \sqrt{\hat{\zeta}_k^+ \hat{\zeta}_2^-} \left\langle w_1^+,w_k^- \right\rangle \\ 
\sqrt{\hat{\zeta}_2^+ \hat{\zeta}_1^-} \left\langle w_2^+,w_1^- \right\rangle  & 0 & \sqrt{\hat{\zeta}_2^+ \hat{\zeta}_3^-} \left\langle w_2^+,w_3^- \right\rangle & \cdots  & \sqrt{\hat{\zeta}_2^+ \hat{\zeta}_k^-} \left\langle w_2^+,w_k^- \right\rangle \\ 
\sqrt{\hat{\zeta}_3^+ \hat{\zeta}_1^-} \left\langle w_3^+,w_1^- \right\rangle &\sqrt{\hat{\zeta}_3^+ \hat{\zeta}_2^-} \left\langle w_3^+,w_2^- \right\rangle  & 0 & \cdots  & \sqrt{\hat{\zeta}_3^+ \hat{\zeta}_k^-} \left\langle w_3^+,w_k^- \right\rangle \\ 
\vdots & \vdots & \ddots & \ddots &  \vdots \\ 
 \sqrt{\hat{\zeta}_k^+ \hat{\zeta}_1^-} \left\langle w_k^+,w_1^- \right\rangle & \sqrt{\hat{\zeta}_k^+ \hat{\zeta}_2^-} \left\langle w_k^+,w_2^- \right\rangle & \sqrt{\hat{\zeta}_k^+ \hat{\zeta}_3^-} \left\langle w_k^+,w_3^- \right\rangle & \cdots  & 0 \\ 
\end{pmatrix}.
\end{eqnarray*}
\end{minipage}}\vspace{0.3cm}\\
Then, we can use Lemma \ref{ALemneglectvalue} to argue that 
\begin{eqnarray*}
\lambda_{\max} \left( H \right) &=&  \lambda_{\max} \left( H^+ \right) + O_p\left( \frac{1}{m}\right),\\
\lambda_{\min} \left( H \right) &=&  \lambda_{\max} \left( H^- \right) + O_p\left( \frac{1}{m}\right).
\end{eqnarray*}
We will see that the covariance between all the entries of $H^+$ is null. However, the entries of $H^+$ and $H^-$ are correlated. This step is very useful to avoid the need to study this correlation.

\subsubsection{Elements of H} \label{Asec:elementH}
Computation of the distribution of the entries of $H$ requires the distributions of $\hat{\zeta}_i^{\pm}$ and $\left\langle w_i^{\pm},w_j^{\pm} \right\rangle$, for $i,j=1,2,...,k$ with $j \neq i$. By the pseudo Invariance (Section \ref{Asec:pseudoinvresidual}), 
$$\hat{\zeta}_i^{\pm} = \lambda^{\pm} \left( \hat{\hat{\Sigma}}_{\tilde{P}_i,X}^{-1/2} \hat{\hat{\Sigma}}_{\tilde{P}_i,Y} \hat{\hat{\Sigma}}_{\tilde{P}_i,X}^{-1/2} \right)-1 + O_p\left( \frac{1}{m} \right). $$
Therefore,  using the Section \ref{Asec:order1}, we obtain the two first moments of the diagonal elements.\\
The off-diagonal terms are more difficult to estimate and we assume that Assumptions \ref{AAss=theta}(A2) and (A3) hold.\\
First, we will express 
$\left\langle w_i^{\pm},w_j^{\pm} \right\rangle$ as a function of the usual statistics when all the eigenvalues are of order $\theta$. Then, we will compute its two first moments. Finally, an argument similar to Lemma \ref{ALemneglectvalue} leads to a result for all perturbations.

\begin{Rem}\ \\
Using \cite[Theorem 4.2]{mainarticle2}, we invite the reader to show that the off diagonal terms are of order $O_p(1/\sqrt{m})$ when at least one eigenvalue $\theta_i$ or $\theta_j$ is finite (Assumption \ref{AAss=theta}(A4) ).
\end{Rem}

\paragraph*{Formula}

Suppose $k>1$ and \\
\scalebox{0.8}{
\begin{minipage}{1\textwidth}
\begin{eqnarray*}
w_i^\pm &=&u\Bigg(\Sigma_{P_k,X}^{-1/2} \left( (\hat{\hat{\theta}}_{P_k,Y,i}-1) \tilde{u}_{P_k,i}\tilde{u}_{P_k,Y,i}^t \right) \Sigma_{P_k,X}^{-1/2}+ \left(\frac{1}{\hat{\hat{\theta}}_{P_k,X,i}}-1\right) e_i e_i^t \Bigg)\\
\end{eqnarray*}
We want to prove the following formula :
\begin{eqnarray*}
\left\langle w_s^{\pm},w_t^{\pm} \right\rangle  
 &=& \frac{\sqrt{\theta_s \theta_t} }{\left(\zeta_{\infty}^{\pm}-2 M_2+1 \right)^2+2 \left( M_2 -1\right)} \times \\
 &&\hspace{0.5cm}\Bigg( \sum_{p=k+1}^m \hat{u}_{P_k,Y,s,p}\hat{u}_{P_k,Y,t,p}+\sum_{p=k+1}^m\hat{u}_{P_k,X,s,p}\hat{u}_{P_k,X,t,p} \\
&&\hspace{1.5cm}- 
\sum_{p=k+1}^m \hat{u}_{P_k,Y,s,p}\hat{u}_{P_k,X,t,p}-
\sum_{p=k+1}^m  \hat{u}_{P_k,Y,t,p}\hat{u}_{P_k,X,s,p}  \\ 
&&\hspace{1.5cm}- \left( \hat{u}_{P_k,X,t,s}+\hat{u}_{P_k,Y,s,t}\right) \left(\tilde{\alpha}^2_{s}-\tilde{\alpha}^2_{t}-
\left(\zeta_{\infty}^{\pm}-2 \left(M_2 -1\right)\right)
 \left( \frac{1}{\theta_t}  -\frac{1}{\theta_s} \right)  \right) \Bigg)\\
 &&\hspace{0.5cm}+ o_{p;m,\theta}\left( \frac{1}{m^{1/2}} \right), 
\end{eqnarray*}\end{minipage}}\vspace{0.3cm}\\
where $\lim_{m,\theta \rightarrow \infty} \frac{o_{p;m,\theta}\left( \frac{1}{m^{1/2}} \right)}{\frac{1}{m^{1/2}}} =0$ with probability tending to $1$.

\begin{itemize}
\item When discussing pseudo residual eigenvectors (see Section \ref{Asec:pseudoinvresidualvector}) we proved that assuming $\epsilon = O_p\left( \frac{1}{m^{1/2}} \right) + o_{p;\theta}\left( 1 \right)$,\\
\scalebox{0.88}{
\begin{minipage}{1\textwidth}
\begin{eqnarray*}
 \left(w_{s,s}\right)^\pm
&=&
\frac{\left(\zeta_{\infty}^{\pm}-2 \left(M_2 -1\right)\right)}{\sqrt{\left(\zeta_{\infty}^{\pm}-2 \left(M_2 -1\right)\right)^2+2 \left(M_2 -1\right) }}+\epsilon,\\
 w_{s,k+1:m}^{\pm}
&=&
\frac{\sqrt{\hat{\theta}_{P_k,Y,s}-1}  \tilde{u}_{P_k,s,k+1:m}}{\sqrt{\left(\zeta_{\infty}^{\pm}-2 \left(M_2 -1\right)\right)^2+2 \left(M_2 -1\right)}+\epsilon},\\
w_{s,1:k \setminus s}^{\pm}&=&
\frac{\sqrt{\hat{\theta}_{P_k,Y,s}-1} \left(\frac{\tilde{u}_{P_k,s,1}}{\sqrt{\hat{\theta}_{P_k,X,1}}},...,\frac{\tilde{u}_{P_k,s,s-1}}{\sqrt{\hat{\theta}_{P_k,X,s-1}}},\frac{\tilde{u}_{P_k,s,s+1}}{\sqrt{\hat{\theta}_{P_k,X,s+1}}},...,\frac{\tilde{u}_{P_k,s,k}}{\sqrt{\hat{\theta}_{P_k,X,k}}}\right) }{\sqrt{\left(\zeta_{\infty}^{\pm}-2 \left(M_2 -1\right)\right)^2+2 \left(M_2 -1\right)}+ \epsilon} .\\
\end{eqnarray*}
\end{minipage}}\vspace{0.3cm}\\
\item It then follows by \cite[Theorem 4.2]{mainarticle2} that 
\begin{eqnarray*}
 \tilde{u}_{P_k,s,t}= \hat{u}_{P_k,X,t,s}+\hat{u}_{P_k,Y,s,t}+O_p\left( \frac{1}{m} \right)+O_p\left( \frac{1}{ \theta m^{1/2}} \right).
\end{eqnarray*}
\end{itemize}
\noindent First, we set $b=+$ or $b=-$ and separate the scalar product in three parts. 
\begin{eqnarray*}
\left\langle w_s^{b},w_t^{b} \right\rangle&=& \underbrace{\sum_{i=s,t} w_{s,i}^{b}w_{t,i}^{b}}_{3)}+\underbrace{\sum_{i\neq s,t}^{k} w_{s,i}^{b}w_{t,i}^{b}}_{1)}+ \underbrace{\sum_{i=k+1}^m w_{s,i}^{b}w_{t,i}^{b}}_{2)}.
\end{eqnarray*}

\begin{itemize}
\item[1) ]
If $k=2$, the second term does not exist. However, if $k>2$, then asymptotically for $i=1,...,k$, $i\neq s,t$,
\begin{eqnarray*}
w_{s,i}^{b}w_{t,i}^{b}&=& \frac{\sqrt{(\hat{\theta}_{P_k,Y,s}-1)(\hat{\theta}_{P_k,Y,t}-1)}}{\left(\zeta_{\infty}^{b}-2\left(M_2-1 \right)\right)^2+2 \left(M_2 -1\right)+\epsilon} \frac{\tilde{u}_{P_k,s,i}}{\sqrt{\hat{\theta}_{P_k,X,t}}}\frac{\tilde{u}_{P_k,t,i}}{\sqrt{\hat{\theta}_{P_k,X,s}}}\\
&=&O_p\left( \frac{1}{m} \right).
\end{eqnarray*}
\item[2) ] By \cite[Theorem 4.2]{mainarticle2},\\
\scalebox{0.75}{
\begin{minipage}{1\textwidth}
\begin{eqnarray*}
\sum_{i=k+1}^m \tilde{u}_{j,i} \tilde{u}_{t,i}&=& \sum_{i=k+1}^m \hat{u}_{\hat{\Sigma}_{Y},j,i} \hat{u}_{\hat{\Sigma}_{Y},t,i} +\sum_{i=k+1}^m \hat{u}_{\hat{\Sigma}_{X},j,i} \hat{u}_{\hat{\Sigma}_{X},t,i} -\sum_{i=k+1}^m \hat{u}_{\hat{\Sigma}_{X},j,i} \hat{u}_{\hat{\Sigma}_{Y},t,i} \\
&& \hspace{2cm} - \sum_{i=k+1}^m \hat{u}_{\hat{\Sigma}_{Y},j,i} \hat{u}_{\hat{\Sigma}_{X},t,i}- \left( \hat{u}_{\hat{\Sigma}_{X},t,j}+\hat{u}_{\hat{\Sigma}_{Y},j,t} \right) \left( \hat{\alpha}^2_{\hat{\Sigma}_{X},j} - \hat{\alpha}^2_{\hat{\Sigma}_{X},t} \right)\\
&& \hspace{2cm}+ O_p\left( \frac{1}{\theta^{1/2} m} \right)+O_p\left( \frac{1}{\theta  m^{1/2}} \right).
\end{eqnarray*}
\end{minipage}}\vspace{0.3cm}\\
Therefore,\\
\scalebox{0.83}{
\begin{minipage}{1\textwidth}
\begin{eqnarray*}
\sum_{i=k+1}^m w_{s,i}^{b}w_{t,i}^{b} &=& \frac{\sqrt{\hat{\theta}_{P_k,Y,s}-1} \sqrt{\hat{\theta}_{P_k,Y,t}-1} }{\sqrt{\left(\zeta_{\infty}^{b}-2 \left(M_2 -1\right)\right)^2+2 \left(M_2 -1\right)}+\epsilon} \sum_{p=k+1}^m \tilde{u}_{P_k,s,p}\tilde{u}_{P_k,t,p}.
\end{eqnarray*}
\end{minipage}}\vspace{0.3cm}\\
\item[3) ] Asymptotically, \\
\scalebox{0.68}{
\begin{minipage}{1\textwidth}
\begin{eqnarray*}
w_{s,s}^{b}w_{t,s}^{b}+w_{s,t}^{b}w_{t,t}^{b}&=&\frac{\zeta_{\infty}^{b}-2 \left(M_2 -1\right)+\epsilon}{\left(\zeta_{\infty}^{b}-2 \left(M_2 -1\right)\right)^2+2 \left(M_2 -1\right) } \left( w_{t,s} + w_{s,t} \right)\\
&=&\frac{\zeta_{\infty}^{b}-2 \left(M_2 -1\right)}{\left(\zeta_{\infty}^{b}-2 \left(M_2 -1\right)\right)^2+2 \left(M_2 -1\right) } 
 \left( \frac{\sqrt{\theta_t}}{\sqrt{\theta_s}}\tilde{u}_{P_k,t,s} +\frac{\sqrt{\theta_s}}{\sqrt{\theta_t}}\tilde{u}_{P_k,s,t} \right)+O_p\left( \frac{1}{ m^{1/2}} \right) \epsilon\\
&=&\frac{\left(\zeta_{\infty}^{b}-2 \left(M_2 -1\right)\right) \sqrt{\theta_s \theta_t}}{\left(\zeta_{\infty}^{b}-2 \left(M_2 -1\right)\right)^2+2 \left(M_2 -1\right) } \\
&&\hspace{0.5cm} \left( \frac{1}{\theta_s}\left(\hat{u}_{P_k,X,s,t}+\hat{u}_{P_k,Y,t,s}\right) +\frac{1}{\theta_t}\left( \hat{u}_{P_k,X,t,s}+\hat{u}_{P_k,Y,s,t} \right) \right) +O_p\left( \frac{1}{m^{1/2}} \right)\epsilon\\
&=&\frac{\left(\zeta_{\infty}^{b}-2 \left(M_2 -1\right) \right) \sqrt{\theta_s \theta_t}}{\left(\zeta_{\infty}^{b}-2 \left(M_2 -1\right)\right)^2+2 \left(M_2 -1\right) } \\
&&\hspace{0.5cm} \left( \frac{1}{\theta_t}  -\frac{1}{\theta_s} \right) \left( \hat{u}_{P_k,X,t,s}+\hat{u}_{P_k,Y,s,t} \right)+o_{p;\theta,m}\left( \frac{1}{m^{1/2}} \right).
\end{eqnarray*}\end{minipage}}\vspace{0.3cm}\\
\end{itemize}

Therefore,  we obtain \\
\scalebox{0.8}{
\begin{minipage}{1\textwidth}
\begin{eqnarray*}
\left\langle w_s^{b},w_t^{b} \right\rangle  
 &=& \frac{\sqrt{\theta_s \theta_t} }{\left(\zeta_{\infty}^b-2 M_2+1 \right)^2+2 \left( M_2 -1\right)} \times \\
 &&\hspace{0.5cm}\Bigg( \sum_{p=k+1}^m \hat{u}_{P_k,Y,s,p}\hat{u}_{P_k,Y,t,p}+\sum_{p=k+1}^m\hat{u}_{P_k,X,s,p}\hat{u}_{P_k,X,t,p} \\
&&\hspace{1.5cm}- 
\sum_{p=k+1}^m \hat{u}_{P_k,Y,s,p}\hat{u}_{P_k,X,t,p}-
\sum_{p=k+1}^m  \hat{u}_{P_k,Y,t,p}\hat{u}_{P_k,X,s,p}  \\ 
&&\hspace{1.5cm}- \left( \hat{u}_{P_k,X,t,s}+\hat{u}_{P_k,Y,s,t}\right) \left(\tilde{\alpha}^2_{t}-\tilde{\alpha}^2_{s}-
\left(\zeta_{\infty}^b-2 \left(M_2 -1\right)\right)
 \left( \frac{1}{\theta_t}  -\frac{1}{\theta_s} \right)  \right) \Bigg)\\
 &&\hspace{0.5cm}+ o_{p;\theta,m}\left( \frac{1}{m^{1/2}} \right). 
\end{eqnarray*}\end{minipage}}\vspace{0.3cm}\\

\paragraph*{Moment}
We separate the formula into three parts 
\begin{enumerate}
\item \scalebox{0.9}{ $ \sum_{p=k+1}^m \hat{u}_{P_k,Y,s,p}\hat{u}_{P_k,Y,t,p} - \hat{u}_{P_k,Y,s,t} \left(\tilde{\alpha}^2_{t}-\tilde{\alpha}^2_{s}-
\left(\zeta_{\infty}^b-2 \left(M_2 -1\right)\right)
 \left( \frac{1}{\theta_t}  -\frac{1}{\theta_s} \right)  \right)$} 
\item \scalebox{0.9}{$ \sum_{p=k+1}^m \hat{u}_{P_k,X,s,p}\hat{u}_{P_k,X,t,p}  + \hat{u}_{P_k,X,s,t} \left(\tilde{\alpha}^2_{t}-\tilde{\alpha}^2_{s}-
\left(\zeta_{\infty}^b-2 \left(M_2 -1\right)\right)
 \left( \frac{1}{\theta_t}  -\frac{1}{\theta_s} \right)  \right) $}
\item \scalebox{0.9}{$ \sum_{p=k+1}^m \hat{u}_{P_k,Y,s,p}\hat{u}_{P_k,X,t,p} + \sum_{p=k+1}^m \hat{u}_{P_k,X,s,p}\hat{u}_{P_k,Y,t,p} $}
\end{enumerate}
Without loss of generality, we present the proof for $s=1$ and $t=2$.\\
In order to compute the moments of the first and second parts, we use the remark of \cite[Theorem 3.3]{mainarticle2},\\
\scalebox{0.87}{
\begin{minipage}{1\textwidth}
\begin{eqnarray*}
&&\hspace{-0.5cm} \hat{u}_{P_2,1,2}  \left(  \frac{ 1}{\theta_1} - \frac{1}{\theta_2}\right)\delta +\sum_{i=3}^m  \hat{u}_{P_2,1,i}  \hat{u}_{P_2,2,i}\\
&&\sim  \Normal \left(0, \frac{\left(1+M_2+\delta \right)^2(M_2-1)+\left(M_4-(M_2)^2\right)-2 \left(1+M_2+\delta \right)\left(M_3-M_2\right)}{\theta_1 \theta_2 m} \right)\\
&&\hspace{1cm} +O_p\left(\frac{1}{\theta m} \right)+O_p\left(\frac{1}{\theta^2 m^{1/2}} \right).
\end{eqnarray*} \end{minipage}}\vspace{0.3cm}\\
\begin{enumerate}
\item \ \\
\scalebox{0.7}{
\begin{minipage}{1\textwidth}
\begin{eqnarray*}
&&\hspace{-0.5cm}\sum_{p=k+1}^m \hat{u}_{P_k,Y,1,p}\hat{u}_{P_k,Y,2,p} - \hat{u}_{P_k,Y,1,2} \left(\tilde{\alpha}^2_{2}-\tilde{\alpha}^2_{1}-
\left(\zeta_{\infty}^b-2 \left(M_2 -1\right)\right)
 \left( \frac{1}{\theta_2}  -\frac{1}{\theta_1} \right)  \right)\\
 && \hspace{0cm} = \sum_{p=k+1}^m \hat{u}_{P_k,Y,1,p}\hat{u}_{P_k,Y,2,p} + \hat{u}_{P_k,Y,1,2} \left(-M_{2,X}-
\zeta_{\infty}^b+2 M_2-1 \right)
 \left( \frac{1}{\theta_1}  -\frac{1}{\theta_2} \right) +O_p\left(\frac{1}{\theta^2 \sqrt{m}}\right) .
\end{eqnarray*} \end{minipage}}\vspace{0.3cm}\\
Using the remark, we set $\delta=-M_{2,X}+2 M_2-1-
\lambda^{b} $, \\
\scalebox{0.5}{
\begin{minipage}{1\textwidth}
 \begin{eqnarray*}
&&\hspace{-1cm}\sum_{p=k+1}^m \hat{u}_{P_k,Y,1,p}\hat{u}_{P_k,Y,2,p} - \hat{u}_{P_k,Y,1,2} \left(\tilde{\alpha}^2_{2}-\tilde{\alpha}^2_{1}-
\left(\zeta_{\infty}^b-2 \left(M_2 -1\right)\right)
 \left( \frac{1}{\theta_2}  -\frac{1}{\theta_1} \right)  \right) +O_p\left(\frac{1}{\theta m} \right)+O_p\left(\frac{1}{\theta^2 m^{1/2}} \right) \\
&& \hspace{-0.5cm} = \Normal \left(0, \frac{\left(1+M_2+M_{2,Y}-M_{2,X}\mp \sqrt{M_2^2-1} \right)^2(M_2-1)+\left(M_4-(M_2)^2\right)-2 \left(1+M_2+M_{2,Y}-M_{2,X}\mp \sqrt{M_2^2-1}  \right)\left(M_3-M_2\right)}{\theta_1 \theta_2 m} \right)
\end{eqnarray*} \end{minipage}}\vspace{0.3cm}\\
\item A similar computation leads to \\
\scalebox{0.56}{
\begin{minipage}{1\textwidth}
\begin{eqnarray*}
&&\hspace{-1cm}\sum_{p=k+1}^m \hat{u}_{P_k,X,1,p}\hat{u}_{P_k,X,2,p}  + \hat{u}_{P_k,X,1,2} \left(\tilde{\alpha}^2_{2}-\tilde{\alpha}^2_{1}-
\left(\zeta_{\infty}^b -2 \left(M_2 -1\right)\right)
 \left( \frac{1}{\theta_2}  -\frac{1}{\theta_1} \right)  \right) +O_p\left(\frac{1}{\theta m} \right)+O_p\left(\frac{1}{\theta^2 m^{1/2}} \right) \\
&& \hspace{-0.5cm} = \Normal \left(0, \frac{\left(1-M_2+2 M_{2,X} \pm \sqrt{M_2^2-1} \right)^2(M_2-1)+\left(M_4-(M_2)^2\right)-2 \left(1-M_2+2 M_{2,X} \pm \sqrt{M_2^2-1}  \right)\left(M_3-M_2\right)}{\theta_1 \theta_2 m} \right).
\end{eqnarray*}\end{minipage}}\vspace{0.3cm}\\
\item We then can easily show that \\
\scalebox{0.76}{
\begin{minipage}{1\textwidth}\begin{eqnarray*}
\sum_{p=k+1}^m \hat{u}_{P_k,Y,1,p}\hat{u}_{P_k,X,2,p} + \sum_{p=k+1}^m \hat{u}_{P_k,X,1,p}\hat{u}_{P_k,Y,2,p} &=& {\rm RV}\left(0, \frac{2 \left( M_{2,X}-1 \right) \left( M_{2,Y}-1 \right)}{\theta_1 \theta_1 m} \right)\\
&&\hspace{1cm}+O_p\left(\frac{1}{\theta m} \right)+O_p\left(\frac{1}{\theta^2 m^{1/2}} \right)
\end{eqnarray*}\end{minipage}}\vspace{0.3cm}\\

Indeed, the covariance between the first and the second term is negligible. This can be shown using the independence between $X$ and $Y$, and \cite[Theorem 3.3]{mainarticle2}.\\
\scalebox{0.87}{
\begin{minipage}{1\textwidth}
\begin{eqnarray*}
&&\cov\left(\sum_{p=k+1}^m \hat{u}_{P_k,Y,1,p}\hat{u}_{P_k,X,2,p}, \sum_{p=k+1}^m \hat{u}_{P_k,X,1,p}\hat{u}_{P_k,Y,2,p} \right) \\
&&\hspace{1cm} = \sum_{p=k+1}^m \E\left[\hat{u}_{P_k,Y,1,p}\hat{u}_{P_k,Y,2,p} \right] \sum_{p=k+1}^m \E\left[\hat{u}_{P_k,X,1,p}\hat{u}_{P_k,X,2,p} \right]
=o_p\left( \frac{1}{\theta^2 m}\right).
\end{eqnarray*}\end{minipage}}\vspace{0.3cm}\\
\end{enumerate}
Because the null covariance between the three parts is easily proven by the independence between $X$ and $Y$, we conclude:\\
\scalebox{0.78}{
\begin{minipage}{1\textwidth}
\begin{eqnarray*}
&&\left\langle w_{i}^b,w_{j}^b \right\rangle \sim {\rm RV}\left(0,\frac{1}{m}
\frac{2 (M_{2,X}-1) (M_{2,Y}-1)+B_{X}^b +B_{Y}^b}{\left((\zeta_{\infty}^\pm-2 M_2+1)^2+2 (M_2-1)\right)^2}
  \right)+O_p\left(\frac{1}{\theta m} \right)+O_p\left(\frac{1}{\theta^2 m^{1/2}} \right),\\
&&B_{X}^+=\left(1-M_2+2 M_{2,X}+\sqrt{M_2^2-1}\right)^2 (M_{2,X}-1)\\
&&\hspace{2.5cm} +2\left(-1+M_2-2M_{2,x}-\sqrt{M_2^2-1}\right)(M_{3,X}-M_{2,X})+(M_{4,X}-M_{2,X}^2),\\
&&B_{Y}^+=\left(1+M_2+M_{2,Y}-M_{2,X}-\sqrt{M_2^2-1}\right)^2 (M_{2,Y}-1)\\
&&\hspace{2.5cm} +2\left(-1-M_2-M_{2,Y}-M_{2,X}-\sqrt{M_2^2-1}\right)(M_{3,Y}-M_{2,Y})+(M_{4,Y}-M_{2,Y}^2),
\end{eqnarray*}\end{minipage}}\vspace{0.3cm}\\
\scalebox{0.78}{
\begin{minipage}{1\textwidth}
\begin{eqnarray*}
&&B_{X}^-=\left(1-M_2+2 M_{2,X}-\sqrt{M_2^2-1}\right)^2 (M_{2,X}-1)\\
&&\hspace{2.5cm} +2\left(-1+M_2-2M_{2,x}+\sqrt{M_2^2-1}\right)(M_{3,X}-M_{2,X})+(M_{4,X}-M_{2,X}^2),\\
&&B_{Y}^-=\left(1+M_2+M_{2,Y}-M_{2,X}+\sqrt{M_2^2-1}\right)^2 (M_{2,Y}-1)\\
&&\hspace{2.5cm} +2\left(-1-M_2-M_{2,Y}+M_{2,X}-\sqrt{M_2^2-1}\right)(M_{3,Y}-M_{2,Y})+(M_{4,Y}-M_{2,Y}^2),
\end{eqnarray*}\end{minipage}}\vspace{0.3cm}\\
where $b=+$ or $b=-$.

\subsubsection{Normality discussion}\label{Asec:HNormal}
Assuming $n_X>>n_Y$, the normality is straightforward to prove using \cite[Theorem 3.1]{mainarticle} and \cite[Theorem 4.2]{mainarticle2}. Nevertheless, when $n_X\sim n_Y$, new marginally normal statistics enter in the formula. These statistics are 
$$\sum_{p=k+1}^m \hat{u}_{P_k,Y,s,p}\hat{u}_{P_k,X,t,p} + \sum_{p=k+1}^m \hat{u}_{P_k,X,s,p}\hat{u}_{P_k,Y,t,p}.$$
Despite this difficulty, the reader can check that assuming large $\theta$, asymptotic joint normality of the entries of $H$ is equivalent to asymptotic joint normality of 
\begin{eqnarray*}
W_{X,s,t},W_{Y,s,t}, \left({W_X}^2\right)_{s,t},\left({W_Y}^2\right)_{s,t} \text{ and } \frac{1}{\sqrt{m}}\sum_{i=k+1}^m W_{X,s,i}W_{Y,t,i}
\end{eqnarray*}
for $s,t=1,2,...,k$. Note that joint normality holds for the first four elements by \cite[Theorem 3.1]{mainarticle}.\\
The part left to the reader is nearly done. \cite[Theorem 3.3 and Lemma 4.1]{mainarticle2} already showed that nearly all the statistics that composed $H$ are functions of the first four elements. A similar proof of \cite[Corollary 3.1]{mainarticle2} shows that for $s,t=1,2,...,k$,
$$\sum_{p=k+1}^m \hat{u}_{P_k,Y,s,p}\hat{u}_{P_k,X,t,p}$$
can be expressed as a function of the statistics.

\pagebreak
\bibliographystyle{imsart-nameyear}
\bibliography{biblio4} 
\addcontentsline{toc}{section}{Bibliography}

\end{document}